\documentclass[a4paper]{article}

\usepackage{amsmath}
\usepackage{amssymb}
\usepackage{amsthm}
\usepackage{textcomp}
\usepackage{nccmath}	%
\usepackage{mathrsfs}
\usepackage{wasysym}
\usepackage{bbm}		%
\usepackage[inline]{enumitem}
\usepackage{mathtools}
\usepackage{braket}		%
\usepackage[normalem]{ulem}	%
\usepackage{nicefrac}		%
\usepackage{cancel}			%
\usepackage[nottoc]{tocbibind}	%
\usepackage{xcolor}
\usepackage[a4paper, left=2.7cm, right=2.7cm, top=3.5cm, bottom=2cm]{geometry}
\usepackage[pdftex=true,hyperindex=true,pdfborder={0 0 0}]{hyperref}
\usepackage{doi}
\usepackage[numbers]{natbib}
\usepackage{soul}			%
\usepackage{tikz}
\usetikzlibrary{shapes,arrows,automata,matrix,fit,patterns,decorations.markings,calc,decorations.pathmorphing,decorations.pathreplacing}

\theoremstyle{plain}
\newtheorem{theorem}{Theorem}[section]
\newtheorem{lemma}[theorem]{Lemma}
\newtheorem{corollary}[theorem]{Corollary}
\newtheorem{proposition}[theorem]{Proposition}

\theoremstyle{definition}
\newtheorem{remark}[theorem]{Remark}
\newtheorem{example}[theorem]{Example}

\newcommand{\exampleqed}{\ensuremath{\ocircle}\par}
\newcommand{\remarkqed}{\ensuremath{\Diamond}\par}

\newcommand{\ZZ}{\mathbb{Z}}			%
\newcommand{\NN}{\mathbb{N}}			%
\newcommand{\RR}{\mathbb{R}}			%
\newcommand{\GG}{\mathbb{G}}			%
\newcommand{\HH}{\mathbb{H}}			%
\newcommand{\XX}{\mathbb{X}}			%
\newcommand{\EE}{\mathbb{E}}			%

\newcommand{\symb}[1]{\mathtt{#1}}		%
\newcommand{\isdef}{\triangleq}			%
\newcommand{\abs}[1]{%
	\left\lvert#1\right\rvert%
}
\newcommand{\bigabs}[1]{%
	\big\lvert#1\big\rvert%
}
\newcommand{\Bigabs}[1]{%
	\Big\lvert#1\Big\rvert%
}
\newcommand{\biggabs}[1]{%
	\bigg\lvert#1\bigg\rvert%
}
\newcommand{\norm}[1]{%
	\left\lVert#1\right\rVert%
}

\newcommand{\Biggnorm}[1]{%
	\Bigg\lVert#1\Bigg\rVert%
}
\renewcommand{\complement}{%
	\mathsf{c}%
}

\newcommand{\density}{%
	\underline{d}%
}
\newcommand{\dd}{\mathrm{d}}			%
\newcommand{\ee}{\mathrm{e}}			%
\newcommand{\smallo}{o}					%
\newcommand{\xPr}{\operatorname{\mathbb{P}}}		%
\newcommand{\indicator}[1]{\mathbbm{1}_{#1}}			%
\newcommand{\xspace}[1]{\mathcal{#1}}	%
\newcommand{\field}[1]{\mathscr{#1}}	%
\newcommand{\family}[1]{\mathcal{#1}}	%

\newcommand*\circled[1]{%
	\mathchoice
		{\tikz[baseline=-3pt]{%
			\node[shape=ellipse,draw,inner sep=0.5pt] (char) {\normalsize\makebox[1ex]{$#1$}};%
		}}%
		{\tikz[baseline=-3pt]{%
			\node[shape=ellipse,draw,inner sep=0.5pt] (char) {\normalsize\makebox[1ex]{$#1$}};%
		}}%
		{\tikz[baseline=-3pt]{%
			\node[shape=ellipse,draw,inner sep=0.5pt] (char) {\tiny\makebox[1ex]{$#1$}};%
		}}%
		{\tikz[baseline=-3pt]{%
			\node[shape=ellipse,draw,inner sep=0.5pt] (char) {\tiny\makebox[1ex]{$#1$}};%
		}}%
}
\newcommand*\hexed[1]{%
	\mathchoice
		{\tikz[baseline=-3pt]{%
			\node[regular polygon,regular polygon sides=6,draw,inner sep=0.5pt] (char) {\normalsize\makebox[1ex]{$#1$}};%
		}}%
		{\tikz[baseline=-3pt]{%
			\node[regular polygon,regular polygon sides=6,draw,inner sep=0.5pt] (char) {\normalsize\makebox[1ex]{$#1$}};%
		}}%
		{\tikz[baseline=-3pt]{%
			\node[regular polygon,regular polygon sides=6,draw,inner sep=0.7pt] (char) {\tiny\makebox[1ex]{$#1$}};%
		}}%
		{\tikz[baseline=-3pt]{%
			\node[regular polygon,regular polygon sides=6,draw,inner sep=0.5pt] (char) {\tiny\makebox[1ex]{$#1$}};%
		}}%
}

\title{%
	Equivalence of relative Gibbs and relative equilibrium measures for actions of countable amenable groups
	\footnotetext{Last update:~\today}
}

\author{%
	Sebasti\'an Barbieri%
	\and
	Ricardo G\'omez%
	\and
	Brian Marcus
	\and
	Siamak Taati
}

\date{}

\mathtoolsset{showonlyrefs}

\begin{document}

\maketitle

\begin{abstract}

We formulate and prove a very general \emph{relative} version of the Dobrushin--Lanford--Ruelle theorem
which gives conditions on constraints of configuration spaces over a finite alphabet such that for every
absolutely summable relative interaction, every translation-invariant relative Gibbs measure is a relative
equilibrium measure and vice versa. Neither implication is true without some assumption on the space
of configurations. We note that the usual finite type condition can be relaxed to a much more general
class of constraints. By ``relative'' we mean that both the interaction and the set of allowed configurations
are determined by a random environment. The result includes many special cases that are well known.
We give several applications including
\begin{enumerate*}[label={\arabic*)}]
	\item Gibbsian properties of measures that maximize pressure among all those
		that project to a given measure via	a topological factor map from one symbolic system to another;
	\item Gibbsian properties of equilibrium measures for group shifts defined on arbitrary countable amenable groups;
	\item A Gibbsian characterization of equilibrium measures in terms of equilibrium condition on lattice slices rather than on finite sets;
    \item A relative extension of a theorem of Meyerovitch, who proved a version of the Lanford--Ruelle theorem which shows that every equilibrium measure on an arbitrary subshift satisfies a Gibbsian
		property on interchangeable patterns.
\end{enumerate*}

	\medskip

	
	\smallskip
	
	
	\renewcommand{\contentsname}{\vspace{-1.5em}}
	{\footnotesize\tableofcontents}
\end{abstract}

\section{Introduction}
\label{sec:intro}

The starting point of Gibbs's approach to equilibrium statistical physics is the postulate
that the macroscopic state
of a system at thermal equilibrium is appropriately described by a probability distribution
that minimizes the free energy. An equivalent formulation is obtained by maximizing the \emph{pressure},
that is, the difference between the entropy and a constant times the expected energy.
In a lattice model in which the microscopic states are configurations of symbols on
an infinite lattice (e.g., the Ising model), there are two interpretations of this hypothesis:
\begin{enumerate}[label={\rm (\roman*)}]
	\item \emph{Local} maximization: the conditional pressure for every finite region of
		the lattice is maximized, so that every finite region is in equilibrium with its surrounding.
		This leads to the concept of \emph{Gibbs measures}.
	\item \emph{Global} maximization: the average pressure per site
		(i.e., Kolmogorov--Sinai entropy minus expected energy per site) is maximized.
		The maximizing measures in this interpretation are referred to as the \emph{equilibrium measures}.
\end{enumerate}
The celebrated theorem of Dobrushin~\cite{Dob68}, Lanford and Ruelle~\cite{LanRue69} says that under broad
conditions, equilibrium measures and shift-invariant Gibbs measures coincide (see~\cite{Rue04}).

\begin{theorem}%
\label{thm:DLR}
	Let $\Sigma$ be a finite set of symbols.
	Let $X\subseteq\Sigma^{\ZZ^d}$ be a $d$-dimensional subshift,
	$\Phi$ an absolutely summable interaction on~$X$,
	and $f_\Phi$ an associated energy observable.
	\begin{enumerate}[label={\rm (\alph*)}] %
		\item \textup{(Dobrushin theorem)} \label{thm:DLR:D} \\
			Assume that $X$ is D-mixing.
			Then, every shift-invariant Gibbs measure for $\Phi$ is an equilibrium measure for $f_\Phi$.
		\item \textup{(Lanford--Ruelle theorem)} \label{thm:DLR:LR} \\
			Assume that $X$ is a subshift of finite type (SFT).
			Then, every equilibrium measure for $f_\Phi$ is a Gibbs measure for~$\Phi$.
	\end{enumerate}
\end{theorem}

\noindent
Here, $X$ is the space of allowed configurations on the $d$-dimensional lattice.
Neither direction is true in general and so some kind of restrictions,
such as D-mixing in part~\ref{thm:DLR:D} and SFT in part~\ref{thm:DLR:LR}, on $X$, are required.
Terminology used in the statement of Theorem~\ref{thm:DLR}, as well
as other terminology used in this section, will be given in Section~\S\ref{sec:prelim}.

We generalize this theorem in several directions.
First, we allow the lattice to be any countable amenable group.
Second, we allow the presence of a random environment that imposes constraints
on the allowed configurations and affects the energy,
and prove the equivalence of local and global maximization \emph{relative to}
this environment.  Third, we relax the ``finite type'' hypothesis in the Lanford--Ruelle direction
to the much weaker
topological Markov property, and discuss the relationship between this and related properties.
We also give several applications.

To be specific, let $\GG$ be a countable amenable group (e.g., $\GG=\ZZ^d$ with $d=1,2,\ldots$),
$\Sigma$ a finite alphabet,
and $\Theta$ a measurable space on which $\GG$ acts via measurable maps.
The group $\GG$ also acts on $\Sigma^\GG$ by translations. For each $\theta\in\Theta$, let $X_\theta\subseteq\Sigma^\GG$ be a non-empty closed set
such that $\Omega\isdef\{(\theta,x): \text{$\theta\in\Theta$ and $x\in X_\theta$}\}$ is measurable
and $X_{g\theta}=\{gx: x\in X_\theta\}$ for every $\theta\in\Theta$ and $g\in\GG$.
We think of $x\in\Sigma^\GG$ as a microscopic configuration of a physical system
and $\theta\in\Theta$ as the external environment.
The fact that $X_\theta$ is not required to be the entire $\Sigma^\GG$
indicates the possibility of ``hard'' (or combinatorial) constraints that the environment
can impose on the system.
We refer to $\Omega$ as a \emph{relative system}.

With suitably formulated relative versions of the hypotheses in
Theorem~\ref{thm:DLR}, our generalization is as follows.

\begin{theorem}%
\label{thm:DLR:relative} Let the environment $\Theta$ and relative system $\Omega$ be as formulated above.
	Let $\nu$ be a $\GG$-invariant probability measure on $\Theta$.
	Let $\Phi$ be an absolutely summable relative interaction on $\Omega$
	and $f_\Phi$ an associated energy observable.
	\begin{enumerate}[label={\rm (\alph*)},ref={(\alph*)}] %
		\item \textup{(Relative Dobrushin theorem)} \label{thm:DLR:D:relative} \\
			Assume that $\Omega$ is D-mixing relative to $\nu$.
			Then, every $\GG$-invariant relative Gibbs measure for $\Phi$ with marginal $\nu$
			is an equilibrium measure for $f_\Phi$ relative to $\nu$.
		\item \textup{(Relative Lanford--Ruelle theorem)} \label{thm:DLR:LR:relative} \\
			Assume that $\Theta$ is a standard Borel space.
			Assume further that $\Omega$ has the topological Markov property relative to $\nu$.
			Then, every equilibrium measure for $f_\Phi$ relative to $\nu$
			is a relative Gibbs measure for $\Phi$ with marginal $\nu$.
	\end{enumerate}
\end{theorem}
\noindent The concepts of relative Gibbs measure, relative equilibrium measure,
relative interaction, absolute summability, relative D-mixing
and relative topological Markov property
are natural analogues of the corresponding non-relative concepts in the relative setting.
The proof of Theorem~\ref{thm:DLR:relative} is given in Section~\S\ref{sec:proof}.
In the non-relative setting, that is, when $\Theta\isdef\{\theta\}$ and $\nu\isdef\delta_{\theta}$,
we recover a generalization of Theorem~\ref{thm:DLR}.

Many aspects of the above generalization are not new.
\begin{enumerate}[label=$\bullet$]
	\item Sepp\"al\"ainen~\cite[Section \S8]{Sep95} proved a relative version of
		the Dobrushin--Lanford--Ruelle theorem on the hyper-cubic lattice $\ZZ^d$.
		In his result, the alphabet is allowed to also be a complete separable metric space.
		On the other hand, the environment space in his setting
		is required to be a complete separable metric space, with $\ZZ^d$ acting by homeomorphisms,
		and the interaction is assumed to be continuous as a function of~$\theta$.
		Moreover, this setting does not allow hard constraints, i.e.,
		$X_\theta=\Sigma^{\ZZ^d}$ for every~$\theta$ (equivalently, $\Omega=\Theta\times\Sigma^{\ZZ^d}$).
		See also the paper by Zegarlinski~\cite{Zeg91}.
	\item Moulin~Ollagnier and Pinchon~\cite{MouPin81} (see Moulin~Ollagnier~\cite[Thm.~7.2.5]{Mou85})
		and Tempelman~\cite[Section \S8]{Tem84})
		have extended the (non-relative)
		Dobrushin--Lanford--Ruelle theorem to countable amenable groups, but again
		these results do not allow hard constraints. On the other hand,
		Templeman allows the alphabet to be an arbitrary $\sigma$-finite measure space.
	\item In the case where the acting group is $\ZZ$, very strong results are known even
		in the relative setting.
		In fact, in this case, if the variations of the energy observable decay rapidly enough,
		there is a unique equilibrium measure which coincides with a unique Gibbs measure,
		and this measure can be described
		explicitly as the unique fixed point of a Ruelle--Perron--Frobenius operator;
		see the survey by Kifer and Liu~\cite{KifLiu06} (Theorem~4.1.1 and the following paragraph)
		and Kifer~\cite{Kif08}.   In this setting, these systems are known as random dynamical systems.
	\item In the case where the acting group is $\ZZ^d$, the framework of a relative system,
		much as we have formulated it above, is given in  Kifer~\cite{Kif95}. In this work the assumptions on $\Omega$ are in some ways more general and
		in some ways less general than ours.
\end{enumerate}

For a given continuous observable $f$, an equilibrium measure achieves the supremum,
over all $\GG$-invariant measures $\mu$, of the difference of the entropy of $\mu$ and
the expected value of $f$ with respect to~$\mu$.
In the standard setting of a continuous $\ZZ^d$-action on a compact metric space,
this supremum is characterized as an intrinsically defined notion
of topological pressure for $f$.  %
Similar variational principles have been established in %
the contexts of the above-mentioned results (see e.g.~\cite{MouPin82,LedWal77,Kif01}).
In our paper, we do not consider such variational principles
(see however Prop.~\ref{prop:variational-principle} for a special case);
rather we focus on conditions which guarantee that every Gibbs measure is an equilibrium
measure and that every equilibrium measure is a Gibbs measure. Also, the
papers~\cite{Sep95, Kif95} include, and are motivated by, large deviations principles,
which is another topic that we do not consider.

Sepp\"al\"ainen~\cite{Sep95} gave several examples to which his result applied.
This includes the Ising model with random external field and the Edwards--Anderson
spin glass model in which the coupling parameters for neighboring spins are i.i.d.\ random.
In these models, there are no hard constraints on the configurations. Below we give
two examples in which there are hard constraints. In both of them we assume that
$\GG$ is finitely generated, and we consider a fixed finite symmetric generating set
$S$ with $S\not\ni 1_\GG$. We consider the Cayley graph of $\GG$ generated by
$S$ as a simple undirected graph with vertex set $\GG$ and edge set
$\EE\isdef\{\{a,b\}: a^{-1}b\in S\}$.

\begin{example}[Ising model on percolation clusters]\label{ex:ising_model} %
	Let $\Theta\isdef\{\symb{0},\symb{1}\}^\GG$, and let $\nu$ be a $\GG$-invariant
	measure on $\Theta$, for instance the Bernoulli measure with parameter $p\in(0,1)$.
	Let $\Sigma\isdef\{\symb{-1},\symb{0},\symb{+1}\}$, and
	for $\theta\in\Theta$, let $X_\theta$ be the set of configurations $x\in\Sigma^\GG$
	for which $x_k=\symb{0}$ if and only if $\theta_k=\symb{0}$.
	Let $h\in\RR$ and consider the relative interaction $\Phi$ defined by
	\useshortskip
	\begin{align}
		\Phi_{\{k\}}(\theta,x) &\isdef -h x_k \\
		\Phi_{\{i,j\}}(\theta,x) &\isdef -x_ix_j	\qquad\text{if $\{i,j\}\in\EE$,}
	\end{align}
	and $\Phi_A(\theta,x)\isdef 0$ whenever $A$ is neither a singleton
	nor an edge in $\GG$. So, the effect of the environment is simply to constrain the configurations.
	Observe that $\Phi(\theta,\cdot)$ is essentially the Ising interaction with
	external magnetic field $h$ on
	the subgraph induced by $\{k\in\GG: \theta_k=\symb{1}\}$.
	
	This system has been studied as a model of a binary alloy
	consisting of a ferromagnetic and a non-magnetic metal~\cite{GriLeb68,Led77,Geo81,HagSchSte00}.
	Each site is chosen at random to carry either a magnetic or a non-magnetic atom.
	The magnetic atoms interact with one another as in the Ising model,
	while the non-magnetic atoms do not interact.
	
	In Section~\S\ref{sec:hypoth} we will see that the relative system $\Omega$
	corresponding to this model satisfies the assumptions    %
	of both parts of Theorem~\ref{thm:DLR:relative} and thus
	the equilibrium measures for $f_\Phi$ relative to~$\nu$ coincide with the
	$\GG$-invariant relative Gibbs measures for $\Phi$ with marginal~$\nu$.
	
	Let us point out that this example can be rephrased so that
	there are no constraints on the configurations and thus fits in
	the setting of Sepp\"al\"ainen's result~\cite{Sep95}. 
	Let $\Theta\isdef\{\symb{0},\symb{1}\}^\EE$
	be the set of bond configurations rather than site configurations.
	Let $\nu$ be the measure induced on $\Theta$ by a Bernoulli measure on the sites
	by letting a bond be open if and only if both its endpoints are open.
	Then we allow an Ising spin ($\symb{\pm 1}$) at every site,
	except that the spins at closed sites will not interact with other
	spins, and therefore will be independent in both the relative
	equilibrium measures and the relative Gibbs measures.
	In this setting, the environment constrains the interaction
	but not the set of allowed configurations.
	\hfill\exampleqed
\end{example}

\begin{example}[Random colorings of random graphs]\label{ex:graphs}
	Let $\Theta\isdef 2^\EE$ denote the set of all subgraphs of $(\GG,\EE)$
	that have the same vertex set $\GG$.
	The group $\GG$ acts on a subgraph $\theta\in\Theta$
	by translation, that is, $g\theta\isdef\{\{ga,gb\}: \{a,b\}\in\theta\}$.
	Let $\Sigma$ be a finite set of colors.
	For a subgraph $\theta\in\Theta$,
	denote by $X_\theta$ the set of valid \emph{$\Sigma$-colorings} of~$\theta$,
	that is, the configurations $x\in\Sigma^\GG$ such that
	$x_a\neq x_b$ whenever $\{a,b\}\in\theta$. %
	Clearly $X_\theta$ is closed and we have $X_{g\theta}=g X_\theta$ for each $\theta$.
	Moreover, the set $\Omega\isdef\{(\theta,x): x\in X_\theta\}$ is measurable.
	
	A $\GG$-invariant measure $\nu$ on $\Theta$ may be viewed as
	the distribution of
	a stationary random subgraph $\pmb{\theta}$ of $(\GG,\EE)$.
	We assume that $\abs{\Sigma}>\abs{S}$ and so $X_{\pmb{\theta}}$ is almost surely non-empty.
	A \emph{max-entropic random coloring} of $\pmb{\theta}$ is a random configuration
	$\pmb{x}$ from $\Sigma^\GG$
	defined in the same probability space as $\pmb{\theta}$
	such that $\pmb{x}\in X_{\pmb{\theta}}$ almost surely
	and	the joint distribution $\mu$ of $(\pmb{\theta},\pmb{x})$
	has maximum possible relative entropy $h_\mu(\Omega\,|\,\Theta)$.
	A \emph{uniform-Gibbs coloring} is a random coloring of $\pmb{\theta}$
	such that for every finite set $A\subseteq\GG$, the conditional distribution
	$\xPr(\pmb{x}_A=\cdot\,|\,\pmb{\theta},\pmb{x}_{A^\complement})$
	is almost surely uniform among all patterns $u\in\Sigma^A$
	for which $u\lor\pmb{x}_{A^\complement}$ is a valid coloring of $\pmb{\theta}$.
	We shall see in Section~\S\ref{sec:hypoth} that the assumptions of the
	relative Dobrushin--Lanford--Ruelle theorem hold and so a stationary random coloring is
	max-entropic if and only if it is uniform-Gibbs. 	
	\hfill\exampleqed
\end{example}

\bigskip

Now we consider some applications.

\paragraph{Equilibrium measures relative to a topological factor.}

Following Ledrappier and Walters~\cite{LedWal77,Wal86}, a related notion of
an equilibrium measure relative to an invariant measure on a topological factor
has been studied, primarily in the context of one-dimensional symbolic dynamics.

Let $\eta:X\to Y$ be a topological factor map
from a one-dimensional SFT $X$ onto another subshift~$Y$.
Let $\nu$ be a fixed shift-invariant measure on $Y$.
Consider an invariant measure $\mu$ on $X$ that projects to $\nu$
and has maximal entropy within the fiber $\eta^{-1}(\nu)$.  %
In~\cite[Thm.~3.3]{AllQua13}, Allahbakhshi and Quas proved that $\mu$
has the following Gibbsian property:
for every finite set $A\Subset\ZZ$ and $\mu$-almost every $x\in X$,
the conditional distribution of the pattern on $A$ given $\eta(x)$ and $x_{A^\complement}$
is uniform among all patterns $u$ on $A$ that are consistent with $x_{A^\complement}$ and $\eta(x)$
(i.e., $u$ and $x_{A^\complement}$ form a configuration that is in $X$ and that maps to $\eta(x)$.)

As an immediate application of Theorem~\ref{thm:DLR:relative}\ref{thm:DLR:LR:relative},
the result of Allahbakhshi and Quas can be generalized in three directions.
First, the SFT condition on $X$ can be replaced by the more general topological Markov property.
Second, we can allow for actions of arbitrary countable amenable groups.
Third, we may include an absolutely summable interaction on $X$, and
obtain a similar Gibbsian property for measures that maximize pressure in the fiber.
The precise statement and further details are given in Section~\S\ref{sec:top_factor}.

\paragraph{Equilibrium measures on group shifts.}

Let $\GG$ be a countable group and $\HH$ a finite group.
The full shift $\HH^\GG$ is itself a group with respect to the pointwise operation
$(x\cdot y)_g \isdef x_g \cdot y_g$ (for $x,y\in \HH^{\GG}$ and $g \in \GG$).
A group shift is a closed
shift-invariant subset $\XX \subseteq \HH^{\GG}$ which is also a subgroup of~$\HH^{\GG}$.
Kitchens and Schmidt~\cite{KitSch88} showed that every group shift over $\GG\isdef\ZZ$
or $\GG\isdef\ZZ^2$ is an SFT.
More generally, any polycyclic-by-finite group has this property~\cite[Thms.~3.8 and~4.2]{Sch95}.
However, this does not hold in general.
For instance, if $\GG$ is a countable group that is not finitely generated and $\HH\isdef\ZZ/2\ZZ$,
then the subshift $\XX \isdef\{\symb{0}^{\GG}, \symb{1}^{\GG}\}$ is a group shift but not an SFT.
More generally, over any countable group $\GG$ that contains a non-finitely generated subgroup,
there are group shifts that are not SFTs~\cite{Sal18}.

Nevertheless, in Section~\S\ref{sec:groupshifts} we show that
every group shift over a countable group has the topological Markov property.
The extended version of the non-relative Lanford--Ruelle theorem
(Theorem~\ref{thm:DLR:relative}\ref{thm:DLR:LR:relative} with trivial environment)
thus gives the following result.

\begin{theorem}[Equilibrium on group shifts]
\label{thm:groupshifts}
	Let $\GG$ be a countable amenable group and $\HH$ a finite group,
	and let $\XX \subseteq \HH^{\GG}$ be a group shift.
	Let $\Phi$ be an absolutely summable interaction on $\XX$ with an associated observable $f_{\Phi}$.
	Then every equilibrium measure on $\XX$ for $f_{\Phi}$ is a Gibbs measure for $\Phi$.
\end{theorem}

Note that the special case of the above theorem with $\GG\isdef\ZZ^d$
follows from the classical Lanford--Ruelle theorem (Theorem~\ref{thm:DLR}\ref{thm:DLR:LR})
and the fact that every group shift on $\ZZ^d$ is an SFT.
The general case requires not only the extension to countable amenable groups
but also the relaxation of the SFT condition to the topological Markov property.

In Section~\S\ref{sec:groupshifts}, we use Theorem~\ref{thm:groupshifts}
to give a sufficient condition for the Haar measure to be the unique measure
of maximal entropy on a group shift.

\paragraph{Relative equilibrium measures on lattice slices.}
Our original motivation to develop a relative Dobrushin--Lanford--Ruelle theorem
was to characterize equilibrium measures on two-dimensional subshifts
in terms of equilibrium conditions on finite-height horizontal strips. 

More specifically, let $Y\subseteq\Sigma^{\ZZ^2}$ be a two-dimensional subshift.
Given a positive integer $N$, we can view $Y$ as a relative system $\Omega_N$
with respect to horizontal shift,
by thinking of each $y\in Y$ as a configuration $x\isdef y_{\ZZ\times[0,N-1]}$
on the horizontal strip $\ZZ\times[0,N-1]$
together with the configuration $\theta\isdef y_{\ZZ\times[0,N-1]^\complement}$
on the complement of the strip as the environment.
In analogy with the Lanford--Ruelle theorem, one may expect that
every equilibrium measure on~$Y$ (with respect to $\ZZ^2$-shift)
is a relative equilibrium measure on $\Omega_N$ (with respect to horizontal shift).
Conversely, analogy with the Dobrushin theorem suggests that
if a $\ZZ^2$-invariant measure $\mu$ is a relative equilibrium measure on $\Omega_N$
(with respect to horizontal shift)
for each positive $N$, then $\mu$ must be an equilibrium measure on $Y$
(with respect to $\ZZ^2$-shift).

We now state a version of this characterization.
Let $\Pi_N$ denote the projection $y\mapsto y_{\ZZ\times[0,N-1]^\complement}$
on the complement of the strip $\ZZ\times[0,N-1]$.
We assume that $Y$ satisfies topological strong spatial mixing (TSSM), defined in
Section~\S\ref{sec:hypoth}, which implies, in this setting, hypotheses of
both parts of the relative Dobrushin--Lanford--Ruelle theorem. 
Examples of subshifts with TSSM include the hard-core subshift
and the subshift of $5$-colorings on $\ZZ^2$ (see~\cite{Bri17}).

\begin{theorem}[Equilibrium vs.\ relative equilibrium on strips]
\label{thm:Z2_slices}
	Let $Y$ be a $\ZZ^2$-subshift that satisfies TSSM.
	Let $\Phi$ be an absolutely summable interaction on~$Y$
	and $\mu$ a $\ZZ^2$-invariant measure on~$Y$.
	Then $\mu$ is an equilibrium measure for $\Phi$ (with respect to $\ZZ^2$-shift)
	if and only if for each positive integer $N$,
	$\mu$ is an equilibrium measure for $\Phi$ relative to its projection $\Pi_N\mu$
	(with respect to horizontal shift).
\end{theorem}

This theorem can be seen as an in-between characterization,
being local in one direction and global in the other.  In Section~\S\ref{sec:slices},
we prove a more general statement for subshifts on countable amenable groups.
In that setting, finite-width horizontal strips are replaced
by finite-width slices, which are unions of finitely many cosets of
a fixed subgroup.
Interestingly, when the subgroup is the trivial subgroup,
we recover the Dobrushin--Lanford--Ruelle theorem.
In principle, Theorem~\ref{thm:Z2_slices} and its generalization may enable better understanding of
an equilibrium measure for a $\GG$-action by a relative equilibrium measure for an action of a subgroup.

\paragraph{Relative version of Meyerovitch's theorem.}
The Dobrushin--Lanford--Ruelle theorem is not valid on arbitrary subshifts,
and the conditions of D-mixing and topological Markov property seem to be the
appropriate hypotheses. Meyerovitch~\cite{Mey13} has generalized the
Lanford--Ruelle theorem by removing the assumption on the subshift
while weakening the conclusion. To state his theorem, we need to introduce
some terminology. Two finite patterns $u,v\in\Sigma^A$ are said to be
\emph{interchangeable} in a subshift $X\subseteq\Sigma^\GG$
if for every $w\in\Sigma^{\GG\setminus A}$, we have $u\lor w\in X$ if and only if $v\lor w\in X$.
For example, in the golden mean shift, the words $\symb{010}$ and $\symb{000}$
are interchangeable, and in the even shift, the words $\symb{001}$ and $\symb{100}$
are interchangeable. Given $B\subseteq\ZZ^d$, we denote by $\xi^B$ the
$\sigma$-algebra on $X$ consisting of the events that depend only on the
pattern seen on $B$. Meyerovitch's result can be restated as follows.

\begin{theorem}[Meyerovitch's theorem]
\label{thm:Tom}
	Let $X\subseteq\Sigma^{\ZZ^d}$ be an arbitrary $d$-dimensional subshift.
	Let $\Phi$ be an absolutely summable interaction on $X$
	and $\mu$ an equilibrium measure for an associated observable~$f_\Phi$.
	Then, every two finite patterns $u,v\in\Sigma^A$ that are interchangeable in $X$
	satisfy
	\begin{align}
		\frac{\mu\big([u]\,\big|\,\xi^{A^\complement}\big)(x)}{\ee^{-E_{A|A^\complement}(u\lor x_{A^\complement})}} &=
			\frac{\mu\big([v]\,\big|\,\xi^{A^\complement}\big)(x)}{\ee^{-E_{A|A^\complement}(v\lor x_{A^\complement})}}
	\end{align}
	for $\mu$-almost every $x\in[u]\cup[v]$.
\end{theorem}

The conclusion of Meyerovitch's theorem becomes equivalent to the
Gibbs property when the subshift has the topological Markov property:
roughly speaking, the topological Markov property means that
every two patterns that share the same sufficiently thick margin are interchangeable.
It turns out that (the countable amenable group version of) Meyerovitch's result follows
from the relative Lanford--Ruelle theorem by means of an appropriate encoding.
In fact, we prove a relative version of Meyerovitch's result using this approach.
The notion of interchangeability extends naturally to relative systems:
we say that two finite patterns $u,v\in\Sigma^A$ are
\emph{interchangeable in }$X_\theta$
if for every $w\in\Sigma^{\GG\setminus A}$,
we have $u\lor w\in X_\theta$ if and only if $v\lor w\in X_\theta$.
The \emph{interchangeability set} of two finite patterns $u,v\in\Sigma^A$
is the set $\Theta_{u,v}$ of all environments $\theta\in\Theta$ for which
$u$ and $v$ are interchangeable.
Extending our earlier notation, we write $\xi^B$ for the $\sigma$-algebra on $\Omega$
generated by the projection $(\theta,x)\mapsto x_B$.
The $\sigma$-algebra on $\Omega$ generated by the environment
will be denoted by $\field{F}_\Theta$.
We prove the following theorem for an arbitrary relative system $\Omega$
on a countable amenable group $\GG$.

\begin{theorem}[Relative version of Meyerovitch's theorem]
\label{thm:Tom:relative}
	Assume that $\Theta$ is a standard Borel space.
	Let $\nu$ be a $\GG$-invariant probability measure on $\Theta$
	and $\Phi$ a relative absolutely summable interaction on~$\Omega$.
	Let $\mu$ be an equilibrium measure for $f_\Phi$ relative to $\nu$.
	Then, every two finite patterns $u,v\in\Sigma^A$ satisfy
	\begin{align}
		\label{eq:tom:rel:conclusion}
		\frac{\mu\big([u]\,\big|\,\xi^{A^\complement}\lor\field{F}_\Theta\big)(\theta,x)}{\ee^{-E_{A|A^\complement}(\theta,u\lor x_{A^\complement})}} &=
			\frac{\mu\big([v]\,\big|\,\xi^{A^\complement}\lor\field{F}_\Theta\big)(\theta,x)}{\ee^{-E_{A|A^\complement}(\theta,v\lor x_{A^\complement})}}
	\end{align}
	for $\mu$-almost every $(\theta,x)\in[u]\cup[v]$ such that $\theta\in\Theta_{u,v}$.
\end{theorem}

We prove Theorem~\ref{thm:Tom:relative} in Section~\S\ref{sec:relative_tom}.
As in the non-relative case,
the conclusion of Theorem~\ref{thm:Tom:relative} becomes equivalent to the relative Gibbs property
when the system has the topological Markov property relative to $\nu$,
and we recover the relative Lanford--Ruelle theorem.
This means that, up to relatively simple reductions,
Theorem~\ref{thm:DLR:relative}\ref{thm:DLR:LR:relative}
and Theorem~\ref{thm:Tom:relative} are equivalent!

\paragraph{Acknowledgements.}
We thank Raimundo Brice\~no for helpful discussions on the lattice slices application,
Nishant Chandgotia for useful remarks and many discussions concerning the relation
between interactions and observables,
and Fran\c{c}ois Ledrappier for raising a question that led to the application on lattice slices
and in turn motivated this paper.
We wish to express our gratitude to PIMS and the Department of Mathematics of the University of British columbia for their support. S. Barbieri was supported by the ANR project CoCoGro (ANR-16-CE40-0005). R. G\'omez was supported by the PAPIIT-UNAM project IN107718 and by PASPA-UNAM.

\section{Preliminaries}
\label{sec:prelim}

\subsection{Setting}
Throughout this article, we work in a general setting in which
the underlying lattice is a countable amenable group.
A countable group $(\GG,\cdot)$ is \emph{amenable} if there is a sequence of
non-empty finite subsets $F_n\subseteq\GG$ that are approximately (right) $\GG$-invariant,
in the sense that
for every $g\in\GG$, $\abs{F_n\triangle F_ng}=\smallo(\abs{F_n})$ as $n\to\infty$.
Such a sequence is called a (right) \emph{F\o{}lner sequence}.
For instance, the hyper-cubic lattice $(\ZZ^d,+)$ (with $d=1,2,\ldots$)
is amenable with the boxes $F_n\isdef[-n,n]^d\cap\ZZ^d$ forming a F\o{}lner sequence.
For many results in ergodic theory including the pointwise
ergodic theorem~\cite{Lin01}, the Shannon--McMillan--Breiman
theorem~\cite{OrnWei83,Lin01}, and the Ornstein isomorphism
theorem~\cite{OrnWei87}, the setting of amenable groups seems to be the right level of generality.
A reader who is not concerned with this level of generality is welcome to take $\GG\isdef\ZZ^d$.

Let $\GG$ be a countable amenable group and $\Sigma$ a finite alphabet.
We will use the notation $A \Subset \GG$ to indicate that $A$ is a finite subset of $\GG$.
We think of $x\in\Sigma^\GG$ as a microscopic configuration of a physical
system, with $x_g$ representing the local state of the system at spatial position~$g \in \GG$.
The group $\GG$ acts on $\Sigma^\GG$ by translations: the action of
an element $g\in\GG$ on a configuration $x\in\Sigma^\GG$ is the shifted configuration $gx$
where $(gx)_k\isdef x_{g^{-1}k}$ for every $k\in \GG$.
Given $A, B \subseteq \GG$, $u \in \Sigma^A$ and $v \in \Sigma^B$ such that $u_{A\cap B} = v_{A\cap B}$,
define $u \lor v \in \Sigma^{A \cup B}$ by $(u \lor v)_g = u_g$ for $g \in A$ and
$(u \lor v)_g = v_g$ for $g \in B$.

The system interacts with an external environment.
The space of all possible states of the environment
is a measurable space $\Theta$ on which $\GG$ acts via measurable maps.
For each $\theta\in\Theta$, let $X_\theta\subseteq\Sigma^\GG$ be a non-empty closed set,
representing the configurations that are consistent with environment $\theta$.
We impose two assumptions on the family $(X_\theta: \theta\in\Theta)$:
\begin{enumerate}[label={\rm (\roman*)}]
	\item (measurability) $\Omega\isdef\{(\theta,x): \text{$\theta\in\Theta$ and $x\in X_\theta$}\}$ is measurable
		in the product $\sigma$-algebra,
	\item (translation symmetry) $X_{g\theta}=g X_\theta$ %
		for each $\theta\in\Theta$ and $g\in\GG$.
\end{enumerate}
We call $\Omega$ a \emph{relative system}.
As an alternative interpretation,
if $\nu$ is a probability measure on $\Theta$, then $\theta\mapsto X_\theta$
is a random set in the probability space $(\Theta,\nu)$.

Recall that the $\sigma$-algebra on $\Omega$
generated by projection on $\Theta$ is denoted by $\field{F}_\Theta$.
We denote by $\xi$ the finite partition of $\Omega$ generated by
the projection $(\theta,x)\mapsto x_{1_\GG}$.
A \emph{cylinder} set is a set of the form
$[q]\isdef\{(\theta,x)\in\Omega: x_A=q\}$ where $q\in\Sigma^A$
is a pattern with (finite) support $A \Subset \GG$.
The set $A$ is called the \emph{base} of~$[q]$.
Given a subset $B\subseteq\GG$,
we write $\xi^B\isdef\bigvee_{k\in B}\xi^k$ (with $\xi^k\isdef k\xi$)
for the $\sigma$-algebra on $\Omega$ generated by cylinder sets
whose bases are included in $B$.

We call a measurable function $f\colon \Omega\to\RR$ an \emph{observable}.
An observable is said to be \emph{relatively local} if
it is $(\field{F}_\Theta\lor\xi^A)$-measurable for some $A\Subset\GG$.
An observable $f$ is \emph{relatively continuous} if
the family $(f(\theta,\cdot): \theta\in\Theta)$ is equicontinuous,
that is, for every $\varepsilon >0$, there exists a set $A \Subset \GG$
such that for every $\theta\in\Theta$ and $x,y \in X_{\theta}$ satisfying $x_A = y_A$,
we have $\abs{f(\theta,x)-f(\theta,y)}<\varepsilon$.
Every relatively local observable is clearly relatively continuous.
The set of bounded relatively continuous observables,
denoted by $C_\Theta(\Omega)$, is a Banach space with the uniform norm.
The bounded relatively local observables form a dense linear subspace of $C_\Theta(\Omega)$.

For a closed subset $Y\subseteq\Sigma^\GG$ and a finite set $A\Subset\GG$,
we write $L_A(Y)$ for the set of all patterns $q\in\Sigma^A$
such that $y_A=q$ for some $y\in Y$.
We define $L(Y)\isdef\bigcup\limits_{A\Subset \GG} L_A(Y)$.

\subsection{Relative interactions and Hamiltonians}
A (relative) \emph{interaction} on $\Omega$ is a collection
$\Phi=(\Phi_A: \text{$A\Subset\GG$})$ of
bounded measurable maps $\Phi_A\colon\Omega\to\RR$ such that
\begin{enumerate}[label={\rm (\roman*)},series=interaction]
	\item (relative locality) $\Phi_A$ is $(\field{F}_\Theta\lor\xi^A)$-measurable,
	\item (translation symmetry) $\Phi_{gA}(\theta,x)=\Phi_A(g^{-1}\theta,g^{-1} x)$.
\end{enumerate}
A relative interaction is \emph{absolutely summable} if
\begin{align}
	\norm{\Phi} \isdef
	\sum_{\substack{A\Subset\GG\\ A\ni 1_\GG}} \norm{\Phi_A} &< \infty \;,
\end{align}
where $\norm{\Phi_A}$ denotes the uniform norm of $\Phi_A$. %

Given an interaction $\Phi$, the \emph{energy content} of a configuration $x\in X$
in a finite set $A$ relative to an environment $\theta\in\Theta$ is
\begin{align}
	E_A(\theta,x) &\isdef \sum_{C\subseteq A} \Phi_C(\theta,x) \;.
\end{align}
The collection $E=(E_A: A\Subset\GG)$ is called
the (relative) \emph{Hamiltonian} defined by $\Phi$.
The \emph{conditional} energy content of $x$ inside $A\Subset\GG$
in the \emph{context} of $B\Subset\GG$ and environment $\theta\in\Theta$ is
\begin{align}
	E_{A|B}(\theta,x) &\isdef E_{A\cup B}(\theta,x) - E_B(\theta,x) \\
		&= \sum_{\substack{C\subseteq A\cup B\\ C\cap(A\setminus B)\neq\varnothing}}
			\Phi_C(\theta,x) \;.
\end{align}
Observe that $E_A$ is relatively continuous with $\norm{E_A}\leq\abs{A}\norm{\Phi}$,
and similarly, $\norm{E_{A|B}}\leq\abs{A\setminus B}\norm{\Phi}$.
The absolute summability of $\Phi$ ensures that the limit
\begin{align}
\label{eq:energy:conditional:limit}
	E_{A|A^\complement}(\theta,x) &\isdef
		\lim_{B\nearrow\GG} E_{A|(B\setminus A)}(\theta,x) \\
		&=	\sum_{\substack{C\Subset\GG\\ C\cap A\neq\varnothing}}
				\Phi_C(\theta,x)
\end{align}
exists along the finite subsets of $\GG$ directed by inclusion.
Moreover, the convergence is uniform in $(\theta,x)$,
hence $E_{A|A^\complement}(\theta,x)$ is %
bounded (namely, $\norm{E_{A|A^\complement}}\leq\abs{A}\norm{\Phi}$) and relatively continuous.

It is easy to see that
\useshortskip
\begin{align}
	\label{eq:Hamiltonian:residue}
	\norm{E_{A|A^\complement}-E_A} &\leq
		\sum_{\substack{C\Subset\GG \\
			C\cap A\neq\varnothing\\ C\cap A^\complement\neq\varnothing}}
			\norm{\Phi_C} \;.
\end{align}
Suppose that $(F_n)_{n\in\NN}$ is a F\o{}lner sequence in $\GG$.
It follows from the absolute summability of $\Phi$ that if we choose $A\isdef F_n$,
the right-hand side of~\eqref{eq:Hamiltonian:residue} becomes
of order $\smallo(\abs{F_n})$ as $n\to\infty$, hence
\begin{align}
	\label{eq:Hamiltonian:residue:Folner}
	\norm{E_{F_n|F_n^\complement}-E_{F_n}}
	&= \smallo(\abs{F_n})
\end{align}
as $n\to\infty$ (see~Sec.~\S\ref{sec:Hamiltonian:residue:Folner:argument}).
Another useful inequality is
\begin{align}
	\label{eq:Hamiltonian:different-sets:residue}
	\norm{E_{B|B^\complement}-E_{A|A^\complement}} &\leq
		\abs{B\setminus A}\norm{\Phi} \;,
\end{align}
which holds whenever $A$ and $B$ are finite and $A\subseteq B$
(see Sec.~\S\ref{sec:Hamiltonian:different-sets:residue:argument}).

The value $\Phi_A(\theta,x)$ is interpreted as the energy resulting
from the interaction between the symbols at sites in $A$ and the environment.
In models from physics, the interaction values are often physically meaningful values,
either being prescribed by the microscopic physics behind the model,
or representing rough microscopic tendencies for alignment or misalignment of
the physical quantities at different locations. The contribution of a
single site to the energy can be measured, for instance, by the following
bounded relatively continuous observable %
\begin{align}
\label{eq:interaction:associated-observable}
	f_\Phi(\theta,x) &\isdef
		\sum_{\substack{A\Subset\GG\\ A\ni 1_\GG}}
		\frac{1}{\abs{A}}\Phi_A(\theta,x) \;.
\end{align}
There are many other choices to distribute the energy contributions
between sites; see~\cite[Sec.~\S3.2]{Rue04} for some other choices.
The key relationship between $\Phi$ and $f_\Phi$ is that
for every F\o{}lner sequence $(F_n)_{n\in\NN}$,
\begin{align}
	\label{eq:interaction-observable:equivalence:1}
	\Bigabs{E_{F_n}(\theta,x)-\sum_{g\in F_n}f_\Phi(g^{-1}\theta,g^{-1}x)}
		&= \smallo(\abs{F_n})
\end{align}
as $n\to\infty$, uniformly in $(\theta,x)\in\Omega$
(see Sec.~\S\ref{sec:interaction-observable:equivalence}).
As a consequence, 
\begin{align}
	\label{eq:interaction-observable:equivalence:2}
	\mu(f_\Phi)
	&=
		\lim_{n\to\infty}\frac{\mu\big(E_{F_n}\big)}{\abs{F_n}}
\end{align}
for every $\GG$-invariant measure $\mu$ on $\Omega$.
From~\eqref{eq:Hamiltonian:residue:Folner} it follows that the above equality remains valid
if we replace $E_{F_n}$ with $E_{F_n|F_n^\complement}$.

\subsection{Relative pressure}

Let $E$ be the relative Hamiltonian associated to a relative absolutely summable interaction $\Phi$
and $\mu$ a $\GG$-invariant probability measure on $\Omega$.
For every $A\Subset\GG$, we define
\begin{align}
	\Psi_\mu(A) &\isdef H_\mu(\xi^A\,|\,\field{F}_\Theta) - \mu(E_A) \;,
\end{align}
where $H_\mu(\xi^A\,|\,\field{F}_\Theta)$ denotes the conditional entropy
of $\xi^A$ given $\field{F}_\Theta$ under $\mu$.
This is the \emph{relative pressure} on $A$ under $\mu$.

The \emph{relative pressure per site} under $\mu$ is given by
\begin{align}
	\psi(\mu) &\isdef \lim_{n\to\infty} \frac{1}{\abs{F_n}}\Psi_\mu(F_n)
\end{align}
where $(F_n)_{n\in\NN}$ is an arbitrary F\o{}lner sequence in $\GG$.
It can be verified using~\cite[Sec.~\S4.7]{KerLi16}
and~\eqref{eq:interaction-observable:equivalence:2}
that the limit exists, is independent of
the choice of the F\o{}lner sequence, and coincides with
$h_\mu(\Omega\,|\,\Theta)-\mu(f_\Phi)$,
where $h_\mu(\Omega\,|\,\Theta)$
denotes the conditional entropy per site
(i.e., the conditional Kolmogorov--Sinai entropy for the $\GG$-action)
of $\mu$ given the $\GG$-invariant $\sigma$-algebra $\field{F}_\Theta$,
and $f_\Phi$ is the energy observable associated to the interaction~$\Phi$.

The \emph{conditional} relative pressure on $A\Subset\GG$ given $B\Subset\GG$
under $\mu$ is
\begin{align}
	\Psi_\mu(A\,|\,B) &\isdef \Psi_\mu(A\cup B) - \Psi_\mu(B) \\
	&=
		H_\mu(\xi^A\,|\,\xi^B\lor\field{F}_\Theta) - \mu(E_{A|B}) \;.
\end{align}
The advantage of this definition is that it has formal properties
similar to those of conditional entropy $H_\mu(\xi^A\,|\,\xi^B)$
and conditional energy $E_{A|B}$.
Most importantly, the conditional relative pressure satisfies the chain rule
\begin{align}
\label{eq:pressure:chain-rule}
	\Psi_\mu(A\cup B\,|\,C) &= \Psi_\mu(B\,|\,C) + \Psi_\mu(A\,|\,B\cup C) \;.
\end{align}
Observe that $\Psi_\mu(A\,|\,B)$ depends only on the restriction of $\mu$
to $\field{F}_\Theta\lor\xi^{A\cup B}$.
Moreover,
\begin{align}
\label{eq:pressure:bound}
	\Psi_\mu(A\,|\,B) &\leq (\log\abs{\Sigma}+\norm{\Phi})\abs{A\setminus B} \;.
\end{align}
The martingale convergence theorem, the monotonicity of
conditional entropy on the condition, the absolute summability of the interaction
(in particular, the existence of the limit in~\eqref{eq:energy:conditional:limit})
and the bounded convergence theorem imply the existence of the limit
\begin{align}
	\Psi_\mu(A\,|\,A^\complement) &\isdef
		\lim_{B\nearrow\GG}\Psi\big(A\,\big|\,(B\setminus A)\big) \\
	&=
		H_\mu(\xi^A\,|\,\xi^{A^\complement}\lor\field{F}_\Theta)
		- \mu(E_{A|A^\complement}) \;.
\end{align}

Let us remark that for a fixed $A\Subset\GG$ and a measure $\nu$ on $\Theta$,
the conditional entropy $H_\mu(\xi^A\,|\,\field{F}_\Theta)$ and as a result
the relative pressure $\Psi_\mu(A)$ are concave
as functions of $\mu$
when $\mu$ runs over measures with marginal $\nu$.
In turn, the conditional entropy per site $h_\mu(\Omega\,|\,\Theta)$ and
the relative pressure per site $\psi(\mu)$ are affine
when restricted to measures $\mu$ with marginal $\nu$.

\subsection{Relative Gibbs measures and relative equilibrium measures}

According to a fundamental hypothesis of equilibrium statistical mechanics,
the macroscopic states of a system at thermal equilibrium are suitably described
by probability distributions maximizing the pressure.
Identifying the equilibrium measures thus amounts to solving an optimization problem,
where the pressure is interpreted as the gain.

On a finite space, the optimization problem is solved by
the Boltzmann distribution.
The uniqueness of the solution is a consequence of the strict concavity of the entropy.
\begin{proposition}[Finitary variational principle]
\label{prop:gibbs-inequality}
	Let $M$ be a finite set
	and $U\colon M\to\RR$ a real-valued function.
	Given a probability distribution $p\colon M\to[0,1]$, define
	\begin{align}
		\Psi(p) &\isdef H(p) - p(U) \;.
	\end{align}
	Then, $\Psi(p)$
	takes its maximum if and only if $p(a)=\ee^{-U(a)}/Z$ for each $a\in M$,
	where $Z\isdef\sum_{a\in M}\ee^{-U(a)}$ is the normalizing constant.
	The maximum value is $\log Z$.
\end{proposition}
\noindent This is well known and easily follows from Jensen's inequality.

A relative Gibbs measure for an absolutely summable relative interaction $\Phi$
is a probability measure on $\Omega$ that is \emph{locally optimal},
in the sense that it maximizes the pressure
on every finite region of the lattice $\GG$ conditioned on
the configuration outside the region and the environment.
In other words, a probability measure $\mu$ on $\Omega$ is a relative Gibbs measure for $\Phi$
if for every $A\Subset\GG$, the conditional probability according to $\mu$
of seeing a pattern $u$ on $A$ given a configuration $x_{A^\complement}$
outside $A$ and an environment $\theta$ is the Boltzmann distribution associated to
the energy function $U(u)\isdef E_{A|A^\complement}(\theta,x_{A^\complement}\lor u)$,
where $E$ is the Hamiltonian associated to $\Phi$.

More specifically, %
for every $A\Subset\GG$, the prescribed distribution of
the pattern on $A$ given a boundary condition $\theta,x_{A^\complement}$
is the Boltzmann distribution
\begin{align}
	\label{eq:boltzman-dist}
	\pi_{\theta,x_{A^\complement}}(u) &\isdef
		\begin{dcases}
			\frac{1}{Z_{A|A^\complement}(\theta,x)}
				\ee^{-E_{A|A^\complement}(\theta,x_{A^\complement}\lor u)}
					& \text{if $x_{A^\complement}\lor u\in X_\theta$,} \\
				0	& \text{otherwise,}
		\end{dcases}	
\end{align}
where $Z_{A|A^\complement}(\theta,x)$ is the normalizing constant known as the \emph{partition function}.
Given $(\theta,x)\in\Omega$, the distribution $\pi_{\theta,x_{A^\complement}}(u)$ extends to
a probability measure $K_A\big((\theta,x),\cdot\big)$ on $\Omega$ by setting
\begin{align}
	K_A\big((\theta,x),[u]\cap W\big) &\isdef
		\indicator{W}(\theta,x)\pi_{\theta,x_{A^\complement}}(u)
\end{align}
for each $u\in\Sigma^A$ and $W\in\field{F}_\Theta\lor\xi^{\GG\setminus A}$.
It can be verified that given a set $W\in\field{F}_\Omega$,
the function $K_A(\cdot,W)$ is measurable.
A probability measure $\mu$ on $\Omega$ is a \emph{relative Gibbs} measure for $\Phi$
if for every $A\Subset\GG$ and each measurable set $W\in\field{F}_\Omega$,
\begin{align}
\label{eq:relative-Gibbs:def:conditioning}
	\mu\big(W\,\big|\,\field{F}_\Theta\lor\xi^{\GG\setminus A}\big)(\cdot) &=
		K_A(\cdot,W)
\end{align}
$\mu$-almost surely.
Notice that in order for $\mu$ to be a relative Gibbs measure, it is sufficient
that the above equality holds for every $W\in\xi^A$.

We refer to the function $K_A(\cdot,\cdot)$ as the \emph{Gibbs kernel} for set $A$.
Every Gibbs kernel $K_A$ naturally defines a linear operator $\nu\mapsto\nu K_A$
on probability measures on $\Omega$ by
\begin{align}
	(\nu K_A)(W) &\isdef \nu\big(K_A(\cdot,W)\big) = \int K_A(\cdot,W)\, \dd\nu \;,
\end{align}
and its adjoint operator $f\mapsto K_A f$ on bounded measurable observables on $\Omega$
by
\begin{align}
	(K_A f)(\omega) &\isdef K_A(\omega,f)
		= \int f(\omega') \, K_A(\omega,\dd\omega') \;.
\end{align}
If $(\pmb{\theta},\pmb{x})$ is a random point chosen according to $\nu$,
the measure $\nu K_A$ can be interpreted as the distribution of $(\pmb{\theta},\pmb{x})$
after resampling the pattern on $A$ according to
$\pi_{\pmb{\theta},\pmb{x}_{A^\complement}}$.

With the above definition, one can see that a measure $\mu$
is relative Gibbs for $\Phi$ if and only if $\mu K_A=\mu$
for every $A\Subset\GG$.
The collection $K=(K_A:A\Subset\GG)$ of the Gibbs kernels for all $A\Subset\GG$
is referred to as the \emph{relative Gibbs specification} associated to $\Phi$.

The local optimality of relative Gibbs measures is an immediate
consequence of Proposition~\ref{prop:gibbs-inequality}.
For future reference, let us spell this out as a corollary
in the following specific way.

\begin{corollary}[Local optimality of relative Gibbs measures]
\label{cor:gibbs:local-optimality}
	Let $\Phi$ be an absolutely summable relative interaction on $\Omega$.
	Let $\Psi$ denote the pressure associated to $\Phi$,
	and $K_A$ the Gibbs kernel associated to $\Phi$ for a set $A\Subset\GG$.
	Then,
	for every probability measure $\mu$ on $\Omega$	and $\mu$-almost every $(\theta,x)\in\Omega$,
	we have \useshortskip %
	\begin{align}
		\label{eq:gibbs:local-optimality:pointwise}
		\MoveEqLeft[3]\nonumber
		H_{\mu(\cdot|\field{F}_\Theta\lor\xi^{\GG\setminus A})(\theta,x)}(\xi^A)
			- \mu\big(E_{A|A^\complement}\,|\,\field{F}_\Theta\lor\xi^{\GG\setminus A}\big)(\theta,x) \\
		&\leq
		H_{K_A((\theta,x),\cdot)}(\xi^A)
			- K_A\big((\theta,x),E_{A|A^\complement}\big)
		\quad \text{\textup{(} $= \log Z_{A|A^\complement}(\theta,x)$ \textup{)}}
	\end{align}
	with equality if and only if $\mu(\cdot\,|\,\field{F}_\Theta \lor\xi^{\GG\setminus A})(\theta,x)=K_A((\theta,x),\cdot)$.
	In particular,
	\begin{align}
		\label{eq:gibbs:local-optimality:integral}
		\Psi_{\mu}(A\,|\,A^\complement) &\leq \Psi_{\mu K_A}(A\,|\,A^\complement) \;,
	\end{align}
	with equality if and only if $\mu K_A=\mu$.
\end{corollary}
\begin{proof}
	To obtain the first inequality,
	apply the finitary variational principle (Prop.~\ref{prop:gibbs-inequality})
	with $M\isdef\{u\in\Sigma^A: x_{A^\complement}\lor u\in X_\theta\}$,
	$U(u)\isdef E_{A|A^\complement}(\theta,x_{A^\complement}\lor u)$, and
	$p(u)\isdef \mu([u]\,|\,\field{F}_\Theta\lor\xi^{\GG\setminus A})(\theta,x)$.
	The second inequality follows from the first inequality
	by integrating both sides with respect to $\mu$.
\end{proof}
\noindent An alternative way to think about the above corollary is that
applying a Gibbs kernel $K_A$ on a measure locally optimizes that measure
on the set $A$.

A relative equilibrium measure on $\Omega$ is a $\GG$-invariant measure that is
\emph{globally optimal} among all $\GG$-invariant measures with the same marginal on $\Theta$.
More generally, let $\nu$ be a $\GG$-invariant measure on~$\Theta$
and $f\in C_\Theta(\Omega)$ an arbitrary bounded relatively continuous observable
(i.e., not necessarily one associated to an absolutely summable interaction).
An \emph{equilibrium measure} for $f$ \emph{relative to} $\nu$ is
a $\GG$-invariant measure on $\Omega$ with marginal $\nu$ on $\Theta$
which maximizes the relative pressure $h_\mu(\Omega\,|\,\Theta)-\mu(f)$
among all $\GG$-invariant measures with the same marginal $\nu$ on $\Theta$.
A measure that is an equilibrium measure relative to its marginal on $\Theta$
is simply said to be a relative equilibrium measure.

\subsection{Types of constraints on configurations}
\label{sec:hypoth}

In this section, we define various classes of constraints on configuration spaces that are sufficient for
the relative Dobrushin--Lanford--Ruelle theorem.

A \emph{subshift} on $\GG$ (or a \emph{$\GG$-subshift}) is a closed $\GG$-invariant
subset of $\Sigma^\GG$.
A subshift $X\subseteq\Sigma^\GG$ is \emph{of finite type} (\emph{SFT})
if there exists a finite set $F\Subset\GG$ and a subset $\family{F}\subseteq\Sigma^F$
such that $x\in X$ if and only if $(g^{-1}x)_F\notin\family{F}$ for all $g\in\GG$.
The elements of $\family{F}$ are called the \emph{forbidden} patterns defining $X$.

\subsubsection{Conditions for the Lanford--Ruelle direction.}
\label{subsec:cond_LR}

The classical constraint on the set of configurations $X$ which enables the
Lanford--Ruelle direction in Theorem~\ref{thm:DLR} is that $X$ is an SFT.
However, this is used only in the form of a Markovian property: the possible
configurations that may appear in a finite support $A\Subset\GG$ given a fixed
configuration in $\GG \setminus A$ do not depend upon the whole complement
but only on a finite subset $B \supseteq A$. 

We say that a closed set $X \subseteq \Sigma^\GG$ satisfies the
\emph{topological Markov property} (\emph{TMP}) if for all
$A \Subset \GG$ there exists a finite set $B \supseteq A$ such that
whenever  $x,x' \in X$ satisfy $x_{B \setminus A} =  x'_{B \setminus A}$,
then $x_B \lor x'_{\GG \setminus A} \in X$.
We call such a set $B$ a \emph{memory set} for $A$ in $X$.
Equivalently, one may think of this property
as follows: if $x_{B \setminus A} =  x'_{B \setminus A}$, then $x_B$ and
$x'_B$ are interchangeable in the sense that every appearance of $x_B$
may be replaced by $x'_B$ and vice-versa.

In the relative setting, similar notions of relative SFT and relative TMP
can be formulated as follows.
Let $\Omega\subseteq\Theta\times\Sigma^\GG$ be a relative system.
We say that $\Omega$ is a \emph{relative SFT} (or is an \emph{SFT relative to $\Theta$})
if there is a finite set $F \Subset \GG$
and a family of subsets $\family{F}_\theta\subseteq\Sigma^F$ (for $\theta\in\Theta$)
such that for each $\theta\in\Theta$, we have
$x\in X_\theta$ if and only if $(g^{-1}x)_F\notin\family{F}_{g^{-1}\theta}$
for all $g\in\GG$.
Similarly, we say that $\Omega$ satisfies the
\emph{topological Markov property relative to $\Theta$} (\emph{relative TMP})
if all the sets $X_\theta$ (for $\theta\in\Theta$) satisfy the TMP
with common choices of the memory sets.
In other words, $\Omega$ has relative TMP if for every $A\Subset\GG$,
there is a finite set $B\supseteq A$ such that 
whenever $\theta\in\Theta$ and $x,x' \in X_\theta$ satisfy $x_{B \setminus A} =  x'_{B \setminus A}$,
then $x_B \lor x'_{\GG \setminus A} \in X_\theta$.

Given a $\GG$-invariant measure $\nu$ on the environment space $\Theta$,
we can also consider the more relaxed conditions of
SFT \emph{relative to $\nu$} and TMP \emph{relative to $\nu$}
under which the corresponding conditions are satisfied for $\nu$-almost every $\theta\in\Theta$
rather than for all $\theta\in\Theta$.
However, by removing a null set from $\Theta$, we can always turn the system
into one that satisfies the condition surely.

Observe that for a subshift $X$ that satisfies TMP,
if $B$ is a memory set for $A$ in $X$, then for all $g \in \GG$, $gB$ is a memory set for $gA$ in $X$.
Similarly, for a relative system that satisfies TMP, it follows that
if $B$ is a memory set for $A$, then $gB$ is a memory set for $gA$.

A closed set $X\subseteq\Sigma^\GG$ 
satisfies the \emph{strong topological Markov property} (\emph{strong TMP})
if there is a finite set $F \Subset \GG$ with $1_\GG \in F$ such that
for every finite set $A\Subset\GG$, the set $AF$ is a memory set for $A$ in $X$.
The notion of relative strong TMP is defined analogously.  %
We remark that for subshifts, TMP and strong TMP are topological conjugacy invariants
and that strong TMP is the conjugacy invariant class generated by the class of topological Markov
fields as defined in~\cite{ChaHanMarMeyPav14, ChaMey16}.

Clearly, every SFT satisfies the strong TMP, and the strong TMP implies the TMP.
Moreover, these collections are all distinct.
The class of subshifts with strong TMP, %
is much larger than the class of SFTs in the sense that the latter is countable
while the former are uncountable:
if $X$ is any $\ZZ$-subshift over $\Sigma$, then the set of all configurations on $\ZZ^2$
whose rows are elements of $X$ and whose columns are constant
satisfies strong TMP but is not necessarily an SFT.
See also \cite{ChaHanMarMeyPav14} (bottom of page 233)
for a simple example of a $\ZZ$-subshift which satisfies the strong TMP but is not an SFT. In Section~\S\ref{subsec:examples} below, we provide an example
of a $\ZZ^2$-subshift that satisfies TMP but not strong TMP.

\subsubsection{Conditions for the Dobrushin direction.}
\label{subsec:cond_D}

The notion of D-mixing used in Theorem~\ref{thm:DLR}\ref{thm:DLR:D}
was introduced by Ruelle~\cite[Sec.~\S4.1]{Rue04} for SFTs on $\ZZ^d$
but remains meaningful in a more general setting.

Let $Y$ be a closed subset of $\Sigma^\GG$.
Given $A \Subset \GG$, we say that a finite set
$B \supseteq A$ is a \emph{mixing set} for $A$ in $Y$ if for every $y,y' \in Y$,
there exists $z \in Y$ satisfying $z_{A} = y_{A}$ and $z_{\GG \setminus B} = y'_{\GG \setminus B}$.
In other words, we can paste a pattern on $A$ from a configuration $y\in Y$
into any other configuration $y'\in Y$ provided that we modify the annulus $B\setminus A$.
We say that $Y$ is \emph{Dobrushin-mixing} (or \emph{D-mixing})
\emph{with respect to} a F\o{}lner sequence $(F_n)_{n\in\NN}$ if for each $n$,
there is a mixing set $\overline{F}_n$ for $F_n$ in $Y$ such that
$\bigabs{\overline{F}_n\setminus F_n}=\smallo(\abs{F_n})$ as $n\to\infty$.
We say that $Y$ satisfies \emph{D-mixing} if it satisfies D-mixing with
respect to some F\o{}lner sequence $(F_n)_{n\in\NN}$.

A relative version of the D-mixing property suitable for our purposes
is the following.
Let $\Omega$ be a relative system and $\nu$ a $\GG$-invariant measure
on its environment space $\Theta$.
We say that $\Omega$ satisfies \emph{D-mixing relative to $\nu$}
\emph{with respect to a F\o{}lner sequence $(F_n)_{n\in\NN}$} if
for $\nu$-almost every $\theta \in \Theta$ and each $n\in\NN$,
there is a mixing set $F_n^\theta$ for $F_n$ in $X_\theta$
such that $\abs{F^\theta_n\setminus F_n}$ is measurable and
$\int\bigabs{F^\theta_n\setminus F_n}\dd\nu(\theta)=\smallo(\abs{F_n})$
as $n\to\infty$.  We say that $\Omega$ satisfies \emph{D-mixing relative to $\nu$}
if it satisfies D-mixing relative to $\nu$ with respect to
some F\o{}lner sequence $(F_n)_{n\in\NN}$.

There are two stronger notions which imply D-mixing and are better known in
the symbolic dynamics community. We say that $Y$ satisfies the
\emph{uniform filling property} (\emph{UFP}) with respect to a F\o lner sequence
$(F_n)_{n \in \NN}$ if there exists a finite set $F \Subset \GG$ such that $F_nF$ is a mixing
set for $F_n$. We say that $Y$ satisfies the \emph{UFP} if $Y$ satisfies the UFP with
respect to some F\o lner sequence. We say that $Y$ is \emph{strongly irreducible}
(\emph{SI}) if there exists $F \Subset \GG$ such that for every two finite sets
$A,B\Subset\GG$ satisfying $AF \cap BF = \varnothing$ and every two configurations $y,y'\in Y$
there is a configuration $z \in Y$ such that $z_A = y_A$ and $z_B = y'_B$. 

The UFP can be regarded as a uniform version of D-mixing:
the fact that $(F_n)_{n \in \NN}$ is F\o lner ensures that
$\abs{F_nF \setminus F_n} = o(\abs{F_n})$.
In turn, a compactness argument shows that SI implies
the UFP. %
An example of a
subshift satisfying UFP but not SI is given in~\cite{KasMad13}.
We do not know of any example of a D-mixing subshift which does not satisfy the UFP.

The relative versions of SI and the UFP
are defined analogously.

\subsubsection{Conditions implying both directions of the theorem}
\label{subsec:cond_extra}

A natural condition that implies both directions of the theorem in the non-relative
setting is that $X$ is the full $\GG$-shift $\Sigma^{\GG}$. In the terminology of
TMP, a $\GG$-subshift is a full $\GG$-shift if and only if
every set $A\Subset\GG$ is a memory set for itself.
In other words, the symbol at each site can be changed \emph{independently}
of the rest of the configuration.
The only $\GG$-subshift
with alphabet $\Sigma$ satisfying that property is $\Sigma^{\GG}$.  However,
the relative version of this notion turns out to be more interesting.

We say that a relative system $\Omega$ has the \emph{relative independence property}
(or \emph{independence property relative to $\Theta$})
if every finite set is a memory set for itself, that is,
if for every $\theta\in\Theta$, every finite set $A \Subset \GG$ and every pair
$x,x' \in X_\theta$, we have $x_A \lor x'_{\GG \setminus A} \in X_\theta$.
Equivalently, $\Omega$ has the relative independence property if for each $\theta\in\Theta$ and $A \Subset \GG$,
any two elements of $L_A(X_\theta)$ are interchangeable in $X_\theta$.
Independence property \emph{relative to} a measure $\nu$ on $\Theta$
is defined accordingly.
Note that, as in the case of a relative SFT,
there is no need for all the sets $X_{\theta}$ to be the same.

Every relatively independent system satisfies the TMP,
and moreover, is relatively D-mixing (with $F_n^\theta\isdef F_n$).
Therefore, in a relatively independent system, both hypotheses of the
relative Dobrushin--Lanford--Ruelle theorem hold.
In Section~\S\ref{sec:relative_tom},
we shall show that under simple reductions, the Lanford--Ruelle theorem for
relatively independent systems implies Theorem~\ref{thm:DLR:relative}~\ref{thm:DLR:LR:relative} for
relative systems with relative TMP.

There is a notion introduced by Brice{\~{n}}o~\cite{Bri17}
that is less restrictive than independence but 
still implies both conditions of the Dobrushin--Lanford--Ruelle theorem.
A closed subset $Y\subseteq \Sigma^{\GG}$ is \emph{topologically strong spatial mixing} (\emph{TSSM}) 
if there exists a set $F \Subset \GG$ such that for any finite disjoint
sets $A,B,S \Subset \GG$ such that $AF \cap BF = \varnothing$,
and for every $u \in L_A(Y)$, $v\in L_B(Y)$ and $w\in L_S(Y)$
such that $u\lor w \in L_{A\cup S}(Y)$ and $v\lor w \in L_{B\cup S}(Y)$, we have
$u\lor v\lor w \in L_{A\cup B\cup S}(Y)$.
In fact, TSSM implies both SI and SFT (see~\cite{Bri17}).

\begin{proposition}[TSSM $\Longrightarrow$ SI SFT]
\label{prop:TSSM_is_SFT_SI}
	Let $X$ be a TSSM $\GG$-subshift. Then $X$ is a strongly irreducible SFT.
\end{proposition}

\begin{proof}
	The set $S$ in the definition of TSSM can be chosen to be empty,
	hence $X$ is SI. In order to show that $X$ is an SFT, let $F$ be the set appearing in the
	definition of TSSM, and without loss of generality, assume that $F\ni 1_\GG$.
	Let $\family{F}\subseteq\Sigma^{FF^{-1}}$
	be the set of all patterns on $FF^{-1}$ that do not occur
	on the elements of $X$, that is,
	$\family{F}\isdef\{q\in\Sigma^{FF^{-1}}: \text{$x_{FF^{-1}}\neq q$ for all $x\in X$}\}$.
	Let $X'$ be the SFT defined by $\family{F}$ as the set of forbidden patterns.
	We show that $X=X'$.
	
	Clearly, $X\subseteq X'$.  Suppose that there exists a configuration
	$y\in X'\setminus X$.
	Let $D\subseteq\GG$ be a minimal set containing $FF^{-1}$
	such that $x_D\neq y_D$ for all $x\in X$.
	By compactness, $D$ is finite and thus $y_D\notin L_D(X)$.
	By the definition of $X'$, we have $y_{FF^{-1}}\in L_{FF^{-1}}(X)$,
	hence there exists an element $g\in D\setminus FF^{-1}$.
	Now, set $A\isdef\{1_\GG\}$, $B\isdef\{g\}$ and $S\isdef D\setminus\{1_\GG,g\}$.
	Observe that $AF\cap BF=F\cap gF=\varnothing$.
	By the choice of $D$,
	we have $y_{A\cup S}\in L_{A\cup S}(X)$ and $y_{B\cup S}\in L_{B\cup S}(X)$
	but $y_D\notin L_D(X)$.  But this contradicts the TSSM property of~$X$.
\end{proof}

The relative version of TSSM demands the existence of a set $F$ for which the
above condition holds on $X_{\theta}$ for all (or $\nu$-almost every) $\theta \in \Theta$.
It can be verified that relative TSSM implies both relative SI and relative SFT.
In Figure~\ref{fig:conditions}, we summarize the relationships between
all of the conditions introduced in this section. %
The same relations hold for their relative counterparts.

\begin{figure}[h!]
	\centering
	\begin{tikzpicture}
	\draw[fill = gray, opacity = 0.3, rounded corners] (-1.7,-1.5) rectangle (10.8,0.5);
	\draw[pattern = north east lines, opacity = 0.2, rounded corners] (-1.7,-1.5) rectangle (10.8,0.5);
	\draw[fill = gray, opacity = 0.3, rounded corners] (-1.2,-0.5) rectangle (11.2,1.5);
	\draw[pattern = north west lines, opacity = 0.2, rounded corners] (-1.2,-0.5) rectangle (11.2,1.5);
	\node (B) at (2,0) {TSSM};
	\node (BB) at (-0.5,0) {Indep.};
	\node (C) at (4.5,1) {SFT};
	\node (D) at (7,1) {strong TMP};
	\node (E) at (10,1) {TMP};
	\node (F) at (4.5,-1) {SI};
	\node (G) at (7,-1) {UFP};
	\node (H) at (9.5,-1) {D-Mixing};
	\draw [->] (BB) to (B);
	\draw [->] (B) to (C);
	\draw [->] (C) to (D);
	\draw [->] (D) to (E);
	\draw [->] (B) to (F);
	\draw [->] (F) to (G);
	\draw [->] (G) to (H);
	\node  at (0.7,1.2) {\textsf{Eq. measure $\implies$ Gibbs}};
	\node  at (0.2,-1.2) {\textsf{Gibbs $\implies$ Eq. measure}};
	\end{tikzpicture}
	\caption{Sufficient conditions for both directions of the (relative) Dobrushin--Lanford--Ruelle theorem.}
	\label{fig:conditions}
\end{figure}
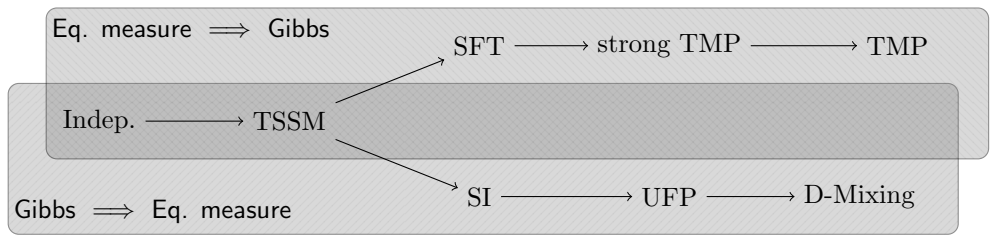

\subsubsection{Examples}
\label{subsec:examples}

This last subsection is dedicated to examples that illustrate
the conditions introduced in the previous subsections.
We begin by examining the two examples given in the Introduction.

In Example~\ref{ex:ising_model} the environment $\theta$ completely
determines the allowed symbols at each site, independently from site to site.
Thus, the Ising model on percolation clusters satisfies the relative independence
property and so the relative Dobrushin--Lanford--Ruelle theorem holds.
Note that in this case, all the sets~$X_{\theta}$ are disjoint.

In Example~\ref{ex:graphs} the coloring condition is a nearest neighbor
condition, and so the relative system is a relative SFT and therefore the
relative Lanford--Ruelle theorem holds.
Moreover, we claim that the assumption $\abs{\Sigma}>\abs{S}$
is sufficient to ensure that the system is relative SI,
and thus the relative Dobrushin theorem holds.
Indeed, let $F \isdef S$ be the set
of generators. Given a subgraph $\theta$ and two valid colorings
$x,x' \in X_{\theta}$, if $AF \cap BF = \varnothing$ then no vertex in $A$ is
adjacent to a vertex in~$B$. The partial configuration $w\isdef x_{A} \lor x'_{B}$
can be inductively extended to a valid coloring of $(\GG,\EE)$ by filling each position in
$\GG\setminus A \cup B$ with a color not already taken by any of its $\abs{S}$
neighbors along the generators.
In fact, the same argument shows that the system is indeed relatively TSSM (see~\cite{Bri17}).

\newcommand{\unobox}{
	\begin{tikzpicture}[scale = 0.25,baseline=0.4pt]
	\draw  (0,0) rectangle +(1,1);
	\draw [fill = black] (0.5,0.5) circle (0.3);
	\end{tikzpicture}
}
\newcommand{\cerobox}{
	\begin{tikzpicture}[scale = 0.25,baseline=0.4pt]
	\draw (0,0) rectangle +(1,1);
	\end{tikzpicture}
}

\begin{example}[A $\ZZ^2$-subshift with TMP but not strong TMP]
	Consider the set $Y$ of all configurations $y \in \{\cerobox,\unobox\}^{\ZZ^2}$ such that every
	$8$-connected component of sites with symbol $\unobox$ is a finite square,
	that is, every $8$-connected component
	of $y^{-1}(\unobox)$ is a set of the form $\vec{u}+[0,n-1]^2\cap \ZZ^2$ for some $\vec{u} \in \ZZ^2$
	and $n \geq 1$. Let $X \subseteq \{\cerobox,\unobox\}^{\ZZ^2}$ be the closure of $Y$.
	We claim that $X$ has TMP but not strong TMP.
	
	In order to see that $X$ does not have strong TMP, let $A_n \isdef [-2n,2n]^2$ and consider $x$ to be equal
	to $\unobox$ in the support $(3n,0)+[-n,n]^2$ and $\cerobox$ everywhere else. Similarly, consider
	$y$ equal to $\unobox$ in $(3n+1,0)+[-n,n]^2$ and $\cerobox$ in the complement.
	Let $B_n \isdef [-4n, 4n]^2$. For any fixed $F \Subset \ZZ^2$ we have that $A_nF \subseteq B_n$
	for all large enough $n$. However, we have that $x_{B_n \setminus A_n} = y_{B_n \setminus A_n}$
	but $z \isdef x_{B_n} \vee y_{A_n^{\complement}}$ is not an element of $X$ as
	$z^{-1}(\unobox) = (2n,0)+([0,n+1]\times [0,n])$ is not a square.
	This is illustrated in Figure~\ref{fig:squarestmp}.

	\begin{figure}[h!]
		\centering
		\begin{tikzpicture}[scale = 0.15]
			\begin{scope}[shift = {(0,0)}]
				\node at (0,-14)  {\scalebox{1}{$x \in X$}};
				\clip[draw,decorate,decoration={random steps, segment length=4pt, amplitude=2pt}] (-11.8,-11.8) rectangle (11.8,11.8); 
				\draw [black!20] (-12,-12) grid (12,12);

				\foreach \i in {-8,...,8}{
					\foreach \j in {-8,...,8}{
						\draw [black!20, pattern = north east lines, opacity = 0.3] (\i,\j) rectangle +(1,1);
					}
				}
				\foreach \i in {-4,...,4}{
					\foreach \j in {-4,...,4}{
						\draw [black!20, fill = white] (\i,\j) rectangle +(1,1);
						\draw [black!20, fill = black!50, opacity=0.3] (\i,\j) rectangle +(1,1);
					}
				}
				\foreach \i in {-2,...,2}{
					\foreach \j in {-2,...,2}{
						\draw [fill = black] (6.5+\i,0.5+\j)circle (0.3);
					}
				}
				\node at (-2,-2)  {\scalebox{1}{$A_n$}};
				\node at (-6,-6)  {\scalebox{1}{$B_n$}};
			\end{scope}
			\begin{scope}[shift = {(30,0)}]
				\node at (0,-14)  {\scalebox{1}{$y \in X$}};
				\clip[draw,decorate,decoration={random steps, segment length=4pt, amplitude=2pt}] (-11.8,-11.8) rectangle (11.8,11.8); 
				\draw [black!20] (-12,-12) grid (12,12);

				\foreach \i in {-8,...,8}{
					\foreach \j in {-8,...,8}{
						\draw [black!20, pattern = north east lines, opacity = 0.3] (\i,\j) rectangle +(1,1);
					}
				}
				\foreach \i in {-4,...,4}{
					\foreach \j in {-4,...,4}{
						\draw [black!20, fill = white] (\i,\j) rectangle +(1,1);
						\draw [black!20, fill = black!50, opacity=0.3] (\i,\j) rectangle +(1,1);
					}
				}
				\foreach \i in {-2,...,2}{
					\foreach \j in {-2,...,2}{
						\draw [fill = black] (7.5+\i,0.5+\j)circle (0.3);
					}
				}
				\node at (-2,-2)  {\scalebox{1}{$A_n$}};
				\node at (-6,-6)  {\scalebox{1}{$B_n$}};
			\end{scope}
	
			\begin{scope}[shift = {(60,0)}]
				\node at (0,-14)  {\scalebox{1}{$x_{A_n} \vee y_{A_n^{\complement}} \notin X$}};
				\clip[draw,decorate,decoration={random steps, segment length=4pt, amplitude=2pt}] (-11.8,-11.8) rectangle (11.8,11.8); 
				\draw [black!20] (-12,-12) grid (12,12);

				\foreach \i in {-8,...,8}{
					\foreach \j in {-8,...,8}{
						\draw [black!20, pattern = north east lines, opacity = 0.3] (\i,\j) rectangle +(1,1);
					}
				}
				\foreach \i in {-4,...,4}{
					\foreach \j in {-4,...,4}{
						\draw [black!20, fill = white] (\i,\j) rectangle +(1,1);
						\draw [black!20, fill = black!50, opacity=0.3] (\i,\j) rectangle +(1,1);
					}
				}
				\foreach \i in {-2,...,3}{
					\foreach \j in {-2,...,2}{
						\draw [fill = black] (6.5+\i,0.5+\j)circle (0.3);
					}
				}
				\node at (-2,-2)  {\scalebox{1}{$A_n$}};
				\node at (-6,-6)  {\scalebox{1}{$B_n$}};
			\end{scope}
			
		\end{tikzpicture}
		\caption{Two configurations $x$ and $y$ which coincide in $B_n\setminus A_n$ but cannot be put together.}
		\label{fig:squarestmp}
	\end{figure}
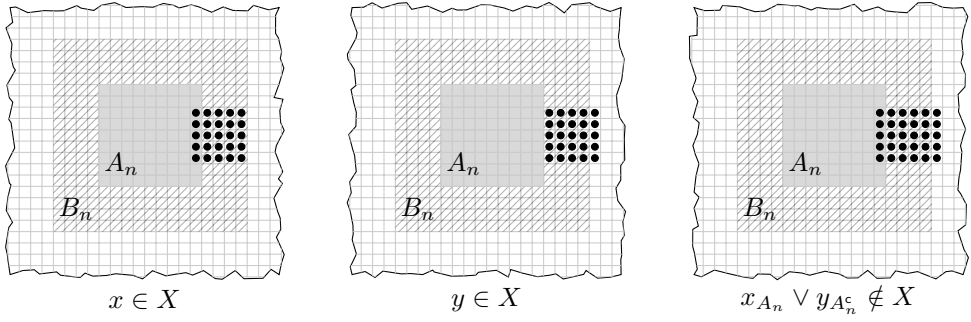

	To see that $X$ has TMP, without loss of generality let $A = [-n,n]^2$.
	We claim that $B \isdef [-10n,10n]^2$ is a memory set for $A$ for every $n \geq 1$. Indeed, let $x,y \in X$ be
	such that $x_{B\setminus A} = y_{B \setminus A}$. We must show that
	$z \isdef x_{A} \vee y_{A^{\complement}} \in X$. 

	First, we claim that we can reduce to the case where $x,y \in Y$, that is, where every
	square is finite.
	Indeed, assume that  $x,y\in X\setminus Y$.  %
	By definition, there are sequences
	$(x^{(m)})_{m \in \NN}$ and $(y^{(m)})_{m \in \NN}$ in~$Y$ converging to
	$x$ and $y$ respectively.
	For $m$ large enough, we have $x^{(m)}_{B\setminus A}=y^{(m)}_{B\setminus A}$.
	The sequence of configurations
	$z^{(m)} \isdef x_B^{(m)} \vee y_{A^{\complement}}^{(m)}$ converges to $z$.
	Thus if we know that each $z^{(m)}$ is in $Y$, then we can conclude that $z \in X$.
	Now, observe that for $m$ large enough, $x_B^{(m)} = x_{B}$ and $y_B^{(m)} = y_{B}$,
	and in particular $x_{B\setminus A}^{(m)}=y_{B\setminus A}^{(m)}$.

	So, let $x,y \in Y$ and suppose that $z \notin Y$. Then there exists an $8$-connected region
	$R_z \subseteq z^{-1}(\unobox)$ which is either an infinite set,
	or a finite set that is not a square.
	The first case cannot occur, as every connected component of $y^{-1}(\unobox)$
	is bounded and $z$ differs from $y$ only on the finite set $A$.
	Therefore, $R_z$ is finite but is not a square.

	As $y \in Y$, we have that $R_z \cap A \neq \varnothing$.
	Let $R_x,R_y \Subset \ZZ^2$ be the finite squares in
	$x^{-1}(\unobox)$ and $y^{-1}(\unobox)$ which contain $R_z \cap A$ and
	$R_z \cap A^{\complement}$ respectively.
	Denote the $8$-boundary of $A$ by $\partial A$,
	that is, $\partial A\isdef[-n-1,n+1]^2\setminus[-n,n]^2$.
	Let $N\isdef [-n-1,n+1] \times \{n+1\}$, $S\isdef [-n-1,n+1] \times \{-n-1\}$,
	$W\isdef \{-n-1\} \times [-n-1,n+1]$ and $E \isdef\{n+1\} \times [-n-1,n+1]$
	be respectively the north, south, west and east $8$-boundaries of $A$,
	so that $\partial A= N\cup E\cup S\cup W$.
	There are three possibilities on how $R_x$ can intersect $\partial A$.

	\begin{itemize}
		\item If $R_x \cap \partial A = \varnothing$, then necessarily
			$R_x \subseteq A$, and since $x_{B\setminus A} = y_{B \setminus A}$,
			we conclude that $R_z = R_x$, which is a square. This contradicts the assumption.
		\item If $R_x$ intersects only one of the sets $N$, $S$, $W$ and $E$,
			then the size of $R_x$ %
			can be at most $(2n+1)\times(2n+1)$, the size of $A$.
			As $R_x$ intersects $A$, %
			we deduce that $R_x\cup\partial R_x\subseteq[-10n,10n]=B$.
			Again, this implies that $R_z = R_x$,
			contradicting the assumption.
		\item If $R_x$ intersects two (or more) of the boundaries $N$, $S$, $W$ and $E$,
			such boundaries must themselves intersect,
			otherwise $R_x$ is not a square.
			Observe that if $\unobox$ appears on two diagonally adjacent sites
			in any configuration from $Y$,
			that is, either the pattern
			\begin{tikzpicture}[scale = 0.2,baseline=1pt]
				\draw (1,0) rectangle +(1,1);
				\draw (0,1) rectangle +(1,1);
				\draw (1,1) rectangle +(1,1);
				\draw (0,0) rectangle +(1,1);
				\node at (1.5,0.5) {\scalebox{0.5}{$\textbf{?}$}};
				\node at (0.5,1.5) {\scalebox{0.5}{$\textbf{?}$}};
				\draw [fill = black] (0.5,0.5) circle (0.3);
				\draw [fill = black] (1.5,1.5) circle (0.3);
			\end{tikzpicture}
			or the pattern
			\begin{tikzpicture}[scale = 0.2,baseline=1pt]
				\draw (1,0) rectangle +(1,1);
				\draw (0,1) rectangle +(1,1);
				\draw (1,1) rectangle +(1,1);
				\draw (0,0) rectangle +(1,1);
				\node at (0.5,0.5) {\scalebox{0.5}{$\textbf{?}$}};
				\node at (1.5,1.5) {\scalebox{0.5}{$\textbf{?}$}};
				\draw [fill = black] (0.5,1.5) circle (0.3);
				\draw [fill = black] (1.5,0.5) circle (0.3);
			\end{tikzpicture}
			appears, then %
			the said pattern is necessarily
			\begin{tikzpicture}[scale = 0.2,baseline=1pt]
				\draw (1,0) rectangle +(1,1);
				\draw (0,1) rectangle +(1,1);
				\draw (1,1) rectangle +(1,1);
				\draw (0,0) rectangle +(1,1);
				\draw [fill = black] (0.5,0.5) circle (0.3);
				\draw [fill = black] (1.5,1.5) circle (0.3);
				\draw [fill = black] (0.5,1.5) circle (0.3);
				\draw [fill = black] (1.5,0.5) circle (0.3);
			\end{tikzpicture}
			(i.e.~$\textbf{?}=\unobox$).
			Thus, in the current case, the information on $x_{\partial A}$
			uniquely determines $A\cap (R_z\cup\partial R_z)$.
			As $\partial A \subseteq B$ and $x_{B\setminus A} = y_{B \setminus A}$
			we conclude that $x_{A \cap R_z} = y_{A \cap R_z} \equiv \unobox$
			and $x_{A\cap\partial R_z}=y_{A\cap\partial R_z}\equiv\cerobox$.
			Thus $A \cap R_z = A \cap R_x = A \cap R_y$ and hence $R_z = R_y$ which is a square.
			This contradicts the assumption.
	\end{itemize} 

	As this example satisfies TMP, the conclusion of the Lanford--Ruelle direction
	holds even though $X$ is not an SFT nor satisfies strong TMP.
	\hfill\exampleqed
\end{example}

In Section~\S\ref{sec:groupshifts}, we will see another example of a subshift that satisfies TMP but not strong TMP.

\subsection{A topology on measures and the Feller property of Gibbs kernels}

An important consequence of the relative TMP %
is the following continuity property of the Gibbs kernels.

\begin{proposition}[Relative Feller property of Gibbs kernels]
\label{prop:kernels:relative-Feller}
	Let $\Omega$ be a relative system and $K_A$ be the Gibbs kernel for $A\Subset\GG$ associated to
	a relative absolutely summable interaction. The following are equivalent.
	\begin{enumerate}[label={\rm (\alph*)},ref={(\alph*)}]
		\item $\Omega$ has the topological Markov property
		relative to $\Theta$. \label{prop:kernels:relative-Feller_a}
		\item For every $A \Subset G$ and $p\in\Sigma^A$, the function $K_A(\cdot,[p])$
		is in $C_\Theta(\Omega)$. \label{prop:kernels:relative-Feller_b}
		\item For every $A \Subset G$ we have $K_A f\in C_\Theta(\Omega)$ whenever $f\in C_\Theta(\Omega)$. \label{prop:kernels:relative-Feller_c}
	\end{enumerate}
\end{proposition}
\begin{proof}\
	\begin{description}[font=\rmfamily\mdseries] %
	\item[\ref{prop:kernels:relative-Feller_a}~$\implies$~\ref{prop:kernels:relative-Feller_c}]
		By definition of the Gibbs kernel, we have
		\begin{align}
			\label{eq:kernels:relative-Feller}
			(K_A f)(\theta,x) &=
				\frac{1}{Z_{A|A^\complement}(\theta,x)}\sum_{u\in \Sigma^A}
					\ee^{-E_{A|A^\complement}(\theta,x_{A^\complement}\lor u)}
					\indicator{X_\theta}(x_{A^\complement}\lor u)
					f(\theta,x_{A^\complement}\lor u) \;.
		\end{align}
		The map $E_{A|A^\complement}$ is relatively continuous
		and $x\mapsto x_{A^\complement}\lor u$ is uniformly continuous. Let $B\supseteq A$ be a memory set for $A$ witnessing the relative TMP. %
		Then, for every $\theta\in\Theta$,
		we have $\indicator{X_\theta}(x_{A^\complement}\lor u)=\indicator{L(X_\theta)}(x_{B\setminus A}\lor u)$,
		from which we can see that the maps $x\mapsto\indicator{X_\theta}(x_{A^\complement}\lor u)$
		are equicontinuous for $\theta\in\Theta$.
		The partition function $Z_{A|A^\complement}(\theta,x)$ has the same form as the sum in~\eqref{eq:kernels:relative-Feller}
		with $f$ replaced with constant~$1$.
		Therefore, if $f\in C_\Theta(\Omega)$, the functions $x\mapsto(K_A f)(\theta,x)$
		are equicontinuous for $\theta\in\Theta$, which means $K_A f\in C_\Theta(\Omega)$.
	
	\item[\ref{prop:kernels:relative-Feller_c}~$\implies$~\ref{prop:kernels:relative-Feller_b}]
		The function $f \isdef \indicator{[p]}$ is in $C_\Theta(\Omega)$, therefore $K_A(\cdot,[p])$
		is in $C_\Theta(\Omega)$.
	
	\item[\ref{prop:kernels:relative-Feller_b}~$\implies$~\ref{prop:kernels:relative-Feller_a}]
		Let %
		$\varepsilon  \isdef \frac{1}{2}\inf_{(\theta,x) \in \Omega} K_A((\theta,x),[x_A])$.
		Since $E_{A|A^\complement}$ is bounded,
		we deduce that $\varepsilon > 0$. %
		By~\ref{prop:kernels:relative-Feller_b}
		we know that for every $p \in \Sigma^A$ we have $K_A(\cdot,[p]) \in C_\Theta(\Omega)$.
		Since $K_A((\theta,x), [p])$ depends only on $\theta$, $x_{A^{\complement}}$ and $p$,
		we can find $B \supseteq A$ such that for all $\theta \in \Theta$, $p \in \Sigma^A$ and $x,y \in X_\theta$,
		if $x_{B \setminus A}= y_{B \setminus A}$, we have
		\begin{align}
		\bigabs{K_A((\theta,x), [p]) - K_A((\theta,y), [p])} &< \varepsilon.
		\end{align}
		In particular, we obtain that if $x,y \in X_\theta$ and $x_{B \setminus A} = y_{B \setminus A  }$, then
		\begin{align}
		\bigabs{K_A((\theta,x), [x_A]) - K_A((\theta,y), [x_A])}
			&< \varepsilon \leq \frac{1}{2}K_A((\theta,x), [x_A])
		\end{align}
		and so $K_A((\theta,y), [x_A])\geq\varepsilon > 0$.  This shows that $x_A \lor y_{A^\complement} \in X_\theta$.
		As the choice of $B$ does not depend upon $\theta$ or $x,y \in X_{\theta}$ we deduce that $\Omega$ satisfies relative TMP.
		\qedhere
	\end{description}
\end{proof}

Let $\xspace{P}_\nu(\Omega)$ denote the space of probability measures
on $\Omega$ with marginal $\nu$ on $\Theta$.
The above proposition suggests topologizing $\xspace{P}_\nu(\Omega)$
by declaring the integration $\mu\mapsto\mu(f)$ continuous for each $f\in C_\Theta(\Omega)$.
The operator $\mu\mapsto\mu K_A$ would then become continuous %
whenever $\Omega$ has the TMP 
relative to $\Theta$.

Recall that $C_\Theta(\Omega)$ is a Banach space with the uniform norm.
When the environment space $\Theta$ is a standard Borel space,
one can identify $\xspace{P}_\nu(\Omega)$
with a closed subset of the dual space $C^*_\Theta(\Omega)$ (Proposition~\ref{prop:Riesz:relative}).
Alaoglu's theorem then implies that
the space $\xspace{P}_\nu(\Omega)$ is compact.
We will use the compactness of $\xspace{P}_\nu(\Omega)$
only at one point in the proof of Theorem~\ref{thm:DLR:relative}\ref{thm:DLR:LR:relative}
to argue that if $\nu$ is $\GG$-invariant and $\mu\in\xspace{P}_\nu(\Omega)$,
then the sequence of averages $\abs{F_n}^{-1}\sum_{g\in F_n}g^{-1}\mu$,
with $(F_n)_{n\in\NN}$ a F\o{}lner sequence, has a ($\GG$-invariant) cluster point.

At the more fundamental level,
the compactness of $\xspace{P}_\nu(\Omega)$ together with the relative Feller property of
the Gibbs kernels can be used to give
a direct proof of the existence of (invariant) relative Gibbs measures.

\begin{proposition}[Existence of invariant relative Gibbs measures]
\label{prop:existence:Gibbs}
    Assume that $\Theta$ is a standard Borel space and $\Omega$ satisfies TMP
    relative to~$\Theta$. 
    Let $\nu$ be a $\GG$-invariant probability measure on $\Theta$.  Then there
    exists a $\GG$-invariant relative Gibbs measure with marginal $\nu$.
\end{proposition}
\begin{proof}
	Since $\Theta$ is a standard Borel space, $\xspace{P}_\nu(\Omega)$ is compact.
	Since $\Omega$ satisfies relative TMP, the Gibbs kernels have the relative Feller property.
	Let $A_1 \subseteq A_2 \subseteq \ldots, $ be a nested sequence of finite
	subsets that exhaust  $\GG$. Let $\mu_0 \in \xspace{P}_\nu(\Omega)$.
	For $n\geq 1$, set $\mu_n \isdef \mu_0 K_{A_n}$.  Then for all $B \subseteq A_n$,
	$\mu_nK_B =  \mu_0 K_{A_n}K_B =   \mu_0 K_{A_n} = \mu_n$.
	So, using the topology on $P_\nu(\Omega)$ and the relative Feller property of the Gibbs kernels,
	any accumulation point $\mu$ of the sequence
	$\mu_n$ is a relative Gibbs measure.  It follows that for any
	$g \in \GG$, $g^{-1}\mu$ is also a relative  Gibbs measure.
	Thus, each $(1/\abs{A})\sum_{g \in A} g^{-1}\mu$ is a relative Gibbs measure.
	For any F\o{}lner sequence $F_n$, any accumulation point of
	\begin{align}
		(1/\abs{F_n})\sum_{g \in F_n} g^{-1}\mu
	\end{align}
	is a $\GG$-invariant relative Gibbs measure with marginal $\nu$.
	The existence of such accumulation points is guaranteed by the compactness of $P_\nu(\Omega)$.
\end{proof}

An example of a subshift on which Gibbs measures (invariant or not) do not exist
is the \emph{sunny-side up} shift $X\subseteq\{\symb{0},\symb{1}\}^\ZZ$,
which is defined as the set of all configurations with at most one occurrence of symbol~$\symb{1}$.

The following crude notion of closeness
between measures will be sufficient for our purposes.

\begin{proposition}[Closeness of measures]
\label{prop:measures:topology}
	Let $\nu$ be a probability measure on $\Theta$ and let $f\in C_\Theta(\Omega)$.
	For every $\varepsilon>0$,
	there exists $B\Subset\GG$ such that
	$\abs{\mu'(f)-\mu(f)}<\varepsilon$ whenever $\mu,\mu'\in\xspace{P}_\nu(\Omega)$
	satisfy $\mu'|_{\field{F}_\Theta\lor\xi^B}=\mu|_{\field{F}_\Theta\lor\xi^B}$
	\textup{(}i.e., $\mu$ and $\mu'$ have the same marginal on $(\theta,x_B)$\textup{)}.
\end{proposition}
\begin{proof}
	Since $f\in C_\Theta(\Omega)$,
	the family $(f(\theta,\cdot):\theta\in\Theta)$ is equicontinuous.
	Let $B\Subset\GG$ be such that
	$\abs{f(\theta,x')-f(\theta,x)}<\varepsilon$ whenever
	$x'_B=x_B$.
	We have
	\begin{align}
		\abs{\mu'(f)-\mu(f)} &=
			\abs{\int \mu'\big(f\,\big|\,\field{F}_\Theta\lor\xi^B\big)\dd\mu' -
			\int \mu\big(f\,\big|\,\field{F}_\Theta\lor\xi^B\big)\dd\mu} \\
		&=
			\abs{\int\Big[
				\mu'\big(f\,\big|\,\field{F}_\Theta\lor\xi^B\big) -
				\mu\big(f\,\big|\,\field{F}_\Theta\lor\xi^B\big)
			\Big]\dd\mu} \\
		&< \varepsilon \;. \\
		& \qedhere
	\end{align}
\end{proof}

\section{Proof of the main theorem}
\label{sec:proof}

\subsection{Relative Gibbs measures are relative equilibrium}
\label{sec:proof:relative-D}

\begin{proof}[Proof of Theorem~\ref{thm:DLR:relative}\ref{thm:DLR:D:relative}]
	Let $\mu$ be a $\GG$-invariant measure on $\Omega$ that projects to $\nu$
	and is relative Gibbs for $\Phi$.
	Let $\mu'$ be another $\GG$-invariant measure that projects to $\nu$.
	We show that $\psi(\mu')\leq\psi(\mu)$.
	
	Let $K=(K_A:A\Subset\GG)$ be the relative Gibbs specification associated to $\Phi$.
	Let $(F_n)_{n\in\NN}$ be a F\o{}lner sequence in $\GG$ with respect to which
	$\Omega$ is D-mixing relative to $\nu$,
	and	denote by $F^\theta_n$ the mixing set corresponding to $F_n$ and $\theta$
	that witnesses the D-mixing condition. %
	Fix $n\in\NN$.
	For $\nu$-almost every $\theta \in \Theta$ and every $x \in X_\theta$,
	let $\mu''_{\theta,x}$ %
	be a measure on $\Omega$ %
	which has the same distribution as
	$\mu'(\cdot | \field{F}_\Theta)(\theta,x)$ on $\xi^{F_n}$ and is supported on
	$\{(\theta, y) \in \Omega : y_{(F^\theta_n)^\complement} = x_{(F^\theta_n)^\complement} \}$.
	We choose $\mu''_{\theta,x}$ in such a way that for every cylinder $[u]$, the
	value $\mu''_{\theta,x}([u])$ is measurable as a function of $(\theta,x)$.

	Observe that
	\begin{align}
		H_{\mu'(\cdot | \field{F}_\Theta)(\theta,x)}\big(\xi^{F_n}\big)
		- \mu'(E_{F_n} \,|\, \field{F}_\Theta)(\theta,x)
		&=
		H_{\mu''_{\theta,x}}(\xi^{F_n})
		- \mu''_{\theta,x}(E_{F_n}) \;,
	\end{align}
	and hence
	\begin{align}
		\label{eq:proof:relative-D:same-distribution}
		\smash{\overbrace{%
			H_{\mu'}\big(\xi^{F_n} \,\big|\, \field{F}_\Theta\big) - \mu'(E_{F_n})
		}^{\Psi_{\mu'}(F_n)}}
		&= \int
		\big[%
		\smash{\overbrace{%
		H_{\mu''_{\theta,x}}(\xi^{F_n})- \mu''_{\theta,x}(E_{F_n})
		}^{\mbox{$\circled{\circ}$}}}%
		\big] \dd\mu(\theta,x) \;,
	\end{align}
	where we have used the fact that $\mu$ and $\mu'$ have the same marginals on $\Theta$
	and that $\circled{\circ}$ does not depend on $x$.
	On the other hand, by the finitary variational principle
	(Prop.~\ref{prop:gibbs-inequality}, or Cor.~\ref{cor:gibbs:local-optimality}), we have
	\begin{align}
		\label{eq:proof:relative-D:variational-principle}
		\underbrace{%
		H_{\mu''_{\theta,x}}(\xi^{F_n^\theta})
		- \mu''_{\theta,x}(E_{F_n^\theta|(F_n^\theta)^\complement})
		}_{\mbox{$\circled{{\bullet}}$}}
		&\leq
		\underbrace{%
		H_{K_{F_n^\theta}((\theta,x),\cdot)}\big(\xi^{F^\theta_n}\big)
		- K_{F_n^\theta}\big((\theta,x), E_{F_n^\theta | (F^\theta_n)^\complement}\big)
		}_{\mbox{$\circled{\star}$}} \;.
	\end{align}
	Here, we are applying this variational principle
	to the set $M\isdef\{u\in\Sigma^{F^\theta_n}: x_{(F^\theta_n)^\complement}\lor u\in X_\theta\}$,
	the energy function
	$U(u)\isdef E_{F^\theta_n|(F^\theta_n)^\complement}(\theta,x_{(F^\theta_n)^\complement}\lor u)$, and
	the distribution $p(u)\isdef\mu''_{\theta,x}([u])$.
	
	Combining~\eqref{eq:Hamiltonian:residue:Folner} and~\eqref{eq:Hamiltonian:different-sets:residue},
	we have
	\begin{align}
		\label{eq:Hamiltonian:residue:mixed}
		\norm{E_{F_n^\theta|(F_n^\theta)^\complement}-E_{F_n}} &\leq
			\abs{F_n^\theta\setminus F_n}\norm{\Phi} + \smallo(\abs{F_n})
	\end{align}
	as $n\to\infty$, with the $\smallo(\abs{F_n})$ term not depending on $(\theta,x)$.
	Therefore,
	\begin{align}
		\abs{\circled{\bullet}-\circled{\circ}} &\leq
			\abs{F_n^\theta\setminus F_n}(\log\abs{\Sigma} + \norm{\Phi}) + \smallo(\abs{F_n}) \;.
	\end{align}
	Integrating with respect to $\mu$ and using the relative D-mixing condition, we get
	\begin{align}
		\abs{\int\circled{\bullet}\,\dd\mu - \int\circled{\circ}\,\dd\mu} &\leq
			(\log\abs{\Sigma} + \norm{\Phi})\int\abs{F_n^\theta\setminus F_n}\dd\mu(\theta,x) + \smallo(\abs{F_n})\\
		\label{eq:proof:relative-D:error-A}
		&=
			\smallo(\abs{F_n})
	\end{align}
	as $n\to\infty$.
	
	For $\circled{\star}$, on the other hand, we have
	\begin{align}
		\int\circled{\star}\,\dd\mu
		&= \label{eq:proof:relative-D:error-F}
			\int\mu(\circled{\star}\,|\,\field{F}_\Theta)(\theta,x)\,\dd\mu(\theta,x) \\
		\label{eq:proof:relative-D:acrobatics} &=
			\int
			\Big[
			H_{\mu(\cdot|\field{F}_\Theta)(\theta,x)}\big(\xi^{F^\theta_n}\,\big|\, \xi^{(F^\theta_n)^\complement}\big)
			- \mu\big(E_{F_n^\theta|(F^\theta_n)^\complement}\,\big|\,\field{F}_\Theta\big)(\theta,x)
			\Big]
			\dd\mu(\theta,x) \\
		&\leq
			\begin{multlined}[t]
			\int
			\Big[
			H_{\mu(\cdot|\field{F}_\Theta)}\big(\xi^{F_n}\,\big|\, \xi^{F_n^\complement}\big)
			- \mu\big(E_{F_n}\,\big|\,\field{F}_\Theta\big)
			\Big]
			\dd\mu \\
			+ (\log\abs{\Sigma} + \norm{\Phi})\int\abs{F_n^\theta\setminus F_n}\dd\mu(\theta,x) + \smallo(\abs{F_n})
			\end{multlined} \\
		\label{eq:proof:relative-D:error-C} &=
			H_\mu\big(\xi^{F_n}\,\big|\, \xi^{F_n^\complement} \lor \field{F}_\Theta\big)
			- \mu(E_{F_n}) + \smallo(\abs{F_n}) \\
		\label{eq:proof:relative-D:error-B} &\leq
			{\underbrace{
			H_\mu\big(\xi^{F_n}\,\big|\, \field{F}_\Theta\big)
			- \mu(E_{F_n})
			}_{\Psi_\mu(F_n)}} + \smallo(\abs{F_n})
	\end{align}
	as $n\to\infty$, where we have again used~\eqref{eq:Hamiltonian:residue:mixed}.
	The equality between $\mu(\circled{\star}\,|\,\field{F}_\Theta)(\theta,x)$
	and the integrand on the right-hand side of~\eqref{eq:proof:relative-D:acrobatics}
	can be seen by partitioning $\Omega$ into countably many $\field{F}_\Theta$-measurable
	subsets over each of which $F_n^\theta$ is constant.
	
	Putting together~\eqref{eq:proof:relative-D:same-distribution},
	\eqref{eq:proof:relative-D:variational-principle},
	\eqref{eq:proof:relative-D:error-A}
	and~\eqref{eq:proof:relative-D:error-B},
	we obtain
	\begin{align}
		\Psi_{\mu'}(F_n) &\leq \Psi_\mu(F_n) + \smallo(\abs{F_n})
	\end{align}
	as $n\to\infty$.
	Dividing by $\abs{F_n}$ and letting $n\to\infty$
	yields $\psi(\mu')\leq\psi(\mu)$ as desired.
\end{proof}

\begin{remark}[Relative inner entropy for Gibbs measures]
\label{rem:inner-entropy}
	A closer look at the proof of Theorem~\ref{thm:DLR:relative}\ref{thm:DLR:D:relative},
	namely~\eqref{eq:proof:relative-D:error-C},
	shows that in fact
	\begin{align}
		\Psi_{\mu'}(F_n) &\leq
			H_\mu\big(\xi^{F_n}\,\big|\, \xi^{F_n^\complement} \lor \field{F}_\Theta\big)
			- \mu(E_{F_n}) + \smallo(\abs{F_n})
	\end{align}
	as $n\to\infty$.
	Choosing $\mu'=\mu$, we obtain
	\begin{align}
		\psi(\mu) &=
			\lim_{n\to\infty}
			\frac{
				H_\mu\big(\xi^{F_n}\,\big|\, \xi^{F_n^\complement} \lor \field{F}_\Theta\big)
			}{\abs{F_n}} - \mu(f_\Phi) \;.
	\end{align}
	In particular, every $\GG$-invariant Gibbs measure relative to $\nu$
	satisfies
	\begin{align}
		h_\mu(\Omega\,|\,\Theta) &=
			\lim_{n\to\infty}
			\frac{
				H_\mu\big(\xi^{F_n}\,\big|\, \xi^{F_n^\complement} \lor \field{F}_\Theta\big)
			}{\abs{F_n}} \;,
	\end{align}
	as long as $\Omega$ is D-mixing relative to $\nu$.
	The corresponding equality in the non-relative setting
	is observed by F\"ollmer and Snell~\cite{FolSne77} and Tempelman~\cite[Sec.~\S5.3]{Tem84}.
	\hfill\remarkqed
\end{remark}

When both relative D-mixing and relative TMP are satisfied,
we can obtain an explicit expression for the maximum pressure
in terms of partition functions,
generalizing the similar expression in the non-relative setting
(see e.g.~\cite[Thm.~3.12]{Rue04}).
Recall the definition of the partition function $Z_{A|A^\complement}(\theta,x)$
for environment $\theta$ and boundary condition $x$ in~\eqref{eq:boltzman-dist}.
Given $A\Subset\GG$ and $\theta\in\Theta$, we may also define the
partition function with \emph{free} boundary condition as
\begin{align}
	Z_A(\theta) &\isdef
		\sum_{u\in L_A(X_\theta)} \ee^{-E_A(\theta,u)} \;,
\end{align}
where $E_A(\theta,u)$ is understood as $E_A(\theta,x)$ for
any $x\in [u]\cap X_\theta$.

\begin{proposition}[Variational principle]
\label{prop:variational-principle}
	Let $\Omega$ be a relative system and $\nu$ a $\GG$-invariant probability measure
	on its environment space $\Theta$.
	Let $\Phi$ be an absolutely summable relative interaction on $\Omega$
	and $f_\Phi$ its associated energy observable.
	Assume that $\Theta$ is a standard Borel space.
	Assume further that $\Omega$ satisfies TMP and D-mixing relative to $\nu$.
	Then,
	\begin{align}
	\label{eq:variational-principle}
		\sup_{\mu\in\xspace{P}_\nu(\Omega)}\big[h_\mu(\Omega\,|\,\Theta)-\mu(f_\Phi)\big]
		&=
			\lim_{n\to\infty}\int\frac{\log Z_{F_n}(\theta)}{\abs{F_n}}\dd\nu(\theta) \;,
	\end{align}
	where $(F_n)_{n\in\NN}$ is a F\o{}lner sequence with respect to which
	the D-mixing condition holds.
	Moreover, every relative $\GG$-invariant Gibbs measure for $\Phi$ with marginal $\nu$
	achieves the supremum in the left hand side of~(\ref{eq:variational-principle}).
\end{proposition}
\begin{proof}
	Let $\mu$ be a $\GG$-invariant relative Gibbs measure for $\Phi$ with marginal $\nu$.
	The existence of relative Gibbs measures is guaranteed by Proposition~\ref{prop:existence:Gibbs}.
	By the relative Dobrushin theorem, $\mu$ achieves the supremum on the left-hand side
	of~\eqref{eq:variational-principle}.
	It remains to show that the pressure of $\mu$ coincides with
	the right-hand side of~\eqref{eq:variational-principle}.
	
	The observation made in Remark~\ref{rem:inner-entropy}
	together with~\eqref{eq:Hamiltonian:residue:Folner} gives the expression
	\begin{align}
		\psi(\mu) &=
			\lim_{n\to\infty}
			\frac{
				\Psi_\mu(F_n\,|\,F_n^\complement)
			}{\abs{F_n}}
	\end{align}
	for the relative pressure of $\mu$.
	Let $K=(K_A:A\Subset\GG)$ be the relative Gibbs specification associated to~$\Phi$.
	Since $\mu$ is a relative Gibbs measure, we have
	\begin{align}
		\Psi_\mu(F_n\,|\,F_n^{\complement}) &=
			\int\Big[
				H_{K_{F_n}((\theta,x),\cdot)}(\xi^{F_n})
					- K_{F_n}\big((\theta,x),E_{F_n|F_n^\complement}\big)
			\Big]
			\dd\mu(\theta,x) \\
		&=
			\int \log Z_{F_n|F_n^\complement}(\theta,x) \;\dd\mu(\theta,x) \;.
	\end{align}
	Thus, we only need to show that
	\begin{align}
	\label{eq:variational-principle:proof:approx1}
		\biggabs{%
			\int \log Z_{F_n|F_n^\complement}(\theta,x)\; \dd\mu(\theta,x) -
			\int \log Z_{F_n}(\theta)\; \dd\mu(\theta,x)
		}
		&=
			\smallo(\abs{F_n})
	\end{align}
	as $n\to\infty$.
	
	Let $F_n^\theta$ be the mixing set for $F_n$ witnessing the D-mixing condition relative to $\nu$.
	In order to prove~\eqref{eq:variational-principle:proof:approx1},
	it is enough to show that
	\begin{align}
		\label{eq:variational-principle:proof:ineqA}
		\log Z_{F_n|F_n^\complement}(\theta,x) &\leq \log Z_{F_n}(\theta) + \smallo(\abs{F_n}) \;, \\
		\label{eq:variational-principle:proof:ineqB}
		\log Z_{F_n}(\theta) &\leq
			\log Z_{F_n^\theta|(F_n^\theta)^\complement}(\theta,x) +
			\norm{\Phi}\abs{F_n^\theta\setminus F_n} \;, \\
		\label{eq:variational-principle:proof:ineqC}
		\int \log Z_{F_n^\theta|(F_n^\theta)^\complement}(\theta,x)\; \dd\mu(\theta,x) &\leq
			\int \log Z_{F_n|F_n^\complement}(\theta,x)\; \dd\mu(\theta,x) + \smallo(\abs{F_n})
	\end{align}
	and use the fact that $\int\abs{F_n^\theta\setminus F_n}\dd\mu(\theta,x)=\smallo(\abs{F_n})$
	by the D-mixing condition.
	
	Inequalities~\eqref{eq:variational-principle:proof:ineqA}
	and~\eqref{eq:variational-principle:proof:ineqB} can be verified
	by a straightforward calculation using the fact that $F_n^\theta$
	is a mixing set for $F_n$ (see Sec.~\S\ref{sec:variational-principle:proof:ineqAnB}).
	Inequality~\eqref{eq:variational-principle:proof:ineqC}
	follows from the fact that the left hand side of~\eqref{eq:proof:relative-D:error-F} is lesser or equal to the right hand side of~\eqref{eq:proof:relative-D:error-C}
	once we recall that $\log Z_{F_n^\theta|(F_n^\theta)^\complement}(\theta,x)$
	is the same as~$\circled{\star}$,
	and that the integral of $\log Z_{F_n|F_n^\complement}(\theta,x)$ is nothing but
	$\Psi_\mu(F_n\,|\,F_n^\complement)=
	H_\mu\big(\xi^{F_n}\,\big|\, \field{F}_\Theta \vee \xi^{F_n^{\complement}}  \big)	- \mu(E_{F_n|F_n^\complement})$,
	which differs from the the right-hand side of~\eqref{eq:proof:relative-D:error-C}
	by no more than $\smallo(\abs{F_n})$.
\end{proof}

\subsection{Relative equilibrium measures are relative Gibbs}
\label{sec:proof:relative-LR}

The idea of the proof of Theorem~\ref{thm:DLR:relative}\ref{thm:DLR:LR:relative}
is as follows:
if a measure $\mu$ on $\Omega$ is not relative Gibbs,
then the conditional relative pressure $\Psi_\mu(A\,|\,A^\complement)$
has to be sub-optimal for some $A\Subset\GG$ (Corollary~\ref{cor:gibbs:local-optimality}).
Therefore, applying the Gibbs kernel $K_A$ on $\mu$ would locally increase
the pressure.  In order to increase the relative pressure per site $\psi$,
we apply the Gibbs kernels on a positive-density set of translations of $A$, one after another.
The translations of $A$ should be sufficiently far apart so that
the applications of the different kernels do not significantly interfere with one another.
The final step is to do the standard averaging procedure to make the new measure $\GG$-invariant.

This strategy for proving a result of this type is not entirely new. The fundamental idea
of making a local improvement in a positive density set in order to achieve a global gain has
been used many times in the literature.
This idea is explicit in the works of F\"ollmer~\cite{Fol73} and Burton and Steif~\cite{BurSte94}
(see also~\cite[Sec.~\S15.4]{Geo88} and the bibliographic notes therein).
Similar ideas have appeared in other contexts, for instance in the proof of
the Garden-of-Eden theorem~\cite{Moo62,Myh63} (see~\cite[Chap.~5]{CecCoo10}).

To follow the above strategy, we need three lemmas.
The first provides a sufficient condition for the uniform convergence of
a certain type of martingale.
The second lemma complements Corollary~\ref{cor:gibbs:local-optimality}
by stating that the improvement achieved by applying a Gibbs kernel
is truly local.  The last lemma ensures the existence of
a non-overlapping packing of copies of a given finite set
with strictly positive uniform lower density.
Without loss of generality, by removing a $\nu$-null set from $\Theta$ if necessary,
we will assume that $\Omega$ has the TMP relative to the entire $\Theta$.

Let $f\colon\Sigma^\GG\to\RR$ be a bounded
measurable function %
and $\mu$ a probability measure on~$\Sigma^\GG$.
According to the martingale convergence theorem,
the conditional expectations $\mu(f\,|\,\xi^B)$ converge
$\mu$-almost surely to $f$ as $B$ grows to $\GG$ along any co-final sequence
of finite subsets of $\GG$.  Marcus and Pavlov~\cite{MarPav15} observed that
if $f$ has a continuous version modulo $\mu$
(i.e., $f=g$ $\mu$-almost surely for a continuous map $g\colon\Sigma\to\RR$),
then the convergence of $\mu(f\,|\,\xi^B)$ is uniform over
a set of full measure and holds in the net sense,
along the family of finite subsets of $\GG$ directed by inclusion.
The following lemma is a relative version of the Marcus--Pavlov lemma.
\begin{lemma}[Relative uniform martingale convergence]
\label{lem:martingale-convergence:uniform:relative}
	Let $f\in C_\Theta(\Omega)$ and let $\nu$ be a probability measure on $\Theta$.
	Then, for every probability measure $\mu\in\xspace{P}_\nu(\Omega)$,
	there is a set of full measure on which $\mu(f\,|\,\xi^B\lor\field{F}_\Theta)$
	converges uniformly to $f$ as $B\nearrow\GG$ along the family of finite
	subsets of $\GG$ directed by inclusion.
	Furthermore, the convergence is also uniform over the choice of $\mu$.
\end{lemma}
\begin{proof}
	Let $\varepsilon>0$.
	Choose a finite set $B_0\Subset\GG$ large enough so that
	$\abs{f(\theta,x)-f(\theta,y)}<\varepsilon$ whenever
	$x_{B_0}=y_{B_0}$.
	For every $B\Subset\GG$ we have
	\begin{align}
	\label{eq:martingale-convergence:uniform:relative:conditioning}
		\mu(f\,|\,\xi^B\lor\field{F}_\Theta)(\theta,x) &=
			\frac{1}{\mu([x_B]\,|\,\field{F}_\Theta)(\theta,x)}
			\int_{[x_B]}f\,\dd\mu(\cdot\,|\,\field{F}_\Theta)(\theta,x)
	\end{align}
	for $\mu$-almost every $(\theta,x)\in\Omega$
	(see Sec.~\S\ref{sec:martingale-convergence:uniform:relative:conditioning:argument}).
	It follows that when $B\supseteq B_0$,
	\begin{align}
		\abs{\mu(f\,|\,\xi^B\lor\field{F}_\Theta)(\theta,x)-f(\theta,x)} &< \varepsilon
	\end{align}
	for $\mu$-almost every $(\theta,x)\in\Omega$.
	This shows the uniform convergence.
	Observe that $B_0$ does not depend on $\mu$.
	Hence the convergence is also uniform in $\mu$.
\end{proof}

Along with Corollary~\ref{cor:gibbs:local-optimality}, the next lemma
constitutes the main ingredient for proving Theorem~\ref{thm:DLR:relative}\ref{thm:DLR:LR:relative}.
It allows to see the improvement predicted by Corollary~\ref{cor:gibbs:local-optimality}
at the level of finite sets.

\begin{lemma}[Local enhancement]
\label{lem:local-enhancement}
	Suppose that $\Psi_\mu(A\,|\,A^\complement)<\Psi_{\mu K_A}(A\,|\,A^\complement)$
	for some $A\Subset\GG$.
	Then, there exists an $\varepsilon>0$ and a finite set $B_0\supseteq A$
	such that
	\begin{align}
		\Psi_\mu\big(A\,\big|\,(B\setminus A)\big) &\leq \Psi_{\mu'K_A}\big(A\,\big|\,(B\setminus A)\big) - \varepsilon
	\end{align}
	for every measure $\mu'$
	with $\mu'|_{\field{F}_\Theta\lor\xi^{B_0}}=\mu|_{\field{F}_\Theta\lor\xi^{B_0}}$
	and every finite set $B\supseteq B_0$.
\end{lemma}
\begin{proof}
	Let $\delta\isdef\Psi_{\mu K_A}(A\,|\,A^\complement)-\Psi_\mu(A\,|\,A^\complement)$
	and set $\varepsilon\isdef\nicefrac{\delta}{7}$.
	We make six separate approximations, and choose $B_0\supseteq A$ large enough
	so that
	the error in each approximation is less than $\nicefrac{\delta}{7}$.
	
	Recall that the convergence $E_{A|(B\setminus A)}\to E_{A|A^\complement}$
	is uniform.  Therefore, if we choose $B_0$ large enough, we can make sure that
	\begin{align}
		\abs{E_{A|(B\setminus A)}(\theta,x)-E_{A|A^\complement}(\theta,x)} &< \nicefrac{\delta}{7}
	\end{align}
	for every $(\theta,x)\in\Omega$, whenever $B\supseteq B_0$.
	With such choice of $B_0$, we have
	\begin{align}
		\label{eq:enhancement-lemma:energy:1}
		\abs{\mu(E_{A|(B\setminus A)})-\mu(E_{A|A^\complement})} &< \nicefrac{\delta}{7} \;, \\
		\label{eq:enhancement-lemma:energy:2}
		\abs{(\mu'K_A)(E_{A|(B\setminus A)})-(\mu'K_A)(E_{A|A^\complement})} &< \nicefrac{\delta}{7}
	\end{align}
	whenever $B\supseteq B_0$.
	Since $\Omega$ has the TMP and $E_{A|A^\complement}$ is in $C_\Theta(\Omega)$, Proposition~\ref{prop:kernels:relative-Feller} implies that the function $K_A E_{A|A^\complement}$
	is in $C_\Theta(\Omega)$.
	Therefore, if we choose $B_0$ large enough, then we have
	\begin{align}
		\label{eq:enhancement-lemma:energy:3}
		\abs{(\mu'K_A)(E_{A|A^\complement})-(\mu K_A)(E_{A|A^\complement})} &< \nicefrac{\delta}{7}
	\end{align}
	whenever $\mu'$ has the same marginal on $(\theta,x_{B_0})$ as $\mu$
	(Proposition~\ref{prop:measures:topology}).
	Combining~\eqref{eq:enhancement-lemma:energy:2}
	and~\eqref{eq:enhancement-lemma:energy:3}, for sufficiently large $B_0\supseteq A$
	we get
	\begin{align}
		\label{eq:enhancement-lemma:energy:2-3}
		\abs{(\mu'K_A)(E_{A|(B\setminus A)}) - (\mu K_A)(E_{A|A^\complement})}
			&< (\nicefrac{2}{7})\delta
	\end{align}
	whenever $B\supseteq B_0$ and %
	$\mu'$ has the same marginal on $(\theta,x_{B_0})$ as $\mu$.

	Using the martingale convergence theorem and the monotonicity of conditional entropy
	with respect to the condition, we know that
	\begin{align}
		H_\mu(\xi^A\,|\,\xi^{B\setminus A}\lor\field{F}_\Theta) &\to
			H_\mu(\xi^A\,|\,\xi^{A^\complement}\lor\field{F}_\Theta)
	\end{align}
	as $B\nearrow\GG$ along the finite subsets of $\GG$ directed by inclusion.
	Therefore, choosing $B_0\supseteq A$ large enough, we get
	\begin{align}
		\label{eq:enhancement-lemma:entropy:1}
		\abs{H_\mu(\xi^A\,|\,\xi^{B\setminus A}\lor\field{F}_\Theta) -
			H_\mu(\xi^A\,|\,\xi^{A^\complement}\lor\field{F}_\Theta)}
			&< \nicefrac{\delta}{7}
	\end{align}
	whenever $B\supseteq B_0$.
	Note that
	\begin{align}
		H_{\mu'K_A}(\xi^A\,|\,\xi^{A^\complement}\lor\field{F}_\Theta) &=
			-\int
				\smash{\overbrace{
				\sum_{p\in\Sigma^A}\indicator{[p]}\cdot
				\log K_A(\cdot,[p])
				}^{\gamma(\cdot)}}
			\,\dd(\mu'K_A)
		=
			\mu'K_A \gamma.
	\end{align}
	Since $\Omega$ has TMP relative to $\Theta$, Proposition~\ref{prop:kernels:relative-Feller} implies that the integrand $\gamma$
	and as a result $K_A \gamma$ are in $C_\Theta(\Omega)$.
	Therefore, if we choose $B_0$ large enough, we can make sure,
	using Proposition~\ref{prop:measures:topology}, that
	\begin{align}
		\label{eq:enhancement-lemma:entropy:2}
		\abs{H_{\mu' K_A}(\xi^A\,|\,\xi^{A^\complement}\lor\field{F}_\Theta) -
			H_{\mu K_A}(\xi^A\,|\,\xi^{A^\complement}\lor\field{F}_\Theta)}
			&< \nicefrac{\delta}{7}
	\end{align}
	whenever $\mu'$ has the same marginal on $(\theta,x_{B_0})$ as $\mu$.
	Lastly, by the martingale convergence theorem, we know that
	for $(\mu'K_A)$-almost every $(\theta,x)\in\Omega$ and
	every $p\in\Sigma^A$,
	\begin{align}
		\underbrace{(\mu'K_A)\big([p]\,\big|\,\xi^{B\setminus A}\lor\field{F}_\Theta\big)(\theta,x)}_{%
			\mu'\big(K_A(\cdot,[p])\,\big|\,\xi^{B\setminus A}\lor\field{F}_\Theta\big)(\theta,x)
		} &\to
			\underbrace{(\mu'K_A)\big([p]\,\big|\,\xi^{A^\complement}\lor\field{F}_\Theta\big)(\theta,x)}_{%
				K_A((\theta,x),[p])
			}
	\end{align}
	as $B$ grows to $\GG$ along any co-final sequence of finite subsets of $\GG$.
	Since the limit has a version $K_A\big((\theta,x),[p]\big)$
	which is in $C_\Theta(\Omega)$ (Proposition~\ref{prop:kernels:relative-Feller}),
	Lemma~\ref{lem:martingale-convergence:uniform:relative} ensures
	that the convergence is uniform both in $(\theta,x)$ (on a set of full $\mu'$-measure)
	and in $\mu'$. %
	Since $E_{A|A^\complement}$ is bounded, for each $p\in\Sigma^A$
	the function $K_A(\cdot,[p])$ is bounded away from~$0$ on $[p]$.
	It follows that the convergence of
	\begin{align}
		H_{\mu'K_A}(\xi^A\,|\,\xi^{B\setminus A}\lor\field{F}_\Theta) &=
			-\int\sum_{p\in\Sigma^A}\indicator{[p]}\cdot
			\log\Big[(\mu'K_A)\big([p]\,\big|\,\xi^{B\setminus A}\lor\field{F}_\Theta\big)\Big]
			\,\dd(\mu'K_A)
	\shortintertext{to}
		H_{\mu'K_A}(\xi^A\,|\,\xi^{A^\complement}\lor\field{F}_\Theta) &=
			-\int\sum_{p\in\Sigma^A}\indicator{[p]}\cdot
			\log K_A(\cdot,[p])
			\,\dd(\mu'K_A)
	\end{align}
	is uniform among all $\mu'\in\xspace{P}_\nu(\Omega)$.
	In particular, choosing $B_0\supseteq A$ large enough, we can ensure that
	\begin{align}
		\label{eq:enhancement-lemma:entropy:3}
		\abs{H_{\mu'K_A}(\xi^A\,|\,\xi^{B\setminus A}\lor\field{F}_\Theta) -
			H_{\mu'K_A}(\xi^A\,|\,\xi^{A^\complement}\lor\field{F}_\Theta)}
			&< \nicefrac{\delta}{7}
	\end{align}
	for every $\mu'\in\xspace{P}_\nu(\Omega)$, whenever $B\supseteq B_0$.
	Combining~\eqref{eq:enhancement-lemma:entropy:2}
	and~\eqref{eq:enhancement-lemma:entropy:3}, for sufficiently large $B_0\supseteq A$
	we get
	\begin{align}
		\label{eq:enhancement-lemma:entropy:2-3}
		\abs{H_{\mu'K_A}(\xi^A\,|\,\xi^{B\setminus A}\lor\field{F}_\Theta) -
			H_{\mu K_A}(\xi^A\,|\,\xi^{A^\complement}\lor\field{F}_\Theta)}
			&< (\nicefrac{2}{7})\delta
	\end{align}
	whenever $B\supseteq B_0$ and %
	$\mu'$ has the same marginal on $(\theta,x_{B_0})$ as $\mu$.
	
	Putting~\eqref{eq:enhancement-lemma:energy:1}, \eqref{eq:enhancement-lemma:energy:2-3},
	\eqref{eq:enhancement-lemma:entropy:1} and~\eqref{eq:enhancement-lemma:entropy:2-3}
	together with the hypothesis
	$\delta=\Psi_{\mu K_A}(A\,|\,A^\complement)-\Psi_\mu(A\,|\,A^\complement)>0$,
	the result follows.
\end{proof}

In the course of the proof, we will need to pack copies of a finite set $P\Subset\GG$
on $\GG$ in a non-overlapping fashion in such a way that the uniform density of the copies
is strictly positive.  On a hyper-cubic lattice $\GG=\ZZ^d$, a periodic packing
does the job.  On a general countable amenable group,
a positive-density non-overlapping packing is achieved by a Delone set.
\begin{lemma}[Existence of Delone sets]
\label{lem:delone}
	Let $\GG$ be a group, and $P,C\subseteq\GG$ subsets satisfying $C\supseteq PP^{-1}$.
	Then, there exists a set $D\subseteq\GG$ satisfying the following two conditions:
	\begin{enumerate}[label={\rm (\roman*)}]
		\item \textup{(Packing)} $d P\cap d'P=\varnothing$ for every two distinct elements $d,d'\in D$,
		\item \textup{(Covering)} $gC\cap D\neq\varnothing$ for every $g\in\GG$.
	\end{enumerate}
\end{lemma}
\begin{proof}
	Let $\field{D}$ denote the family of all subsets of $\GG$ that satisfy
	the packing condition.
	This family is partially ordered by inclusion.
	Furthermore, every chain in $\field{D}$ has an upper bound in $\field{D}$, namely the union
	of its elements.  By Zorn's lemma, $\field{D}$ has a maximal element, which we call $D$.
	We claim that $D$ also satisfies the covering condition.
	For if $D$ does not satisfy the covering condition, there must exist an element $g\in\GG$ such that
	$gPP^{-1}\cap D=\varnothing$, or equivalently, $gP\cap d P=\varnothing$ for every $d\in D$.
	It follows that $\{g\}\cup D$ is in $\field{D}$, contradicting the maximality of $D$.
\end{proof}

We leave it to the reader to show that when $\GG$ is countable and $P$ and $C$ are finite,
the existence of Delone sets can be established without resorting to the axiom of choice.
Let us remark that in the case that $\GG$ is a countable amenable group and $P$ and $C$ are finite,
the covering condition in the above lemma ensures that $D$ has
positive \emph{uniform lower density}
\begin{align}
	\density(D) &\isdef
		\liminf_{n\to\infty}\inf_{g\in\GG} \frac{\abs{gD\cap F_n}}{\abs{F_n}}
\end{align}
with respect to every F\o lner sequence $(F_n)_{n\in\NN}$.
Indeed, for every $g\in\GG$ and $h\in F_n$ there exists at least one $c\in C$
such that $hc\in gD$.  It follows that $\abs{gD\cap F_nC}\abs{C}\geq\abs{F_n}$.
On the other hand, $\abs{gD\cap F_n}\geq\abs{gD\cap F_nC}-\abs{F_n\triangle F_nC}$.
Hence, $\density(D)\geq\abs{C}^{-1}$.
We are now ready to prove the relative Lanford--Ruelle theorem.

\begin{proof}[Proof of Theorem~\ref{thm:DLR:relative}\ref{thm:DLR:LR:relative}]
	Let $\mu$ be a $\GG$-invariant measure on $\Omega$ with marginal $\nu$
	and suppose that $\mu$ is not a relative Gibbs measure for $\Phi$.
	We show that $\mu$ is not an equilibrium measure for $f_\Phi$ relative to $\nu$
	by constructing another $\GG$-invariant measure $\overline{\mu}_+$ with marginal $\nu$
	that has strictly larger relative pressure per site.
	Let $K=(K_A: A\Subset\GG)$ be the relative Gibbs specification
	associated to $\Phi$, and $\Psi_\mu$ the relative pressure under $\mu$.
	
	Since $\mu$ is not relative Gibbs, there exists a set $A\Subset\GG$
	such that $\mu K_A\neq \mu$.  According to Corollary~\ref{cor:gibbs:local-optimality},
	this implies that
	\begin{align}
		\Psi_{\mu}(A\,|\,A^\complement) &< \Psi_{\mu K_A}(A\,|\,A^\complement) \;.
	\end{align}
	Let $\varepsilon>0$, and take $B\isdef B_0\supseteq A$ %
	as guaranteed by Lemma~\ref{lem:local-enhancement}.
	Thus,
	\begin{align}
		\label{eq:relative-LR:enhancement}
		\Psi_{\mu}\big(A\,|\,(B'\setminus A)\big) &\leq
			\Psi_{\mu' K_A}\big(A\,|\,(B'\setminus A)\big) - \varepsilon
	\end{align}
	whenever $B'\supseteq B$ and $\mu'$ has the same marginal on $(\theta,x_{B})$ as $\mu$.
	
	Let $D\subseteq\GG$ be a Delone set with packing shape $B$ and covering shape $BB^{-1}$
	(Lemma~\ref{lem:delone}).
	Let $k_1,k_2,\ldots$ be an arbitrary enumeration of the elements of $D$.
	Let $A_i\isdef k_i A$ and $B_i\isdef k_i B$ for $i=1,2,\ldots$.
	Define $\mu^{(0)}\isdef\mu$ and $\mu^{(i)}\isdef\mu^{(i-1)} K_{A_i}$ for $i\geq 1$.
	From the facts that the sets $A_i$ are disjoint
	and the kernels $K_{A_i}$ are proper %
	(i.e., $K_{A_i}$ keeps the marginal on $(\theta,x_{A_i^\complement})$ intact)
	it follows that the limit $\mu_+\isdef\lim_{i\to\infty}\mu^{(i)}$ exists.
	Note however that $\mu_+$ may depend on the enumeration of $D$,
	and more importantly, is not necessarily $\GG$-invariant.
	
	Let $(F_n)_{n\in\NN}$ be a fixed F\o lner sequence.
	We average over the $\GG$-orbit of $\mu_+$ to construct a $\GG$-invariant measure
	$\overline{\mu}_+$.
	More specifically, let $\overline{\mu}_+$ be an accumulation point of
	the sequence
	\begin{align}
		\overline{\mu}^{(m)}_+ &\isdef \frac{1}{\abs{F_m}}\sum_{g\in F_m}g^{-1}\mu_+
	\end{align}
	as $m\to\infty$.
	Any such accumulation point will be a $\GG$-invariant measure.
	The existence of accumulation points is guaranteed
	by the compactness of $\xspace{P}_\nu(\Omega)$,
	whose argument relies on $\Theta$ being a standard Borel space
	(see Appendix~\S\ref{apx:measures:topology}).
	
	To show that $\overline{\mu}_+$ has strictly larger pressure per site than $\mu$,
	we compare the pressure of $g^{-1}\mu_+$ and $g^{-1}\mu$ on $F_n$ for arbitrary $g\in\GG$
	and show that uniformly in $g$, there is a gap of at least
	$\varepsilon\,\density(D)\abs{F_n} + \smallo(\abs{F_n})$ between them, that is,
	\begin{align}
		\inf_{g\in\GG} \Psi_{g^{-1}\mu_+}(F_n) =
		\inf_{g\in\GG} \Psi_{\mu_+}(gF_n) &\geq
			\Psi_\mu(F_n) + \varepsilon\,\density(D)\abs{F_n} + \smallo(\abs{F_n})
	\end{align}
	as $n\to\infty$.
	By the concavity of the relative pressure, for each $m$, we have
	\begin{align}
		\Psi_{\overline{\mu}^{(m)}_+}(F_n) &\geq
			\inf_{g\in\GG} \Psi_{g^{-1}\mu_+}(F_n) \;.
	\end{align}
	Taking the limit as $m\to\infty$
	and using the continuity of the pressure gives
	\begin{align}
		\Psi_{\overline{\mu}_+}(F_n) &\geq
			\inf_{g\in\GG} \Psi_{g^{-1}\mu_+}(F_n) \geq
			\Psi_\mu(F_n) + \varepsilon\,\density(D)\abs{F_n} + \smallo(\abs{F_n})
	\end{align}
	as $n\to\infty$.
	Dividing by $\abs{F_n}$ and letting $n\to\infty$
	will then yield the result.
	
	For $g\in\GG$,
	let $D^g_n\isdef\{k\in D: kB\subseteq gF_n\}$ and
	$\widehat{D}^g_n\isdef\{k\in D: kB\cap gF_n\neq\varnothing\}$.
	Note that $\abs{D^g_n}\geq\density(D)\abs{F_n}+\smallo(\abs{F_n})$
	and $\bigabs{\widehat{D}^g_n\setminus D^g_n}=\smallo(\abs{F_n})$ as $n\to\infty$
	uniformly in $g$. %
	Let $k_{\ell_1},k_{\ell_2},\ldots,k_{\ell_m}$ be the elements of $D^g_n$
	ordered according to the previously fixed enumeration of $D$.
	Let $R^g_n\isdef \bigcup_{k\in(\widehat{D}^g_n\setminus D^g_n)}(kA\cap gF_n)$ be
	the union of $A$-neighborhoods of the elements of $D$ that intersect $gF_n$
	but are not entirely included in $gF_n$.
	Using the chain rule, we decompose $\Psi_\mu(gF_n)$ and $\Psi_{\mu_+}(gF_n)$ as follows:
	\begin{align}
		\label{eq:decomposition:mu}
		\Psi_\mu(gF_n) &=
			\Psi_\mu\Big(gF_n\setminus \big[R^g_n\cup\bigcup_{i=1}^m A_{\ell_i}\big]\Big) +
			\sum_{i=1}^m\Psi_\mu\Big(A_{\ell_i}\,\Big|\,gF_n\setminus\big[R^g_n\cup\bigcup_{j=i}^m A_{\ell_j}\big]\Big) \nonumber\\
		&\qquad +
			\Psi_\mu\big(R^g_n\,\big|\, gF_n\setminus R^g_n\big)\\
		\label{eq:decomposition:mu-plus}
		\Psi_{\mu_+}(gF_n) &=
			\Psi_{\mu_+}\Big(gF_n\setminus \big[R^g_n\cup\bigcup_{i=1}^m A_{\ell_i}\big]\Big) +
			\sum_{i=1}^m\Psi_{\mu_+}\big(A_{\ell_i}\,\Big|\,gF_n\setminus\big[R^g_n\cup\bigcup_{j=i}^m A_{\ell_j}\big]\big) \nonumber\\
		&\qquad +
			\Psi_{\mu_+}\big(R^g_n\,\big|\, gF_n\setminus R^g_n\big) \;.
	\end{align}
	Observe that the first terms on the right-hand sides of~\eqref{eq:decomposition:mu}
	and~\eqref{eq:decomposition:mu-plus} are identical,
	because the two measures $\mu$ and $\mu_+$ have the same marginals on
	$(\theta,x_{\GG\setminus\bigcup_{k\in D}kA})$, and in particular on
	$(\theta,\allowbreak x_{gF_n\setminus[R^g_n\cup\bigcup_{i=1}^m A_{\ell_i}]})$.
	On the other hand, the last terms in~\eqref{eq:decomposition:mu}
	and~\eqref{eq:decomposition:mu-plus} are each bounded by
	\begin{align}
		\big(\log\abs{\Sigma}+\norm{\Phi}\big)\abs{R^g_n} &\leq
			\big(\log\abs{\Sigma}+\norm{\Phi}\big)\abs{A}\bigabs{\widehat{D}^g_n\setminus D^g_n} \;,
	\end{align}
	which is $\smallo(\abs{F_n})$ as $n\to\infty$ uniformly in $g$.
	To compare the middle terms, observe that
	on $(\theta,\allowbreak x_{A_{\ell_i}\cup(gF_n\setminus[R^g_n\cup\bigcup_{j=i}^m A_{\ell_j}])})$,
	the measure $\mu_+$ has the same marginal as $\mu^{(\ell_i)}=\mu^{(\ell_i-1)}K_{A_{\ell_i}}$.
	Therefore,
	\begin{align}
		\Psi_{\mu_+}\Big(A_{\ell_i}\,\Big|\,gF_n\setminus\big[R^g_n\cup\bigcup_{j=i}^m A_{\ell_j}\big]\Big) &=
			\Psi_{\mu^{(\ell_i-1)}K_{A_{\ell_i}}}\Big(A_{\ell_i}\,\Big|\,gF_n\setminus\big[R^g_n\cup\bigcup_{j=i}^m A_{\ell_j}\big]\Big) \;.
	\end{align}
	Since $\mu^{(\ell_i-1)}$ and $\mu$ have the same marginals on $B_{\ell_i}$,
	from~\eqref{eq:relative-LR:enhancement} we get
	\begin{align}
		\Psi_{\mu_+}\Big(A_{\ell_i}\,\Big|\,gF_n\setminus\big[R^g_n\cup\bigcup_{j=i}^m A_{\ell_j}\big]\Big) &\geq
			\Psi_\mu\Big(A_{\ell_i}\,\Big|\,gF_n\setminus\big[R^g_n\cup\bigcup_{j=i}^m A_{\ell_j}\big]\Big) + \varepsilon \;.
	\end{align}
	It follows that, uniformly in $g$,
	\begin{align}
		\Psi_{\mu^+}(gF_n) &\geq \Psi_\mu(gF_n) + \varepsilon m + \smallo(\abs{F_n}) \\
			&=
				\Psi_\mu(F_n) + \varepsilon\,\density(D)\abs{F_n} + \smallo(\abs{F_n})
	\end{align}
	as claimed.	
\end{proof}

\section{Equilibrium measures relative to a topological factor}
\label{sec:top_factor}

In the setting of topological factor maps between subshifts,
we have the following extension of the result of Allahbakhshi and Quas~\cite[Thm.~3.3]{AllQua13}
as a corollary of Theorem~\ref{thm:DLR:relative}\ref{thm:DLR:LR:relative}.
\begin{theorem}[Gibbs property for equilibrium measures relative to a topological factor]
\label{thm:top_factor} 
	Let $X$ and~$Y$ be $\GG$-subshifts,
	$\eta$ a topological factor map from $X$ onto $Y$,
	$\nu$  a $\GG$-invariant measure on $Y$, and
	$\Phi$ an absolutely summable interaction on $X$.
	Assume that $X$ satisfies the TMP.
	Then, every invariant measure $\mu$ projecting to $\nu$
	that maximizes the pressure for $f_\Phi$ within the fiber $\eta^{-1}(\nu)$
	satisfies the following Gibbs property:
	for every $A\Subset\GG$ and $u\in L_A(X)$, and $\mu$-almost every $x\in X$, we have
	\begin{align}
		\label{eq:top_factor:gibbs}
		\MoveEqLeft
		\mu\big([u]\,\big|\,\xi^{A^\complement}\lor\eta^{-1}(\field{F}_Y)\big)(x) \\
		&=	\nonumber
			\begin{dcases}
				\frac{1}{Z^{\eta}_{A|A^\complement}(x)}\ee^{-E_{A|A^\complement}(x_{A^\complement}\lor u)}
					& \text{if $x_{A^\complement}\lor u\in X$ and $\eta(x_{A^\complement}\lor u)=\eta(x)$,} \\
				0
					& \text{otherwise,}
			\end{dcases}
	\end{align}
	where $\field{F}_Y$ is the $\sigma$-algebra on $Y$ and $Z^{\eta}_{A|A^\complement}(x)$ is the appropriate normalizing constant.
\end{theorem}

To see how topological factor maps fit in the setting of relative systems,
let $X$ be a $\GG$-subshift, $Y$ a compact metric space with a continuous
$\GG$-action and $\eta\colon X \to Y$ a topological factor map, that is,
a $\GG$-equivariant continuous surjection from $X$ onto $Y$.
Regarding $\Theta \isdef Y$ as an environment space and setting $X_y \isdef \eta^{-1}(y)$,
we obtain a relative system $\Omega \isdef \{(\eta(x),x):  x \in X\}$,
which is nothing other than the graph of $\eta$.

Let $\nu$ be a $\GG$-invariant measure on $Y$.
Via the natural topological conjugacy
$X \to \Omega$, $x \mapsto (\eta(x),x)$, there is a one-to-one correspondence
between $\GG$-invariant measures $\mu$  on $\Omega$ that
project to $\nu$ and $\GG$-invariant measures on $X$ that project to $\nu$.
Let $\Phi$ be an absolutely summable interaction on $X$, and note that
$\Phi$ can be considered, via the same conjugacy, as an absolutely summable
relative interaction on $\Omega$.
(Note however that the class of absolutely summable relative interactions on $\Omega$
is larger than those obtained in this fashion.)

With the above correspondence, the invariant measures $\mu$ (on $X$) that maximize pressure for $f_\Phi$
among all invariant measures projecting to $\nu$
are identified with the equilibrium measures (on $\Omega$) for $f_\Phi$ relative to $\nu$.
Indeed, the pressure of $\mu$ can be written as
\begin{align}
	h_\mu(X) - \mu(f_\Phi) &= h_\nu(Y) + h_\mu(\Omega\,|\,Y) - \mu(f_\Phi) \;.
\end{align}
Since $h_\nu(Y)$ is independent of $\mu$, maximizing the pressure $h_\mu(X) - \mu(f_\Phi)$
is equivalent to maximizing the relative pressure $h_\mu(\Omega\,|\,Y) - \mu(f_\Phi)$.
Likewise, relative Gibbs measures on $\Omega$ for $\Phi$ with marginal $\nu$ correspond
precisely to measures on $X$ that project to $\nu$ and satisfy~\eqref{eq:top_factor:gibbs}.

\begin{proof}[Proof of Theorem~\ref{thm:top_factor}]
	Following the above discussion, it is sufficient to show that
	the TMP on~$X$ implies the TMP on~$\Omega$ relative to $\nu$.
	The result then follows from Theorem~\ref{thm:DLR:relative}\ref{thm:DLR:LR:relative}.

	Every continuous shift-commuting map between subshifts can be expressed
	as a sliding factor map (see e.g.~\cite{LinMar95}).
	Hence, denoting the alphabet of $Y$ by $\Gamma$, there exists a set $F \Subset \GG$ and
	a map %
	$M\colon L_F(X) \rightarrow \Gamma$ such that
	\useshortskip
	\begin{align}
		\eta(x)_g = M\big((g^{-1}x)_{F}\big)
	\end{align}
	for every $x\in X$ and $g\in\GG$.

	Let $A\Subset\GG$ and let $B\supseteq A$ be a memory set for $A$
	witnessing the TMP of $X$.
	Note that every $\tilde{B}\Subset\GG$ with $\tilde{B}\supseteq B$ is also a memory set for $A$.
	We claim that if we choose $\tilde{B}$ large enough such that $AF\cap\tilde{B}^\complement F=\varnothing$
	(in particular, if we set $\tilde{B}\isdef B\cup AFF^{-1}$),
	then $\tilde{B}$ is also a memory set for $A$ witnessing the TMP of $\Omega$ relative to $Y$.
	Indeed, let $y\in Y$ and $x,x'\in X$ be such that $\eta(x)=\eta(x')=y$
	and $x_{\tilde{B}\setminus A}=x'_{\tilde{B}\setminus A}$.
	By the TMP of $X$, the configuration $w\isdef x_{\tilde{B}}\lor x'_{A^\complement}$
	is in $X$.  On the other hand, it is easy to see that
	if the condition $AF\cap\tilde{B}^\complement F=\varnothing$ is satisfied,
	then we also have $\eta(w)=y$.
\end{proof}

\section{Equilibrium measures on group shifts}
\label{sec:groupshifts}

As stated in the Introduction, not all group shifts are SFTs~\cite{Sal18}.
In fact, a group shift may not even satisfy the strong TMP. For instance, if
$\GG \isdef \bigoplus_{n \in \NN} \ZZ/2\ZZ$ is the direct sum of countably many
copies of $\ZZ/2\ZZ$ and $\HH \isdef \ZZ /2\ZZ$,
then the group shift $\XX \isdef \{\symb{0}^{\GG}, \symb{1}^{\GG}\}$
does not satisfy the strong TMP.
Indeed, note that the subgroup $\langle F\rangle\subseteq\GG$
generated by each finite subset $F\Subset\GG$ is finite.
Suppose that $F$ is such that $AF$ is a memory set for $A$.
Choosing $A\isdef\langle F\rangle$ yields that $AF=A$
is a memory set for $A$, which is absurd.

However, all group shifts satisfy TMP as long as they are defined on a
countable group $\GG$. The proof is a straightforward adaptation of
Lemma 2.2 in~\cite{KitSch88}.
	
\begin{proposition}[Group shifts have TMP]
\label{prop:groupshifts_are_wtmp}
	Let $\GG$ be a countable group and $\HH$ a finite group.
	Then every group shift $\XX \subseteq \HH^{\GG}$ satisfies
	the TMP.%
\end{proposition}

\begin{proof}
	For disjoint $A,B\Subset\GG$ and $x\in \XX$, let us define $L_{A|B}(x)$
	as the set of all patterns $p\in L_A(\XX)$ such that $p\lor x_B$ is in $L_{A\cup B}(\XX)$.
	Denote by $L_{A|B}$ the set $L_{A|B}((1_{\HH})^{\GG})$.
	
	Observe that $L_{A|B}$ is a subgroup of $L_A(\XX)$.
	Let us verify that for every $x\in \XX$, the set $L_{A|B}(x)$ is a left coset of $L_{A|B}$.
	Clearly, $x_A \cdot L_{A|B} \subseteq L_{A|B}(x)$.
	Conversely, let $u \in L_{A|B}(x)$.
	Let $z \in \XX$ be such that $z_{A\cup B} = u \lor x_B$.
	Then $(x^{-1}\cdot z)_B = 1_{\HH}^B$ and hence
	$(x^{-1} \cdot z)_A = x_A^{-1} \cdot u$ is in $L_{A|B}$.
	It follows that $u\in x_A \cdot L_{A|B}$.
	Therefore, $L_{A|B}(x) = x_A \cdot L_{A|B}$.
		
	Now let $g_0,g_1,g_2,\dots$ be an enumeration of $\GG$ and
	$B_n \isdef \{g_0,\dots, g_n\}\setminus A$.  Clearly,	
	\begin{align}
		L_A(\XX) \supseteq L_{A|B_0} \supseteq L_{A|B_1} \supseteq L_{A|B_2} \supseteq \cdots \;.
	\end{align}
	As $L_A(\XX)$ is finite, this chain eventually stabilizes,
	and thus there exists an $N \in \NN$
	such that $L_{A|B_{N+m}}=L_{A|B_N}$ for all $m\geq 0$.
	It follows that	$L_{A|B_{N+m}}(x)=L_{A|B_N}(x)$
	for every $x\in \XX$ and all $m\geq 0$.
	But this is equivalent to saying that $C\isdef A\cup B_N$ is a memory set
	for $A$.
	We conclude that $\XX$ satisfies the TMP.
\end{proof}

Theorem~\ref{thm:groupshifts} follows immediately from Proposition~\ref{prop:groupshifts_are_wtmp}
and the extended version of the non-relative Lanford--Ruelle theorem
(Theorem~\ref{thm:DLR:relative}\ref{thm:DLR:LR:relative} on the system $\Omega \isdef \Theta \times \XX$
in which $\Theta\isdef\{\theta\}$ is singleton and $\nu\isdef\delta_\theta$).

We now give an algebraic interpretation of Theorem~\ref{thm:groupshifts}
in the case $\Phi\equiv 0$ and as a corollary, find a sufficient condition
for the uniqueness of the measure of maximal entropy on group shifts.

More generally, let $\XX$ be a compact metric group on which a countable group $\GG$
acts by continuous automorphisms.
A point $z\in\XX$ is said to be \emph{homoclinic} (or \emph{asymptotic})
if for every open neighborhood $U\ni 1_{\XX}$, there is a finite set $F\Subset\GG$
such that $gz\in U$ for all $g\in\GG\setminus F$.
The homoclinic points of $\XX$ form a subgroup of $\XX$
denoted by $\Delta(\XX)$.
The homoclinic points in a group shift are precisely the \emph{finitary} configurations,
that is, the configurations in which all but at most finitely many of the sites
have the identity symbol.

Let us call a probability measure $\mu$ on $\XX$
an \emph{almost Haar} measure if it is invariant under the action of
the homoclinic subgroup of $\XX$ by left-translations, that is, if
$\mu(z^{-1}U)=\mu(U)$ for every measurable $U\subseteq\XX$ and each $z\in\Delta(\XX)$.
Clearly, the Haar measure is almost Haar,
but in general, there can be many other almost Haar measures.
For instance, when $\HH\isdef\ZZ/2\ZZ$ and $\GG$ is an arbitrary countable group,
every probability measure on the group shift $\XX\isdef \{\symb{0}^{\GG}, \symb{1}^{\GG}\}$
is almost Haar, simply because $\XX$ has no homoclinic point other than
its identity element $\symb{0}^{\GG}$.

The almost Haar measures on a group shift are precisely the Gibbs measures
for the trivial interaction $\Phi\equiv 0$.
\begin{proposition}[almost Haar $\equiv$ uniform Gibbs]
	Let $\GG$ be a countable group and $\HH$ a finite group,
	and let $\XX\subseteq\HH^\GG$ be a group shift.
	A probability measure $\mu$ on $\XX$ is almost Haar if and only if it is
	Gibbs for the interaction $\Phi\equiv 0$.
\end{proposition}
\begin{proof}
	First, suppose that $\mu$ is a Gibbs measure for $\Phi\equiv 0$.
	Let $z$ be a homoclinic point.
	Let $A\Subset\GG$ be the support of $z$, that is, $A \isdef z^{-1}(\HH\diagdown \{1_{\HH}\})$, and set $w\isdef z_A$.
	Let $u\in L_A(\XX)$ and $Q\in \xi^{A^\complement}$.
	By the (uniform) Gibbs property of $\mu$, we have
	\begin{align}
		\mu\big([u]\,|\,\xi^{A^\complement}\big) &=
			\mu\big([w^{-1}u]\,|\,\xi^{A^\complement}\big)
	\end{align}
	$\mu$-almost surely.  Integrating over $Q$ gives
	\begin{align}
		\mu([u]\cap Q) &= \mu([w^{-1}u]\cap Q) = \mu\big(z^{-1}([u]\cap Q)\big)
	\end{align}
	which implies that $\mu$ is invariant under left-translation by $z$.
	Since $z$ was arbitrary, we find that $\mu$ is almost Haar.
	
	Conversely, suppose that $\mu$ is almost Haar.
	Let $A\Subset\GG$ be a finite set and $u,v\in L_A(\XX)$.
	If there is no configuration $x\in \XX$ for which both $x_{A^\complement}\lor u$
	and $x_{A^\complement}\lor v$ are in $\XX$, there is nothing to show.
	So, suppose that there exists a configuration $\hat{x}\in \XX$ such that
	$\hat{x}_{A^\complement}\lor u,\hat{x}_{A^\complement}\lor v\in \XX$.
	Since $\XX$ is a group shift, $z\isdef(\hat{x}_{A^\complement}\lor u)(\hat{x}_{A^\complement}\lor v)^{-1}$
	is in $\XX$.  Note that $z$ is a homoclinic point with support $A$ and $w\isdef z_A=uv^{-1}$.
	By the almost Haar property, for every $Q\in\xi^{A^\complement}$ we have
	\begin{align}
		\mu([u]\cap Q) &= \mu\big(z^{-1}([u]\cap Q)\big)
			= \mu([w^{-1}u]\cap Q) = \mu([v]\cap Q) \;.
	\end{align}
	This implies, by the definition of conditional probability,
	that $\mu([u]\,|\,\xi^{A^\complement})=\mu([v]\,|\,\xi^{A^\complement})$ $\mu$-almost surely.
	We conclude that $\mu$ is a Gibbs measure for $\Phi\equiv 0$.
\end{proof}

As a corollary, we have the following restatement of the special case of
Theorem~\ref{thm:groupshifts} with $\Phi\equiv 0$.
\begin{corollary}[Maximal entropy $\Longrightarrow$ almost Haar]
\label{cor:groupshifts:almostHaar}
	Let $\GG$ be a countable amenable group and $\HH$ a finite group,
	and let $\XX \subseteq \HH^{\GG}$ be a group shift.
	Then every measure of maximal entropy on $\XX$
	(with respect to the action of~$\GG$) is almost Haar.
\end{corollary}

Observe that when $\Delta(\XX)$ is dense,
the Haar measure is the unique almost Haar measure on $\XX$.
Therefore, we find the following corollary.
See~\cite[Thm.~8.6]{ChuLi15} %
and~\cite{LinSch99,BoySch08} for closely related results.
\begin{corollary}[Uniqueness of measure of maximal entropy]
\label{cor:groupshifts:intrinsic-ergodicity}
	Let $\GG$ be a countable amenable group and $\HH$ a finite group.
	Let $\XX\subseteq\HH^{\GG}$ be a group shift and suppose that
	its homoclinic subgroup $\Delta(\XX)$ is dense in $\XX$.
	Then, the Haar measure on $\XX$ is the unique measure of maximal entropy on $\XX$
	(with respect to the action of $\GG$).
\end{corollary}

\section{Relative equilibrium measures on lattice slices}
\label{sec:slices}

Recall from the Introduction that
a two-dimensional subshift $Y\subseteq\Sigma^{\ZZ^2}$ can be viewed as
a one-dimensional relative system $\Omega_N$ (for $N$ an arbitrary positive integer)
in which the environment space~$\Theta_N$
consists of the configurations on the complement of the horizontal strip $\ZZ\times[0,N-1]$
that are admissible in $Y$,
and for each $\theta\in\Theta$,
the set $X_\theta$ consists of all configurations $x$ of the strip $\ZZ\times[0,N-1]$
that are consistent with $\theta$ in that $\theta\lor x\in Y$.
Note that $\ZZ$ acts on $\Omega_N$ by horizontal shifts.

In this section, we shall prove Theorem~\ref{thm:Z2_slices},
which states that under suitable conditions on~$Y$,
the equilibrium measures on $Y$ are precisely the $\ZZ^2$-invariant
measures that are relative equilibrium on $\Omega_N$ for each $N$.
In fact, we prove this in a more general setting in which $\ZZ^2$
is replaced with an arbitrary countable amenable group $\GG$,
the horizontal strip is replaced with a union of a finite number of cosets
of a fixed subgroup $\HH\subseteq\GG$ (called a \emph{slice} of $\GG$),
and the horizontal $\ZZ$-action is replaced with the action of~$\HH$.

Before introducing the general setting, let us give a few examples
to show why the above-mentioned equivalence cannot hold without
some assumption on~$Y$.

\newcommand{\cuadritoA}{
	\begin{tikzpicture}[scale = 0.2,baseline=0.1pt]
	\draw (0,0) rectangle  +(1,1);
	\end{tikzpicture}
}
\newcommand{\cuadritoB}{
	\begin{tikzpicture}[scale = 0.2,baseline=0.1pt]
	\draw [fill = black!20] (0,0) rectangle  +(1,1);
	\draw [fill = black] (0.5,0.5) circle (0.3);
	\end{tikzpicture}
}
\newcommand{\cuadritoC}{
	\begin{tikzpicture}[scale = 0.2,baseline=0.1pt]
	\draw [fill = black!20] (0,0) rectangle  +(1,1);
	\draw [fill = white] (0.2,0.5) -- (0.5,0.8) -- (0.8,0.5) -- (0.5,0.2) -- cycle;
	\end{tikzpicture}
}

\begin{example}[Equilibrium but not relative equilibrium I]
	\label{ex:slices_notTMP_yesRelDmix}
	Let $Y$ be the $\ZZ^2$-subshift over the alphabet $\Sigma \isdef \{\cuadritoA, \cuadritoB, \cuadritoC \}$
	consisting of all configurations
	in which the two symbols $\cuadritoB$ and $\cuadritoC$ appear in at most one horizontal row
	(see Fig.~\ref{fig:slices_examples}), that is,
	\begin{align}
		Y &\isdef
			\big\{
				y\in \{\cuadritoA, \cuadritoB, \cuadritoC \}^{\ZZ^2}:
				\text{$y_{(u,v)},y_{(u',v')}\in \{\cuadritoB, \cuadritoC \}$
				implies $v'=v$}
			\big\} \;.
	\end{align} 
	The only non-wandering point in $Y$ is the uniform configuration $\cuadritoA^{\ZZ^2}$.
	Thus the atomic measure $\mu$ supported at $\cuadritoA^{\ZZ^2}$ is the only $\ZZ^2$-invariant measure
	on $Y$.  In particular, $\mu$ is the unique measure of maximal entropy on $Y$.
	However, given its marginal on~$\Theta_1$,
	$\mu$ does not maximize relative entropy on~$\Omega_1$.
	Namely, consider the measure $\mu'$ under which each site outside the strip $\ZZ\times\{0\}$
	has almost surely the symbol~$\cuadritoA$, while the symbols inside the strip $\ZZ\times\{0\}$
	are chosen independently uniformly at random from $\Sigma$.
	Note that $\mu'$ is invariant under horizontal shift
	and has the same marginal as $\mu$ on $\Theta_1$.
	On the other hand, $h_{\mu}(\Omega_1\,|\,\Theta_1)=0$
	while $h_{\mu'}(\Omega_1\,|\,\Theta_1)=\log 3$.
	Let us observe that $Y$ does not satisfy TMP, but the relative system $\Omega_1$
	is relatively D-mixing (even more, it has the relative independence property).	 
	\hfill\exampleqed
\end{example}

\begin{example}[Equilibrium but not relative equilibrium II]
\label{ex:slices_yesTMP_notRelDmix}
	Let us consider a variant of the subshift from the previous example
	in which there is an additional constraint that
	the symbol~$\cuadritoA$ cannot occur in the same row as
	the symbols~$\cuadritoB$ and $\cuadritoC$ (see Fig.~\ref{fig:slices_examples}).
	Namely, let
	\begin{align}
		Y &\isdef
			\big\{
				y\in \{\cuadritoA, \cuadritoB, \cuadritoC \}^{\ZZ^2}:
				\text{$y_{(u,v)},y_{(u',v')}\in \{\cuadritoB, \cuadritoC \}$
				and $y_{(u'',v'')}=\cuadritoA$ imply $v''\neq v'=v$}
			\big\} \;.
	\end{align}
	As in the previous example, the atomic measure $\mu$ supported at~$\cuadritoA^{\ZZ^2}$
	is the unique measure of maximal entropy on~$Y$,
	but given its marginal on $\Theta_1$,
	$\mu$ does not maximize relative entropy on $\Omega_1$.
	In this case, the maximum relative entropy is achieved by the measure
	under which the sites outside $\ZZ\times\{0\}$ are almost surely given the symbol~$\cuadritoA$
	and the sites in $\ZZ\times\{0\}$ are given random symbols
	chosen independently and uniformly from $\{\cuadritoB, \cuadritoC \}$.
	In contrast to the previous example,
	in this example, $Y$ does have TMP (even strong TMP)
	but $\Omega_1$ is not relatively D-mixing.
	\hfill\exampleqed
\end{example}

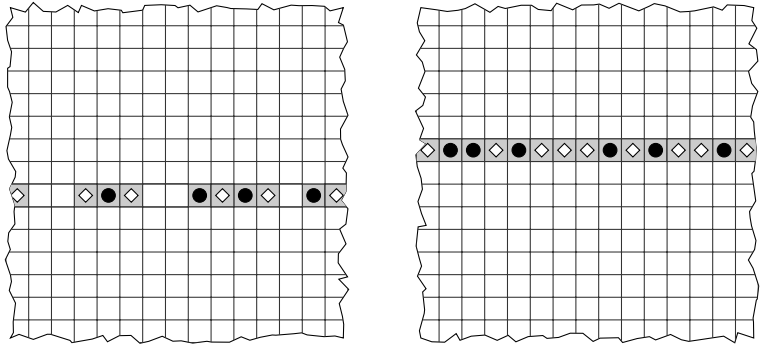
\begin{figure}[h!]
	\centering
	\begin{tikzpicture}[scale = 0.3]

	\begin{scope}[shift = {(0,0)}]
	
	\clip[draw,decorate,decoration={random steps, segment length=4pt, amplitude=2pt}] (0.2,0.2) rectangle (14.8,14.8); 
	\draw [black!80] (0,0) grid (15,15);
	\foreach \j in {0,...,14}{
		\draw [black!80, fill = black!20] (\j,6) rectangle +(1,1);
		\draw [fill = black] (\j +0.5,6.5) circle (0.3); 
	}
	\foreach \j in {0,3,5,6,7,9,11,12,14}{
		\draw [black!80, fill = black!20] (\j,6) rectangle +(1,1);
		\draw [fill = white] (\j+0.2,6.5) -- (\j+0.5,6.8) -- (\j+0.8,6.5) -- (\j+0.5,6.2) -- cycle;
	}
	\foreach \j in {1,2,6,7,12}{
		\draw [black!80, fill = white] (\j,6) rectangle +(1,1);
	}
	
	\end{scope}
	
	\begin{scope}[shift = {(18,0)}]
	
	\clip[draw,decorate,decoration={random steps, segment length=4pt, amplitude=2pt}] (0.2,0.2) rectangle (14.8,14.8); 
	\draw [black!80] (0,0) grid (15,15);
	\foreach \j in {0,...,14}{
		\draw [black!80, fill = black!20] (\j,8) rectangle +(1,1);
		\draw [fill = black] (\j +0.5,8.5) circle (0.3); 
	}
	\foreach \j in {0,3,5,6,7,9,11,12,14}{
		\draw [black!80, fill = black!20] (\j,8) rectangle +(1,1);
		\draw [fill = white] (\j+0.2,8.5) -- (\j+0.5,8.8) -- (\j+0.8,8.5) -- (\j+0.5,8.2) -- cycle;
	}
	\end{scope}
	
	\end{tikzpicture}
	\caption{%
		Configurations from the subshifts in Example~\ref{ex:slices_notTMP_yesRelDmix} (left)
		and Example~\ref{ex:slices_yesTMP_notRelDmix} (right).
	}
	\label{fig:slices_examples}
\end{figure}

\begin{example}[Relative equilibrium but not equilibrium]
\label{ex:slices_yesTMP_notDmix}
	Let $Y \isdef \{\symb{0}\}^{\ZZ^2} \cup \{\symb{1},\symb{2}\}^{\ZZ^2}$.
	Consider the atomic measure	$\mu$ supported at $\symb{0}^{\ZZ^2}$.
	For every positive $N$, $\mu$ maximizes relative entropy
	given its marginal on $\Theta_N$.
	Nevertheless, $\mu$ is not a measure of maximal entropy for~$Y$.
	Note that $Y$ is an SFT, thus has TMP.  It however does not satisfy D-mixing.
	\hfill\exampleqed
\end{example}

Let us now introduce the general setting.
Let $\GG$ be a countable amenable group and $\HH$ a subgroup of $\GG$.
A union of finitely many right cosets of $\HH$ in $\GG$ is called
an \emph{$\HH$-slice} of $\GG$.
Symbolic configurations on an $\HH$-slice can be viewed in a natural way
as configurations with a larger alphabet on~$\HH$.
Namely, given an $\HH$-slice $S$, we choose a collection
$F \isdef \{a_1, \ldots, a_k\}$ of
representatives of distinct right cosets of $\HH$ participating in $S$,
so that $S=\HH F$.
With some abuse of notation,
we will identify the configurations $x\in \Sigma^S$ on the slice $S$
with the configurations $\hat{x}\in (\Sigma^F)^\HH$ on $\HH$
via the natural bijection between $\Sigma^{\HH F}$ and $(\Sigma^F)^\HH$
given by $(\hat{x}_h)_f= x_{hf}$ for $h\in\HH$ and $f\in F$.
Likewise, for $A\Subset\HH$ we identify $\Sigma^{AF}$ with $(\Sigma^F)^A$.
For $d\in\HH F$, we define $p_S(d)$ as the unique element $h\in\HH$
such that $d\in hF$.

Let $Y$ be a subshift on $\GG$.
Each $\HH$-slice of $\GG$ defines a relative system on which $\HH$ acts.
Namely, let $S\isdef\HH F$ be an $\HH$-slice.
For $B\subseteq\GG$, let $\Pi_B$ denote the projection $y\mapsto y_B$.
We introduce a relative system $\Omega_S$
by considering $\Theta_S\isdef \Pi_{S^\complement}(Y)$ as the environment space
and defining
$X_\theta \isdef \{x\in(\Sigma^F)^\HH: \theta\lor x\in Y\}$
as the set of configurations consistent with $\theta\in\Theta_S$.
Note that $\HH$ acts on $\Omega_S$ by translations,
and that this action is topologically conjugate to the action of $\HH$ on $Y$.

An interaction $\Phi$ on $Y$ induces a relative interaction $\widehat{\Phi}$
on $\Omega_S$.
Namely, for every finite subset $A \subseteq \HH$ and every
$u \in  (\Sigma^F)^A$ and $\theta \in \Theta_S$, let
\begin{align}
	\widehat{\Phi}_A(\theta, u) &\isdef
		\sum_{\substack{B \Subset \GG\\ \qquad\mathclap{p_S(B \cap S)= A}\qquad}}
			\Phi _{B}(\theta_{B\setminus S}\lor u_{B \cap S}) \;.
\end{align}
Note that  $\widehat{\Phi}$ is absolutely summable if $\Phi$ is
(in fact, $\norm{\smash{\widehat{\Phi}}} \le \abs{F} \norm{\Phi}$).

\begin{proposition}[Gibbs kernels vs.\ relative Gibbs kernels]
	\label{prop:slices-interactions}
	Let $S\isdef\HH F$ be an $\HH$-slice of $\GG$.
	Let $\Phi$ be an absolutely summable interaction on $Y$
	and $\widehat{\Phi}$ the corresponding relative interaction on $\Omega_S$.
	Let $K$ be the Gibbs specification on $Y$ associated to $\Phi$
	and $\widehat{K}$ the relative Gibbs specification on $\Omega_S$ associated to $\widehat{\Phi}$.
	Let $y\in Y$, and set $\theta\isdef y_{S^\complement}$ and $x\isdef y_S$.
	Then, for every $A\Subset\HH$ and $u\in(\Sigma^F)^A$, we have
	\begin{align}
		\widehat{K}_A\big((\theta,x),[u]\big) &=
			K_{AF}(y,[u]) \;.
	\end{align}
\end{proposition}

\begin{proof}
	Let $E$ denote the Hamiltonian on $Y$ associated to $\Phi$,
	and let $\widehat{E}$ denote the relative Hamiltonian on $\Omega_S$
	associated to $\widehat{\Phi}$.
	Clearly, $x_{\HH \setminus A} \lor u\in X_\theta$
	if and only if $y_{\GG\setminus AF}\lor u\in Y$.
	If either of the latter conditions is satisfied, we have
	\begin{align}
		\widehat{E}_{A|A^{\complement}}(\theta,x_{\HH \setminus A} \lor u)
		& = \sum_{\substack{C\Subset \HH \\ C \cap A \neq\varnothing}}
			\widehat{\Phi}_C(\theta,x_{\HH \setminus A} \lor u) \\
		& = \sum_{\substack{C\Subset \HH \\ C \cap A \ne \varnothing}}
			\sum_{\substack{B \Subset \GG \\ p_S(B \cap S) = C }}
			\Phi_{B}(\theta_{B\setminus S}\lor (x_{\HH \setminus A} \lor u )_{B \cap S})
			\\
		& = E_{(AF)|(\GG \setminus AF)}(y_{\GG \setminus AF } \lor u) \;.
	\end{align}
	The result then follows from the definitions of the Gibbs kernels.
\end{proof}

Before stating the main result, let us verify that TMP on a subshift
implies relative TMP with respect to slices.

\begin{proposition}[TMP $\Longrightarrow$ relative TMP]
	\label{prop:slices-TMP}
	Let $S\isdef\HH F$ be an $\HH$-slice of $\GG$.
	If $Y$ satisfies TMP, then $\Omega_S$ satisfies relative TMP.
\end{proposition}

\begin{proof}
	Let $A \Subset \HH$ and let $B \supseteq AF$ be a memory set for $AF$
	witnessing the TMP of $Y$.
	Since any finite superset of a memory set is also a memory set, we may
	assume that $B \cap S = CF$ for some $C \Subset \HH$.
	We claim that $C$ is a memory set for $A$
	in the relative system $\Omega_S$.

	Let $\theta \in \Theta_S$ and  $x,x' \in X_\theta$ be such that
	$x_{C\setminus A} = x'_{C\setminus A}$.  Let $y,y' \in Y$
	be such that $y_{\GG\setminus S} = y'_{\GG\setminus S} = \theta$,
	$y_{S} = x$ and $y'_{S} = x'$.
	Since $B$ is a memory set for $AF$ in $Y$,
	there is a configuration $\tilde{y} \in Y$ that agrees with
	$y$ on $B$, and thus on $CF = B \cap S$ and with $y'$ on $\GG \setminus AF$.
	In other words, $\tilde{y}_{\GG\setminus S} = \theta$ and $\tilde{y}$ agrees with $y$ on $CF$
	and with $y'$ on $S \setminus AF$.
	In particular, if we set $w \isdef \tilde{y}_{S}$,
	then $w\in X_\theta$ and
	$w$ agrees with $x$ on $C$ and with $x'$ on $\HH\setminus A$.
	This means that $C$ is a memory set for $A$ in $\Omega_S$.
\end{proof}

Now we can state the main general result of this section.

\begin{theorem} %
\label{thm:DLR:slices}
	Let $Y$ be a subshift on a countable amenable group $\GG$.
	Let $\Phi$ an absolutely summable interaction on $Y$
	and $\mu$ a $\GG$-invariant probability measure on $Y$.
	Let $\HH$ be a subgroup of $\GG$.
	\begin{enumerate}[label={\rm (\alph*)},ref={(\alph*)}]
		\item \label{thm:DLR:slices:LR}
			\textup{(Lanford--Ruelle theorem for slices)} \\
			Assume that $Y$ satisfies TMP.
			Assume further that $\mu$ is an equilibrium measure on $Y$ for~$\Phi$.
			Let $S\isdef\HH F$ be an $\HH$-slice of $\GG$,
			and denote by $\widehat{\Phi}$ the relative interaction on $\Omega_S$
			corresponding to $\Phi$.
			Then, $\mu$ is an equilibrium measure on $\Omega_S$ for $\widehat{\Phi}$
			relative to~$\Pi_{S^\complement}\mu$,
			provided that $\Omega_S$ is D-mixing relative to $\Pi_{S^\complement}\mu$.
		\item \label{thm:DLR:slices:D}
			\textup{(Dobrushin theorem for slices)} \\
			Assume that for every $\HH$-slice $S\isdef\HH F$ of $\GG$,
			$\Omega_S$ satisfies relative TMP.
			Assume further that for every $\HH$-slice $S$,
			$\mu$ is an equilibrium measure on $\Omega_S$ for $\widehat{\Phi}$
			relative to $\Pi_{S^\complement}\mu$, where $\widehat{\Phi}$ denotes the relative interaction
			on $\Omega_S$ corresponding to $\Phi$.
			Then, $\mu$ is an equilibrium measure on $Y$ for $\Phi$,
			provided that $Y$ is D-mixing.
	\end{enumerate}
\end{theorem}

\begin{proof}
	Let $K$ denote the Gibbs specification on $Y$ for $\Phi$.
	\begin{enumerate}[label={\rm (\alph*)}]
		\item[\ref{thm:DLR:slices:LR}]
		Let $\widehat{K}$ denote the relative Gibbs specification on $\Omega_S$ for $\widehat{\Phi}$. 
		Since $Y$ satisfies the TMP,
		$\mu$ is a Gibbs measure for $\Phi$ by the (non-relative) Lanford--Ruelle theorem
		(Theorem~\ref{thm:DLR:relative}\ref{thm:DLR:LR:relative} with trivial environment).
		By Proposition \ref{prop:slices-interactions},
		for every $A\Subset\HH$ and $u\in(\Sigma^F)^A$,
		and $\mu$-almost every $(\theta,x)\in\Omega_S$,
		\begin{align}
			\mu\big([u]\,\big|\,\field{F}_\Theta \lor (\xi^F)^{\HH \setminus A}\big)(\theta,x)
				& = \mu\big([u] \,\big|\, \xi^{\GG \setminus AF})(\theta\lor x) \\
				& =  K_{AF}(\theta\lor x, [u] ) \\
				& = \widehat{K}_A((\theta, x), [u])
		\end{align}
		and so $\mu$ is a relative Gibbs measure on $\Omega_S$.
		Now, assuming that $\Omega_S$ is D-mixing relative to $\Pi_{S^\complement}\mu$, 
		by the relative Dobrushin theorem (Theorem~\ref{thm:DLR:relative}\ref{thm:DLR:D:relative}),
		$\mu $ is a relative equilibrium measure on $\Omega_S$
		for $\widehat{\Phi}$ relative to $\Pi_{S^\complement}\mu$.

		\item[\ref{thm:DLR:slices:D}]
		Let $S\isdef\HH F$ be an arbitrary $\HH$-slice in $\GG$.
		Let $\widehat{K}$ denote the relative Gibbs specification on $\Omega_S$ for $\widehat{\Phi}$. 
		Since $\mu$ is a relative equilibrium measure for $\widehat{\Phi}$,
		we can apply the relative Lanford--Ruelle theorem
		(Theorem~\ref{thm:DLR:relative}\ref{thm:DLR:LR:relative})
		to get that $\mu$ is relative Gibbs for $\widehat{\Phi}$.
		Using Proposition \ref{prop:slices-interactions},
		it follows that for
		every $A \Subset \HH$ and $u \in \Sigma^{AF}$
		and $\mu$-almost every $y\in Y$, %
		we have
		\begin{align}
			\mu\big([u] \,\big|\, \xi^{\GG \setminus AF}\big)(y)
				& = \mu\big([u] \,\big|\, \field{F}_\Theta \lor (\xi^F)^{\HH \setminus A})(y_{S^\complement},y_S) \\
				& = \widehat{K}_A((y_{S^\complement}, y_S), [u] ) \\
				& =  K_{AF}(y, [u]).
		\end{align}
		Thus, $\mu$ satisfies the Gibbs condition for sets of the form
		$AF$, with $A \subseteq \HH$. Since the collection of
		sets of the form $AF$, for all such $A$ and $F$, forms a cofinal
		subset of the collection of finite subsets of $\GG$, $\mu$ is a Gibbs
		measure for $\Phi$ (see Remark~1.24 in \cite{Geo88}). Since $\mu$ is
		$\GG$-invariant and Gibbs, it is an equilibrium measure
		by the (non-relative) Dobrushin theorem
		(Theorem~\ref{thm:DLR:relative}\ref{thm:DLR:D:relative} with trivial environment).
		\qedhere
	\end{enumerate}
\end{proof}

Note that in part~\ref{thm:DLR:slices:D} of the above
theorem, we can use Proposition~\ref{prop:slices-TMP} to replace
the condition of relative TMP for every slice with the condition that
$Y$ satisfies TMP.

\begin{remark}[Recovering Dobrushin--Lanford--Ruelle theorem]
	When $\HH$ is the trivial subgroup, the statement of Theorem~\ref{thm:DLR:slices}
	recovers the statement of the Dobrushin--Lanford--Ruelle theorem
	(Theorem~\ref{thm:DLR:relative} with trivial environment).
	Indeed, in this case, $\HH$-slices of $\GG$ are precisely the finite subsets of $\GG$
	and thus the conditions of relative TMP and relative D-mixing
	become trivial.
	Note that according to Corollary~\ref{cor:gibbs:local-optimality},
	$\mu$ is a Gibbs measure if and only if
	it is a relative equilibrium measure for every $\HH$-slice of $\GG$. %
	\hfill\exampleqed
\end{remark}

\begin{remark}[The missing counter-example]
	Examples~\ref{ex:slices_notTMP_yesRelDmix} and~\ref{ex:slices_yesTMP_notRelDmix}
	show that neither of the two conditions	in Theorem~\ref{thm:DLR:slices}\ref{thm:DLR:slices:LR},
	namely TMP and relative D-mixing,
	can be dropped.
	On the other hand, Example~\ref{ex:slices_yesTMP_notDmix}
	shows that Theorem~\ref{thm:DLR:slices}\ref{thm:DLR:slices:D}
	would not hold if we dropped the D-mixing condition.
	This begs the question of whether the remaining condition of TMP can be dropped
	in Theorem~\ref{thm:DLR:slices}\ref{thm:DLR:slices:D}.
	
	However, a counter-example in which Theorem~\ref{thm:DLR:slices}\ref{thm:DLR:slices:D}
	fails in absence of TMP would be more complicated to construct.
	In fact, as the following argument suggests, such an example may require
	$Y$ to satisfy D-mixing but not the UFP (see Sec.~\S\ref{subsec:cond_D}),
	at least when $\GG\isdef\ZZ^2$ and $\HH\isdef\ZZ$.
	We do not know if such a subshift exists.
	
	Consider the basic case of horizontal strips on two-dimensional subshifts,
	thus $\GG\isdef\ZZ^2$ and $\HH\isdef\ZZ$.
	Suppose that $Y\subseteq\Sigma^{\ZZ^2}$ has the UFP with respect to
	the sequence of boxes $F_n\isdef[-n,n]^2$.
	Let us sketch an argument showing that if a $\ZZ^2$-invariant measure $\mu$ on $Y$ has
	maximal relative entropy (with respect to horizontal shift) on every horizontal strip,
	it also maximizes entropy on $Y$ (with respect to two-dimensional shift). 

	Indeed, let $\mu'$ be any other $\ZZ^2$-invariant measure on $Y$
	and suppose that the $\ZZ^2$-entropy of $\mu'$ is larger than
	the $\ZZ^2$-entropy of $\mu$. %
	Then there exists $\varepsilon>0$ such that
	$H_{\mu'}(\xi^{F_n}) \geq H_\mu(\xi^{F_n}) + \varepsilon\abs{F_n}$
	for all sufficiently large $n$.
	By the UFP, there exists a non-negative integer $r$ such that
	for every $y,y'\in Y$ and $n\in\NN$, there exists a configuration $\tilde{y}\in Y$
	that agrees with $y$ on $F_n$ and with $y'$ outside $F_{n+r}$.
	Now, consider the strip $S\isdef\ZZ\times[-n-r,n+r]$ and the sequence $\ldots,B_{-1},B_0,B_1,\ldots$ of translates of $F_n$ contained in $\ZZ\times[-n,n]$ in such a way that each $B_k$ is at distance $r+1$ from $B_{k-1}$ and $B_{k+1}$.
	Let us draw a random configuration $\pmb{y}$ from $\Sigma^{\ZZ^2}$ by choosing
	$\pmb{y}_{B_k}$ (for $k\in\ZZ$) according to $\mu'$,
	and $\pmb{y}_{\ZZ^2\setminus S}$ according to $\mu$, all independently of one another.
	By the UFP, the remaining symbols can be chosen in such a way
	that $\pmb{y}$ is (almost surely) in $Y$.  Let $\tilde{\mu}_0$ be the distribution of $\pmb{y}$.
	This is not necessarily horizontally invariant, so let $\tilde{\mu}$ be a horizontally invariant
	measure obtained from $\tilde{\mu}_0$ by the standard averaging procedure.
	One can now verify that when $n$ is large enough,
	the relative entropy of $\tilde{\mu}$ on $S$ given its complement
	is larger than that of $\mu$, contradicting the assumption.
	\hfill\exampleqed
\end{remark}

In concrete examples,
the conditions of Theorem~\ref{thm:DLR:slices} (TMP, D-mixing and relative D-mixing)
can be cumbersome to verify.
Clearly, these conditions are satisfied if $Y$ is a full shift.
A more relaxed condition covering important examples such as the hard-core model
is the notion of TSSM introduced in Section~\S\ref{sec:hypoth}.
The following corollary (which contains Theorem~\ref{thm:Z2_slices} as a special case)
is a handy version of Theorem~\ref{thm:DLR:slices} in which generality
is traded for simplicity.

\begin{corollary}[Dobrushin--Lanford--Ruelle theorem for slices: handy version] %
	Let $Y$ be a subshift on a countable amenable group $\GG$,
	and assume that $Y$ satisfies TSSM.
	Let $\Phi$ be an absolutely summable interaction on $Y$.
	Let $\HH$ be a subgroup of $\GG$.
	Let $\mu$ be a $\GG$-invariant probability measure on $Y$.
	Then $\mu$ is an equilibrium measure for $\Phi$
	if and only if for every $\HH$-slice $S$ of $\GG$,
	$\mu$ is an equilibrium measure on $\Omega_S$ for $\widehat{\Phi}$ relative to $\Pi_{S^\complement}\mu$,
	where $\widehat{\Phi}$ denotes the relative interaction corresponding to $\Phi$ on~$\Omega_S$.
\end{corollary}

\begin{proof}
	By Theorem~\ref{thm:DLR:slices},
	it suffices to show that $Y$ satisfies TMP and is D-mixing, and that
	for every $\HH$-slice $S\isdef\HH F$,
	the relative system $\Omega_S$ is relatively D-mixing.
	From Proposition~\ref{prop:TSSM_is_SFT_SI}, we know that
	$Y$ is an SI SFT, in particular,
	it satisfies TMP and is D-mixing.
	Thus, it remains to show that $\Omega_S$ is relatively D-mixing.
	We shall in fact show that $\Omega_S$ is relatively SI.

	Indeed, let $R\Subset\GG$ be a finite set that certifies the TSSM property of~$Y$.
	Fix $\theta\in\Theta$ and let $x,y\in X_\theta$.
	Let $A,B\Subset\HH$ be such that $(AF)R\cap (BF)R=\varnothing$.
	Let $g_0,g_1,\ldots$ be an enumeration of the elements of $\GG\setminus S$ and
	set $M_n\isdef\{g_0,g_1,\ldots,g_n\}$ for $n\in\NN$.
	Note that $\theta_{M_n}\lor x_A$ and $\theta_{M_n}\lor y_B$ are in $L(Y)$.
	Therefore, by TSSM, there is a configuration $z^{(n)}\in Y$ such that
	$z^{(n)}_{M_n}=\theta_{M_n}$, $z^{(n)}_{AF}=x_A$ and $z^{(n)}_{BF}=y_B$.
	Let $z$ be an accumulation point of $z^{(n)}$ as $n\to\infty$.
	Since $Y$ is closed, $z\in Y$.
	On the other hand, $z_{\GG\setminus S}=\theta$, $z_{AF}=x_A$ and $z_{BF}=y_B$. Note that if we define $D \isdef (FR)(FR)^{-1} \cap \HH$ then whenever $A(FR) \cap B(FR) \neq \varnothing$ we have that $AD \cap BD \neq \varnothing$. This shows that $X_{\theta}$ is strongly irreducible with the finite set $D$ as a witness. Since $D$ does not depend upon $\theta$, we find that $\Omega_S$
	is relatively SI.
\end{proof}

\section{Relative version of Meyerovitch's theorem}
\label{sec:relative_tom}

Before proving Theorem~\ref{thm:Tom:relative} we need to introduce two
technical tools. One is the concept of non-overlapping patterns and the
second one is a subshift which separates shapes. Let $A\subseteq\GG$
be a finite set. We say that two patterns $u,v\in\Sigma^A$ are
\emph{non-overlapping} in $\Omega$
if
\begin{align}
	g_1([u]\cup [v])\cap g_2([u]\cup [v]) &= \varnothing
\end{align}
whenever $g_1,g_2\in\GG$ are two distinct elements with $g_1A\cap g_2 A\neq\varnothing$.
The \emph{hard-core} shift with \emph{shape}~$A$ is defined as
\useshortskip
\begin{align}
	Y &\isdef
		\left\{
			y\in\{\symb{0},\symb{1}\}^\GG:
			\text{%
			$y_k=y_\ell=\symb{1}$ implies either $k=\ell$ or $kA\cap \ell A=\varnothing$}
		\right\}\;.
\end{align}
If we think of symbol $\symb{1}$ as a particle with shape $A$,
then $Y$ consists of all configurations of particles whose volumes do not overlap.

We will proceed through the proof in two steps.
First, we treat the simpler case
in which $u$ and $v$ are non-overlapping in $\Omega$.
We encode the relative system $\Omega$ into another relative system $\widehat{\Omega}$
in which the symbolic part contains only the information about the occurrences of
$u$ and $v$ wherever they are interchangeable.
This new system will have the relative TMP, %
even more, it will have the relatively independence property,
and thus the relative Lanford--Ruelle theorem will yield the result.
In the second step, we treat the general case where $u$ and $v$	might overlap.
This time we use an auxiliary subshift $Y$ (namely, the hard-core shift with shape $A$)
and construct a new relative system $\widetilde{\Omega}\isdef\Omega\times Y$
in which the symbolic part has an extra layer $y\in Y$ chosen independently of
$x$ and~$\theta$.
The auxiliary subshift $Y$ consists of configurations of particles on $\GG$
that are sufficiently far apart.
Associated to $u$ and $v$, there are two non-overlapping patterns $\tilde{u}$ and $\tilde{v}$,
which are simply $u$ and $v$ with a particle on top.
Since $\tilde{u}$ and $\tilde{v}$ are non-overlapping, the result of the first step will hold.
The general result for $u$ and $v$ will then follow from the independence of the auxiliary layer.

\begin{proof}[Proof of Theorem~\ref{thm:Tom:relative}]
	Let $u,v\in\Sigma^A$ be non-overlapping in $\Omega$.
	Without loss of generality, we shall assume that $A \ni 1_{\GG}$; if not, we reduce to this case
	by shifting $x,\theta,u$ and $v$ appropriately.
	Let $Z\isdef \{ \hexed{\ },\hexed{u},\hexed{v}\}^{\GG}$ and consider the map
	\useshortskip
	\begin{align}
		\phi \colon \Omega \to (\Theta \times (\Sigma \cup \{\sun,\star\})^{\GG}) \times  Z
	\end{align}
	where $\phi(\theta,x) \isdef ((\theta,\widehat{x}),\widehat{z})$ is defined by leaving $\theta$ unchanged
	and setting
	\begin{align}
		\widehat{x}_k &\isdef
			\begin{cases}
				\sun & \text{if $k^{-1}(\theta,x) \in [u]\cup[v]$ and $k^{-1}\theta \in \Theta_{u,v}$,}\\
				\star &
					\text{if $k\in\ell A\setminus \ell$ for some $\ell\in\GG$ where
					$\ell^{-1}(\theta,x) \in [u]\cup[v]$ and $\ell^{-1}\theta \in \Theta_{u,v}$,}\\
				x_k & \text{otherwise.}
		\end{cases}\\
		\widehat{z}_k &\isdef
			\begin{cases}
				\hexed{u} & \text{if $k^{-1}(\theta,x) \in [u]$ and $k^{-1}\theta \in \Theta_{u,v}$,} \\
				\hexed{v} & \text{if $k^{-1}(\theta,x) \in [v]$ and $k^{-1}\theta \in \Theta_{u,v}$,} \\
				\hexed{\ } & \text{otherwise.}
			\end{cases}
	\end{align}
	In other words, $\widehat{x}$ is obtained from $x$ by erasing the appearances of $u$ and $v$
	wherever they are interchangeable
	(i.e., at positions $k$ such that $u$ and $v$ are interchangeable for $k^{-1}\theta$).
	Each erased pattern is replaced by the symbols $\star$ and $\sun$,
	where $\sun$ indicates the reference point of the occurrence.
	The information regarding the erased occurrences of $u$ and $v$
	is then recorded in $\widehat{z}$.

	The map $\phi$ is clearly ${\GG}$-equivariant, bijective and measurable.
	Furthermore, given $((\theta,\widehat{x}),\widehat{z})=\phi(\theta,x)$,
	one can recover $x$ from $\widehat{x}$ and $\widehat{z}$ alone, by means of a block map.
	More precisely, each symbol $x_k$ can be recovered by looking
	at the restrictions of $\widehat{x}$ and $\widehat{z}$ to $kA^{-1}$
	using the local rule
	\begin{align}
		\Xi \colon (\Sigma \cup \{\sun,\star\})^{A^{-1}} \times \{\hexed{\ },\hexed{u},\hexed{v}\}^{A^{-1}}
			\to \Sigma
	\end{align}
	given by
	\begin{align}
		\Xi(p,q) &\isdef
			\begin{cases}
				u_a & \text{if $p_{a^{-1}}= \sun$ and $q_{a^{-1}}= \hexed{u}$ for some $a\in A$,} \\
				v_a & \text{if $p_{a^{-1}}= \sun$ and $q_{a^{-1}}= \hexed{v}$ for some $a\in A$,} \\
				p_{1_{\GG}} & \text{otherwise.}
			\end{cases}
	\end{align}
	The local rule $\Xi$ is well-defined because $u$ and $v$ are non-overlapping.

	Consider the system $\widehat{\Omega} \isdef \phi(\Omega)$ where the environment
	$\widehat{\Theta}$ is the set of all $(\theta,\widehat{x})$ that appear in the projection
	of $\widehat{\Omega}$ on the first coordinate and
	$\widehat{X}_{(\theta,\widehat{x})}$ is the set of all $\widehat{z} \in Z$ that are consistent
	with $(\theta,\widehat{x})$ in $\widehat{\Omega}$.
	The new system $\widehat{\Omega}$ has the relative TMP %
	--- even more, it has the relative independence property.
	Let $\widehat{\mu} \isdef \phi\mu$, and define $\widehat{\nu}$ as the projection of $\widehat{\mu}$
	onto $\widehat{\Theta}$.
	Define a relative interaction $\widehat{\Phi}$ on $\widehat{\Omega}$ by
	\useshortskip
	\begin{align}
		\widehat{\Phi}_{B}((\theta,\widehat{x}),\widehat{z}) & \isdef
			\sum_{C: C \cdot A^{-1} = B}\Phi_C\big(\phi^{-1}((\theta,\widehat{x}),\widehat{z})\big)\\
		& =
			\sum_{C: C \cdot A^{-1} = B}
			\Phi_C\Big(
				\theta ,\big\{\Xi\big(
					(c^{-1}\widehat{x})_{A^{-1}},(c^{-1}\widehat{z})_{A^{-1}}
				\big)\big\}_{c \in C}
			\Big),
	\end{align}
	and let $\widehat{E}$ denote the corresponding relative Hamiltonian. It is easy to verify that
	$\widehat{\Phi}$ is	absolutely summable,
	and that, for every $\GG$-invariant probability measure $\mu$,
	\begin{align}\label{eqn:hamiltoniano_tom}
		\mu(f_\Phi) &= \widehat{\mu}(f_{\widehat{\Phi}})
	\end{align}
	(see Sec.~\S\ref{sec:hamiltoniano_tom:argument}).

	We claim that $\widehat{\mu}$ is an equilibrium measure for $\widehat{\Phi}$
	relative to $\widehat{\nu}$.
	Indeed, let $\widehat{\underline{\mu}}$ be any other $\GG$-invariant measure
	that projects to $\widehat{\nu}$, and let $\underline{\mu}$ be the induced measure on $\Omega$.
	Since $\mu$ is assumed to be an equilibrium measure for $\Phi$ relative to $\nu$
	and $\underline{\mu}$ projects to $\nu$,
	we have
	\begin{align}
		h_{\underline{\mu}}(\Omega \mid \Theta) - \underline{\mu}(f_\Phi)
		&\leq
			h_{\mu}(\Omega \mid \Theta)  - \mu(f_\Phi) \;.
	\end{align}
	By the chain rule,
	$h_{{\underline\mu}}({\Omega} \mid {\Theta}) =
		h_{\widehat{\underline{\mu}}}(\widehat{\Theta} \mid {\Theta}) +
		h_{\widehat{\underline{\mu}}}(\widehat{\Omega} \mid \widehat{\Theta})$
	and
	$h_{{\mu}}({\Omega} \mid {\Theta}) =
		h_{{\widehat{\mu}}}(\widehat{\Theta} \mid {\Theta}) +
		h_{{\widehat{\mu}}}(\widehat{\Omega} \mid \widehat{\Theta})$.
	As both $\widehat{\underline{\mu}}$ and $\widehat{\mu}$ project to $\widehat{\nu}$, we have
	$h_{\widehat{\underline{\mu}}}(\widehat{\Theta} \mid {\Theta}) =
		h_{\widehat{\mu}}(\widehat{\Theta} \mid {\Theta})$.
	Putting this together with equation~(\ref{eqn:hamiltoniano_tom}) yields
	\begin{align}
		h_{\widehat{\underline{\mu}}}(\widehat{\Omega} \mid \widehat{\Theta})
			- \widehat{\underline{\mu}}(f_{\widehat{\Phi}})
		&\leq
			h_{\widehat{\mu}}(\widehat{\Omega} \mid \widehat{\Theta})
			- \widehat{\mu}(f_{\widehat{\Phi}}) \;,
	\end{align}
	which establishes the claim.
	
	Denote by $[\hexed{u}]$ and $[\hexed{v}]$ the cylinder set
	consisting of all points $((\theta,\widehat{x}),\widehat{z})\in\widehat{\Omega}$
	in which respectively $\hexed{u}$ and $\hexed{v}$ appear at position $1_\GG$ of $\widehat{z}$.
	Recall that $\xi$ denotes the partition of $\Omega$ induced by by
	the projection $(\theta,x)\mapsto x_{1_\GG}$.
	Similarly, we denote by $\widehat{\xi}$ the partition of
	$\widehat{\Omega}$ induced by the projection
	$((\theta,\widehat{x}),\widehat{z})\mapsto\widehat{z}_{1_\GG}$,
	and write $\field{F}_{\widehat{\Theta}}$ for the $\sigma$-algebra on $\widehat{\Omega}$
	generated by $\widehat{\Theta}$.
	Applying Theorem~\ref{thm:DLR:relative}\ref{thm:DLR:LR:relative},
	we know that $\widehat{\mu}$ is a relative Gibbs measure for $\widehat{\Phi}$,
	thus
	for $\widehat{\mu}$-almost every $((\theta,\widehat{x}),\widehat{z})\in [\hexed{u}] \cup [\hexed{v}]$,
	\begin{align}
		\MoveEqLeft
		\widehat{\mu}\big([\hexed{u}] \,\big|\, \widehat{\xi}^{\{1_\GG\}^{\complement}}
			\vee \field{F}_{\widehat{\Theta}}\big)
			((\theta,\widehat{x}),\widehat{z}) \nonumber \\
		\label{eq_tomrellr_a}
		&=
			\frac{
				\ee^{-\widehat{E}_{\{1_\GG\} \mid
					\{1_\GG\}^{\complement}}((\theta,\widehat{x}),\widehat{z}_{\{1_\GG\}^{\complement}}
					\vee \hexed{u})}
			}{
				\ee^{-\widehat{E}_{\{1_\GG\} \mid
					\{1_\GG\}^{\complement}}((\theta,\widehat{x}),\widehat{z}_{\{1_\GG\}^{\complement}}
					\vee \hexed{u})}
				+
				\ee^{-\widehat{E}_{\{1_\GG\} \mid
					\{1_\GG\}^{\complement}}((\theta,\widehat{x}),\widehat{z}_{\{1_\GG\}^{\complement}}
					\vee \hexed{v})}
			} \;, \\
		\MoveEqLeft
		\widehat{\mu}\big([\hexed{v}] \,\big|\, \widehat{\xi}^{\{1_\GG\}^{\complement}}
			\vee \field{F}_{\widehat{\Theta}}\big)
			((\theta,\widehat{x}),\widehat{z}) \nonumber \\
		\label{eq_tomrellr_b}
		&=
			\frac{
				\ee^{-\widehat{E}_{\{1_\GG\} \mid
					\{1_\GG\}^{\complement}}((\theta,\widehat{x}),\widehat{z}_{\{1_\GG\}^{\complement}}
					\vee \hexed{v})}
			}{
				\ee^{-\widehat{E}_{\{1_\GG\} \mid
					\{1_\GG\}^{\complement}}((\theta,\widehat{x}),\widehat{z}_{\{1_\GG\}^{\complement}}
					\vee \hexed{u})}
				+
				\ee^{-\widehat{E}_{\{1_\GG\} \mid
					\{1_\GG\}^{\complement}}((\theta,\widehat{x}),\widehat{z}_{\{1_\GG\}^{\complement}}
					\vee \hexed{v})}
			} \;.
	\end{align}
	Putting equations~\eqref{eq_tomrellr_a} and~\eqref{eq_tomrellr_b} together,
	we obtain
	\begin{align}
	\label{eq_tomrellr_c}
		\frac{
			\widehat{\mu}\big([\hexed{u}] \,\big|\, \widehat{\xi}^{\{1_\GG\}^{\complement}}
			\vee \field{F}_{\widehat{\Theta}}\big)
			((\theta,\widehat{x}),\widehat{z})
		}{
			\ee^{-\widehat{E}_{\{1_\GG\} \mid
				\{1_\GG\}^{\complement}}((\theta,\widehat{x}),\widehat{z}_{\{1_\GG\}^{\complement}}
				\vee \hexed{u})}
		} &=
		\frac{
			\widehat{\mu}\big([\hexed{v}] \,\big|\, \widehat{\xi}^{\{1_\GG\}^{\complement}}
			\vee \field{F}_{\widehat{\Theta}}\big)
			((\theta,\widehat{x}),\widehat{z})
		}{
			\ee^{-\widehat{E}_{\{1_\GG\} \mid
				\{1_\GG\}^{\complement}}((\theta,\widehat{x}),\widehat{z}_{\{1_\GG\}^{\complement}}
				\vee \hexed{v})}
		}
	\end{align}
	for $\widehat{\mu}$-almost every
	$((\theta,\widehat{x}),\widehat{z})\in [\hexed{u}]\cup [\hexed{v}]$.

	On one hand, letting $(\theta,x)=\phi^{-1}((\theta,\widehat{x}),\widehat{z})$, we have
	\begin{align}
		\widehat{E}_{\{1_\GG\} \mid \{1_\GG\}^{\complement}}
			((\theta,\widehat{x}),\widehat{z}_{\{1_\GG\}^{\complement}} \vee \hexed{u})
		& =
			\sum_{B \ni 1_\GG} \widehat{\Phi}_B((\theta,\widehat{x}),\widehat{z}_{\{1_\GG\}^{\complement}}
				\vee \hexed{u}) \\
		& =
			\sum_{B \ni 1_\GG} \sum_{C: C\cdot A^{-1} = B}
				\Phi_C(\theta, x_{A^{\complement}} \vee u) \\
		& =
			\sum_{C: C\cdot A^{-1} \ni 1_\GG} \Phi_C(\theta, x_{A^{\complement}} \vee u) \\
		& =
			\sum_{C: C\cap A\neq\varnothing}
				\Phi_C(\theta, x_{A^{\complement}} \vee u) \\
		\label{eq_tomrellr_d_1}
		& =
			E_{A \mid A^{\complement}}(\theta, x_{A^{\complement}} \vee u) \;,
	\end{align}
	and by a similar argument
	\begin{align}
		\label{eq_tomrellr_d}
		\widehat{E}_{\{1_\GG\} \mid \{1_\GG\}^{\complement}}
			((\theta,\widehat{x}),\widehat{z}_{\{1_\GG\}^{\complement}} \vee \hexed{v})
		&=
			E_{A \mid A^{\complement}}(\theta, x_{A^{\complement}} \vee v) \;.
	\end{align}
	On the other hand,
	\begin{align}
		\label{eq_tomrellr_e}
		\widehat{\mu}([\hexed{u}] \mid \widehat{\xi}^{\{1_\GG\}^{\complement}}
			\vee \field{F}_{\widehat{\Theta}})((\theta,\widehat{x}),\widehat{z})
			& = \mu([u] \mid (\xi_{X})^{A^\complement}\vee\field{F}_{\Theta},[u]\cup [v])(\theta,x) \;,\\
		\label{eq_tomrellr_f}
		\widehat{\mu}([\hexed{v}] \mid \widehat{\xi}^{\{1_\GG\}^{\complement}}
			\vee \field{F}_{\widehat{\Theta}})((\theta,\widehat{x}),\widehat{z})
			& = \mu([v] \mid (\xi_{X})^{A^\complement}\vee\field{F}_{\Theta},[u]\cup [v])(\theta,x) \;.
	\end{align}
	Putting together
	equations~\eqref{eq_tomrellr_c}, \eqref{eq_tomrellr_d_1}, \eqref{eq_tomrellr_d},
	\eqref{eq_tomrellr_e} and~\eqref{eq_tomrellr_f},
	we get that for $\mu$-almost every $(\theta,x)\in [u]\cup [v]$ satisfying $\theta \in \Theta_{u,v}$,
	\begin{align}
		\label{eq_tomrellr_g}
		\frac{
			\mu([u] \mid \xi^{A^\complement}\vee\field{F}_{\Theta})(\theta,x)
		}{
			\ee^{-E_{A \mid A^\complement}(\theta, x_{A^{\complement}} \vee u)}
		}
		&=
		\frac{
			\mu([v] \mid \xi^{A^\complement}\vee\field{F}_{\Theta})(\theta,x)
		}{
			\ee^{-E_{A \mid A^\complement}(\theta, x_{A^{\complement}} \vee v)}
		} \;.
	\end{align}
	This concludes the proof in the case where $u$ and $v$ are non-overlapping.

	We now consider the general case. If $u = v$, the result is immediate.
	Otherwise, let $Y$ be the hard-core shift with shape $A$.
	We claim that there must exist a measure of maximal entropy $\pi$ on $Y$
	such that $\pi([\symb{1}])>0$.
	This can be seen in various ways, for instance by verifying that
	$Y$ has positive topological entropy, or by invoking the Lanford--Ruelle theorem.
	For a more direct argument, note that
	if $\pi_0$ is a $\GG$-invariant measure such that $\pi_0([\symb{1}]) = 0$, then clearly
	$h_{\pi_0}(Y) = 0$.
	Hence, it is enough to show that there exists a $\GG$-invariant measure
	giving positive measure to $[\symb{1}]$.
	By Lemma~\ref{lem:delone}, there exists a set $D \subseteq \GG$ which is $A$-separated
	and has positive uniform lower density with respect to
	a F\o{}lner sequence $(F_n)_{n \in \NN}$.
	Let $w\in\{\symb{0},\symb{1}\}^\GG$ be the configuration with $w_k\isdef\symb{1}$
	if and only if $k\in D$,
	and define $\pi_n \isdef \abs{F_n}^{-1}\sum_{g \in F_n}g^{-1}\delta_{w}$.
	Any accumulation point of $(\pi_n)_{n \in \NN}$ is a $\GG$-invariant measure
	$\pi$ that satisfies $\pi([\symb{1}])>0$.
	
	Now consider the system $\widetilde{\Omega}\isdef\Omega \times Y$
	as a relative system with environment $\Theta$
	and $\widetilde{X}_\theta\isdef\{(x,y): \text{$x\in X_\theta$ and $y\in Y$}\}$.
	Endow $\widetilde{\Omega}$ with the measure $\tilde{\mu} \isdef \mu \times \pi$
	and the interaction $\widetilde{\Phi}_C(\theta,(x,y)) \isdef \Phi_C(\theta,x)$.
	By construction, $\tilde{\mu}$ is an equilibrium measure for $\widetilde{\Phi}$ relative to $\nu$.
	Consider now the patterns $\tilde{u},\tilde{v} \in (\Sigma \times \{\symb{0},\symb{1}\})^A$
	defined by
	\begin{align}
		\tilde{u}_a &=
			\begin{cases}
				(u_a,\symb{1}) & \text{if $a = 1_\GG$,} \\
				(u_a,\symb{0}) & \text{otherwise,}\\
			\end{cases}
		& \tilde{v}_a &=
			\begin{cases}
				(v_a,\symb{1}) & \text{if $a = 1_\GG$,} \\
				(v_a,\symb{0}) & \text{otherwise.}\\
			\end{cases}
	\end{align}
	
	By the definition of $Y$ and the fact that $u \neq v$, the patterns $\tilde{u},\tilde{v}$ are
	non-overlapping in $\widetilde{\Omega}$. We can thus apply the result
	for non-overlapping patterns to obtain that for
	$\tilde{\mu}$-almost every $(\theta,(x,y))\in [\tilde{u}] \cup [\tilde{v}]$
	such that $\theta \in \Theta_{\tilde{u},\tilde{v}}$,
	\begin{align}
		\label{eq_relmegatom_a}
		\frac{
			\tilde{\mu}\big([\tilde{u}]\, \big|\,\tilde{\xi}^{A^{\complement}}\vee\field{F}_\Theta\big)
				(\theta,(x,y))
		}{
			\ee^{-\widetilde{E}_{A|A^\complement}(\theta,\tilde{u}\lor (x,y)_{A^\complement})}
		} & =
		\frac{
			\tilde{\mu}\big([\tilde{v}]\, \big|\,\tilde{\xi}^{A^{\complement}}\vee\field{F}_\Theta\big)
				(\theta,(x,y))
		}{
			\ee^{-\widetilde{E}_{A|A^\complement}(\theta,\tilde{v}\lor (x,y)_{A^\complement})}
		} \;,
	\end{align}
	where $\tilde{\xi}$ denotes the partition of $\widetilde{\Omega}$
	induced by $(\theta,(x,y))\mapsto (x_{1_\GG},y_{1_\GG})$
	and $\widetilde{E}$ is the relative Hamiltonian associated to $\widetilde{\Phi}$.
	With some abuse of notation, we write $\field{F}_\Theta$ for the $\sigma$-algebras
	generated by $\Theta$ both in $\Omega$ and in $\widetilde{\Omega}$.
		
	By the definition of $\widetilde{\Phi}$, we have that
	\begin{align}
		\label{eq_relmegatom_b}
		\widetilde{E}_{A|A^\complement}(\theta,(x,y)) &= E_{A|A^\complement}(\theta,x) \;.
	\end{align}	
	Furthermore, as $\tilde{\mu} = \mu \times \pi$, we have
	\begin{align}
		\label{eq_relmegatom_c}
		\tilde{\mu}\big([\tilde{u}]\, \big|\,\tilde{\xi}^{A^{\complement}}\vee
			\field{F}_\Theta\big)(\theta,(x,y))
		& =
			\mu\big([u]\, \big|\,\xi^{A^{\complement}}\vee\field{F}_\Theta\big)(\theta,x)\cdot
			\pi\big([\symb{1}]\, \big|\,\zeta^{A^{\complement}} \big)(y) \;, \\
		\label{eq_relmegatom_c1}
		\tilde{\mu}\big([\tilde{v}]\, \big|\,\tilde{\xi}^{A^{\complement}}\vee
			\field{F}_\Theta\big)(\theta,(x,y))
		& =
			\mu\big([v]\, \big|\,\xi^{A^{\complement}}\vee\field{F}_\Theta\big)(\theta,x)\cdot
			\pi\big([\symb{1}]\, \big|\,\zeta^{A^{\complement}} \big)(y) \;,
	\end{align}
	where $\zeta$ stands for the partition of $Y$ generated by the symbol at the origin.
	
	Substituting \eqref{eq_relmegatom_b}, \eqref{eq_relmegatom_c} and~\eqref{eq_relmegatom_c1}
	in equation~\eqref{eq_relmegatom_a}, we get	that
	for $\mu$-almost every $(\theta,x)\in [u]\cup [v]$ satisfying $\theta \in \Theta_{u,v}$
	and $\pi$-almost every $y \in [\symb{1}]$,
	\begin{align}
		\MoveEqLeft
		\frac{
			\mu\big([u]\, \big|\,\xi^{A^{\complement}}\vee\field{F}_\Theta\big)(\theta,x)\cdot
			\pi\big([\symb{1}]\, \big|\,\zeta^{A^{\complement}} \big)(y)
		}{
			\ee^{-E_{A|A^\complement}(\theta,x_{A^\complement} \vee u)}
		} \nonumber\\
		\label{eq_relmegatom_d}
		& =
		\frac{
			\mu\big([v]\, \big|\,\xi^{A^{\complement}}\vee\field{F}_\Theta\big)(\theta,x)\cdot
			\pi\big([\symb{1}]\, \big|\,\zeta^{A^{\complement}} \big)(y)
		}{
			\ee^{-E_{A|A^\complement}(\theta,x_{A^\complement} \vee v)}
		} \;.
	\end{align}	
	Note that we may replace the condition ``$y \in [\symb{1}]$'' by
	``$y_{\GG \setminus \{1_\GG\}} \vee \symb{1} \in Y$'' and the equality will still hold.
	Also, if we integrate the factor
	$\pi\big([\symb{1}]\, \big|\,\zeta^{A^{\complement}} \big)(y)$ with respect to $\pi$,
	we obtain
	\begin{align}
		\pi([\symb{1}]) &=
			\int_{y: y_{\{1_\GG\}^\complement} \vee \symb{1} \in Y }
				\pi\big([\symb{1}]\, \big|\,\zeta^{A^{\complement}}\big)(y)\,\dd\pi(y)
			+
			\int_{y: y_{\{1_\GG\}^\complement} \vee \symb{1} \notin Y }
				\pi\big([\symb{1}]\, \big|\,\zeta^{A^{\complement}}\big)(y)\,\dd\pi(y) \;,
	\end{align}
	where the second term is $0$.
	Thus, integrating~\eqref{eq_relmegatom_d} with respect to $\pi$, we obtain
	\begin{align}
		\label{eq_relmegatom_e}
		\frac{
			\mu\big([u]\, \big|\,\xi^{A^{\complement}}\vee\field{F}_\Theta\big)(\theta,x)
			\cdot \pi\big([\symb{1}]\big)
		}{
			\ee^{-E_{A|A^\complement}(\theta,x_{A^\complement} \vee u)}
		}
		& =
		\frac{
			\mu\big([v]\, \big|\,\xi^{A^{\complement}}\vee\field{F}_\Theta\big)(\theta,x)
			\cdot \pi\big([\symb{1}]\big)
		}{
			\ee^{-E_{A|A^\complement}(\theta,x_{A^\complement} \vee v)}
		} \;.
	\end{align}
	As $\pi([\symb{1}]) > 0$, it follows that	
	\begin{align}
		\label{eq_relmegatom_f}
		\frac{
			\mu\big([u]\, \big|\,\xi^{A^{\complement}}\vee\field{F}_\Theta\big)(\theta,x)
		}{
			\ee^{-E_{A|A^\complement}(\theta,x_{A^\complement} \vee u)}
		}
		& =
		\frac{
			\mu\big([v]\, \big|\,\xi^{A^{\complement}}\vee\field{F}_\Theta\big)(\theta,x)
		}{
			\ee^{-E_{A|A^\complement}(\theta,x_{A^\complement} \vee v)}
		} \;.
	\end{align}
	for $\mu$-almost every $(\theta,x)\in [u]\cup [v]$ such that $\theta \in \Theta_{u,v}$.
	This concludes the proof of the theorem.
\end{proof}

We have used the relative Lanford--Ruelle theorem to prove
Theorem~\ref{thm:Tom:relative}.  We now show the converse implication,
so that the two theorems are really equivalent under fairly simple reductions.
More specifically, we show that when $\Omega$ has the relative TMP,
the conclusion of Theorem~\ref{thm:Tom:relative}
becomes equivalent to saying that $\mu$ is a relative Gibbs measure for $\Phi$ with marginal $\nu$.

\begin{proof}[Proof of Theorem~\ref{thm:DLR:relative}\ref{thm:DLR:LR:relative} using Theorem~\ref{thm:Tom:relative}]
	Let $A \Subset \GG$ and let $B\supseteq A$ be a memory set for $A$ witnessing
	the TMP of $\Omega$ relative to $\nu$.
	Let $u,v\in\Sigma^A$ be arbitrary patterns.  Then, for every $w\in\Sigma^{B\setminus A}$
	and $\nu$-almost every $\theta\in\Theta$, the patterns $w\lor u$ and $w\lor v$ are
	interchangeable for $\theta$ provided they are both in $L_B(X_\theta)$. %
	From Theorem~\ref{thm:Tom:relative}, it follows that for every $w\in\Sigma^{B\setminus A}$,
	\begin{align}
		\frac{
			\mu\big([w\lor u]\,\big|\,\xi^{B^\complement}\lor\field{F}_\Theta\big)(\theta,x)
		}{
			\ee^{-E_{B|B^\complement}(\theta,w\lor u\lor x_{B^\complement})}
		}
		&=
		\frac{
			\mu\big([w\lor v]\,\big|\,\xi^{B^\complement}\lor\field{F}_\Theta\big)(\theta,x)
		}{
			\ee^{-E_{B|B^\complement}(\theta,w\lor v\lor x_{B^\complement})}
		}
	\end{align}
	for $\mu$-almost every $(\theta,x)\in[w]$ such that $w\lor u,w\lor v\in L_B(X_\theta)$.
	If we apply the chain rule to the numerators above and decompose the exponents in the denominators,
	and then cancel the common factor
	\begin{align}
		\frac{
			\mu\big([w]\,\big|\,\xi^{B^\complement}\lor\field{F}_\Theta\big)(\theta,x)
		}{
			\ee^{-E_{(B\setminus A)|B^\complement}(\theta,w\lor x_{B^\complement\cup A})}
		},
	\end{align}
	then the resulting expression simplifies to
	\begin{align}
		\frac{
			\mu\big([u]\,\big|\,\xi^{A^\complement}\lor\field{F}_\Theta\big)(\theta,x)
		}{
			\ee^{-E_{A|A^\complement}(\theta,u\lor x_{A^\complement})}
		}
		&=
		\frac{
			\mu\big([v]\,\big|\,\xi^{A^\complement}\lor\field{F}_\Theta\big)(\theta,x)
		}{
			\ee^{-E_{A|A^\complement}(\theta,v\lor x_{A^\complement})}
		}
	\end{align}
	for $\mu$-almost every $(\theta,x)\in[w]$ such that
	$x_{A^\complement}\lor u,x_{A^\complement}\lor v\in X_\theta$.
	This is true for every $w\in\Sigma^{B\setminus A}$.
	The latter equality is equivalent to $\mu$ being a relative Gibbs measure.
\end{proof}

Considering the fact that in the proof of Theorem~\ref{thm:Tom:relative} we
only applied the relative Lanford--Ruelle theorem on a relatively independent system,
and that the relative Lanford--Ruelle theorem can be deduced from
Theorem~\ref{thm:Tom:relative} as shown above, we obtain that the following
three statements are essentially equivalent in the relative setting:
\begin{itemize}
	\item
		The relative Lanford--Ruelle theorem for systems which satisfy relative independence.
	\item
		The relative Lanford--Ruelle theorem for systems satisfying the relative TMP. %
	\item
		The relative version of Meyerovitch's theorem.
\end{itemize}

If we restrict exclusively to the non-relative setting, the Lanford--Ruelle
theorem for subshifts with TMP (or even for SFTs)
does not follow from the Lanford--Ruelle theorem for full shifts. Similarly,
Meyerovitch's theorem cannot be deduced from Lanford--Ruelle through
a simple recoding. The addition of an environment in the relative setting
can be used as a tool to fix a measure on a restricted portion of a dynamical
system and give information about measures which project to that portion
and are optimal outside of it. Hence, the three statements become equally
powerful in this setting. We see this as an indication that the relative setting
is the appropriate level of generalization for these results.

\bibliographystyle{plainurl}
\bibliography{bibliography}

\appendix

\section{Appendix}

\subsection{Topology of $\xspace{P}_\nu(\Omega)$}
\label{apx:measures:topology}

Let $\nu$ be a probability measure on $\Theta$.
Every measure $\mu\in\xspace{P}_\nu(\Omega)$ defines
a positive linear functional $J$ on the Banach space $C_\Theta(\Omega)$.
Every such functional is \emph{$\nu$-normalized} meaning that
$J(\indicator{A\times\Sigma^\GG})=\nu(A)$ for every measurable $A\subseteq\Theta$.
When $\Theta$ is a standard Borel space
(i.e., isomorphic, as a measurable space, to a Borel subset of a complete separable metric space),
the converse is also true.

\begin{proposition}[Relative Riesz theorem]
\label{prop:Riesz:relative}
	Assume that $\Theta$ is a standard Borel space.  %
	Then, for every $\nu$-normalized positive linear functional $J$ on $C_\Theta(\Omega)$
	there corresponds a unique measure $\mu\in\xspace{P}_\nu(\Omega)$
	such that $\mu(f)=J(f)$ for all $f\in C_\Theta(\Omega)$.
\end{proposition}
\begin{proof}
	Without loss of generality (by passing through an isomorphism),
	we can assume that $\Theta$ is a compact metric space equipped with its Borel $\sigma$-algebra
	(see e.g. \cite[Thm.~13.1.1]{Dud04}).
	Then the set $C(\Omega)$ of all continuous functions on $\Omega$
	is a Banach subspace of $C_\Theta(\Omega)$.
	By the Riesz representation theorem, the restriction of $J$ to $C(\Omega)$
	identifies a unique probability measure
	$\mu$ on $\Omega$ such that $\mu(f)=J(f)$ for all $f\in C(\Omega)$.
	
	Let $g\colon\Theta\to\RR$ be a bounded measurable function and $[u]$ a cylinder set.
	Consider $f(\theta,x)\isdef g(\theta)\indicator{[u]}(\theta,x)$.
	Then $f$ is relatively local.  Furthermore, every relatively local function on $\Omega$
	is a linear combination of functions of this form.
	Since the relatively local functions are dense in $C_\Theta(\Omega)$
	and both $J$ and $\mu$ are continuous on $C_\Theta(\Omega)$,
	it is enough to verify that $\mu(f)=J(f)$.
	
	Let $\varepsilon>0$.  By Lusin's theorem, there exists a function $g_\varepsilon\in C(\Theta)$
	and a closed set $E\subseteq\Theta$ such that $g=g_\varepsilon$ on $E$ and
	$\nu(\Theta\setminus E)<\varepsilon$ and $\mu\big((\Theta\setminus E)\times\Sigma^\GG\big)<\varepsilon$.
	Furthermore, we can choose $g_\varepsilon$ such that $\norm{g_\varepsilon}\leq\norm{g}$.
	Define
	$f_\varepsilon(\theta,x)\isdef g_\varepsilon(\theta)\indicator{[u]}(\theta,x)$.
	Since $f_\varepsilon\in C(\Omega)$, we have $\mu(f_\varepsilon)=J(f_\varepsilon)$.
	Note that
	\begin{align}
		\abs{\mu(f_\varepsilon)-\mu(f)} &\leq \mu(\abs{f-f_\varepsilon})
			\leq (\norm{f}+\norm{f_\varepsilon})\mu\big((\Theta\setminus E)\times\Sigma^\GG\big)
			<2\norm{f}\varepsilon \;.
	\intertext{Similarly, since $J$ is positive linear, we have}
		\abs{J(f_\varepsilon)-J(f)} &\leq J(\abs{f-f_\varepsilon})
			\leq (\norm{f}+\norm{f_\varepsilon})\nu(\Theta\setminus E)
			< 2\norm{f}\varepsilon \;.
	\end{align}
	Therefore, $\abs{\mu(f)-J(f)}<4\norm{f}\varepsilon$.
	Since $\varepsilon$ is arbitrary, the claim follows.
\end{proof}

A consequence of the above proposition is that when $\Theta$ is a standard Borel space,
the space $\xspace{P}_\nu(\Omega)$ is compact.
Indeed, as a set of linear functionals, $\xspace{P}_\nu(\Omega)$
is a closed subset of the unit ball in the dual space $C^*_\Theta(\Omega)$,
thus the compactness follows from Alaoglu's theorem.
We do not know whether the assumption that $\Theta$ is standard Borel
is necessary for the compactness of $\xspace{P}_\nu(\Omega)$.

\subsection{Omitted arguments}
\subsubsection{Verification of~(\ref{eq:Hamiltonian:residue:Folner})}
\label{sec:Hamiltonian:residue:Folner:argument}

Let $B\Subset\GG$ be a finite set and define
$\partial^-_B F_n\isdef\{g\in F_n: gB\cap F_n^\complement\neq\varnothing\} =
F_n\setminus\bigcap_{b\in B}F_n b^{-1}$.
We have
\begin{align}
	\sum_{\substack{C \Subset\GG \\
		C\cap F_n\neq\varnothing\\ C\cap F_n^\complement\neq\varnothing}}
		\norm{\Phi_C}
	&=
		\sum_{\substack{C \Subset \GG \\
		C\cap\partial^-_B F_n\neq\varnothing\\ C\cap F_n^\complement\neq\varnothing}}
		\norm{\Phi_C} +
		\sum_{\substack{C \Subset \GG \\
		C\cap F_n\neq\varnothing, C\cap F_n^\complement\neq\varnothing\\
		C\cap\partial^-_B F_n=\varnothing}}
		\norm{\Phi_C} \\
	&\leq
		\abs{\partial^-_B F_n}\norm{\Phi} +
		\underbrace{\abs{F_n\setminus\partial^-_B F_n}}_{\leq\abs{F_n}}
			\sum_{\substack{C \Subset \GG \\
				C\ni 1_\GG, C\not\subseteq B}}
				\norm{\Phi_C} \;.
\end{align}
The first term is $\smallo(\abs{F_n})$, whereas the second term
is of the form $c_B\abs{F_n}$ where $c_B\to 0$ as $B\nearrow\GG$
along the finite subsets of $\GG$ directed by inclusion.
\qed

\subsubsection{Verification of~(\ref{eq:Hamiltonian:different-sets:residue})}
\label{sec:Hamiltonian:different-sets:residue:argument}
We have
{
\belowdisplayskip=0pt
\begin{align}
	\norm{E_{B|B^\complement}-E_{A|A^\complement}} &=
		\Biggnorm{
			\sum_{\substack{C\Subset\GG\\ C\cap B\neq\varnothing}}\Phi_C -
			\sum_{\substack{C\Subset\GG\\ C\cap A\neq\varnothing}}\Phi_C
		} \\
	&\leq
		\sum_{\substack{C\Subset\GG\\C\cap B\neq\varnothing\\ C\cap A=\varnothing}}\norm{\Phi_C} \\
	&\leq
		\sum_{\substack{C\Subset\GG\\ C\cap(B\setminus A) \neq\varnothing }}\norm{\Phi_C} \\
	&\leq
		\sum_{c\in(B\setminus A)}\sum_{C\ni c}\norm{\Phi_C} \\
	&=
		\abs{B\setminus A}\norm{\Phi} \;.
\end{align}
}
\qed

\subsubsection{Verification of~(\ref{eq:interaction-observable:equivalence:1})}
\label{sec:interaction-observable:equivalence}
Using the definition of $f_\Phi$, for every finite set $A\Subset\GG$, we have
\begin{align}
	\Bigabs{E_A(\theta,x) - \sum_{g\in A}f_\Phi(g^{-1}\theta,g^{-1}x)} &\leq
		\sum_{\substack{C\Subset\GG \\
			C\cap A\neq\varnothing\\ C\cap A^\complement\neq\varnothing}}
			\frac{\abs{A\cap C}}{\abs{C}}\norm{\Phi_C} \;.
\end{align}
For $A\isdef F_n$,
the estimate~\eqref{eq:interaction-observable:equivalence:1}
follows as in~\eqref{eq:Hamiltonian:residue:Folner}
(see Sec.~\S\ref{sec:Hamiltonian:residue:Folner:argument}).
\qed

\subsubsection{Verification of~(\ref{eq:variational-principle:proof:ineqA})
	and~(\ref{eq:variational-principle:proof:ineqB})}
\label{sec:variational-principle:proof:ineqAnB}
Inequality~\eqref{eq:variational-principle:proof:ineqA} follows by writing
\begin{align}
	Z_{F_n|F_n^\complement}(\theta,x) &=
		\sum_{\substack{u\in\Sigma^{F_n}\\ \mathclap{u\lor x_{F_n^\complement}\in X_\theta}}}
		\ee^{-E_{F_n|F_n^\complement}(\theta,u\lor x_{F_n^\complement})} 
	\leq
		\sum_{u\in L_{F_n}(X_\theta)}
		\ee^{-E_{F_n}(\theta,u)+\smallo(\abs{F_n})}
	=
		Z_{F_n}(\theta)\; \ee^{\smallo(\abs{F_n})} \;.
\end{align}
In order to verify~\eqref{eq:variational-principle:proof:ineqB},
let us use the shorthand $\partial F_n^\theta\isdef F_n^\theta\setminus F_n$.
We can write
\begin{align}
	Z_{F_n^\theta|(F_n^\theta)^\complement}(\theta,x) &=
		\mathop{%
			\sum_{u\in\Sigma^{F_n}}
			\sum_{\partial u\in\Sigma^{\partial F_n^\theta}}
		}_{u\lor\partial u\lor x_{(F_n^\theta)^\complement}\in X_\theta}
		\ee^{
			-E_{F_n}(\theta,u)
			-E_{\partial F_n^\theta|(\partial F_n^\theta)^\complement}(\theta,u\lor\partial u\lor x_{(F_n^\theta)^\complement})
		} \\
	&=
		\sum_{u\in L_{F_n}(X_\theta)}\ee^{-E_{F_n}(\theta,u)}
		\sum_{\substack{\partial u\in\Sigma^{\partial F_n^\theta}\\
			\mathclap{u\lor\partial u\lor x_{(F_n^\theta)^\complement}\in X_\theta}}}
			\ee^{-E_{\partial F_n^\theta|(\partial F_n^\theta)^\complement}(\theta,u\lor\partial u\lor x_{(F_n^\theta)^\complement})} \;.
\end{align}
Now observe that, since $F_n^\theta$ is a mixing set for $F_n$,
the second sum in the latter inequality is non-empty.
It follows that
\useshortskip
{
\belowdisplayskip=0pt
\begin{align}
	Z_{F_n^\theta|(F_n^\theta)^\complement}(\theta,x) &\geq
		\smash{\sum_{u\in L_{F_n}(X_\theta)}}\ee^{-E_{F_n}(\theta,u)}
			\ee^{-\abs{\smash{\partial F_n^\theta}}\norm{\Phi}}
	=
		Z_{F_n}(\theta)\; \ee^{-\abs{\smash{\partial F_n^\theta}}\norm{\Phi}} \;.
\end{align}
}
\qed

\subsubsection{Verification of~(\ref{eq:martingale-convergence:uniform:relative:conditioning})}
\label{sec:martingale-convergence:uniform:relative:conditioning:argument}
The right-hand side is $(\xi^B\lor\field{F}_\Theta)$-measurable
and for every $[u]\in\xi^B$ and $W\in\field{F}_\Theta$ we have
\begin{align}
	\int_{[u]\cap W}
		\frac{
			\mu(\indicator{[x_B]}f\,|\,\field{F}_\Theta)(\theta,x)
		}{
			\mu([x_B]\,|\,\field{F}_\Theta)(\theta,x)
		} \,\dd\mu(\theta,x)
	&=
		\mu\bigg(
			\indicator{[u]}\indicator{W}
			\frac{
				\mu(\indicator{[u]}f\,|\,\field{F}_\Theta)
			}{
				\mu([u]\,|\,\field{F}_\Theta)
			}
		\bigg) \\
	&=
		\mu\bigg(
			\indicator{[u]}
			\frac{
				\mu(\indicator{[u]}\indicator{W}f\,|\,\field{F}_\Theta)
			}{
				\mu([u]\,|\,\field{F}_\Theta)
			}
		\bigg) \\
	&=
		\mu\bigg(
		\mu\bigg(
			\indicator{[u]}
			\frac{
				\mu(\indicator{[u]}\indicator{W}f\,|\,\field{F}_\Theta)
			}{
				\mu([u]\,|\,\field{F}_\Theta)
			}
		\,\bigg|\,
			\field{F}_\Theta
		\bigg)
		\bigg) \\
	&=
		\mu\bigg(
			\frac{
				\mu(\indicator{[u]}\indicator{W}f\,|\,\field{F}_\Theta)
			}{
				\cancel{\mu([u]\,|\,\field{F}_\Theta)}
			}
			\cancel{\mu(\indicator{[u]}\,|\,\field{F}_\Theta)}
		\bigg) \\
	&=
		\int_{[u]\cap W}f\,\dd\mu
\end{align}
If two bounded measurable functions have equal integrals over
each element of a generating semi-algebra, they are almost surely equal.
\qed

\subsubsection{Verification of~(\ref{eqn:hamiltoniano_tom})}
\label{sec:hamiltoniano_tom:argument}
For every $(\theta,x)\in\Omega$, we have
\begin{align}
	f_{\widehat{\Phi}}(\phi(\theta,x)) &=
		\sum_{B\ni 1_\GG} \frac{1}{\abs{B}} \widehat{\Phi}(\phi(\theta,x)) \\
	&=
		\sum_{B\ni 1_\GG} \frac{1}{\abs{B}} \sum_{C: C\cdot A^{-1}=B} \Phi_C(\theta,x) \\
	&=
		\sum_{C: C\cdot A^{-1}\ni 1_\GG} \frac{1}{\abs{C\cdot A^{-1}}} \Phi_C(\theta,x) \;.
\end{align}
Integrating with respect to a measure $\mu$, we get
\begin{align}
\label{eq:hamiltoniano_tom:argument:expression:1}
	\widehat{\mu}(f_{\widehat{\Phi}}) &=
		\sum_{C: C\cdot A^{-1}\ni 1_\GG} \frac{1}{\abs{C\cdot A^{-1}}}\;\mu(\Phi_C)	\;.	
\intertext{Compare this with the expression}
\label{eq:hamiltoniano_tom:argument:expression:2}
	\mu(f_\Phi) &=
		\sum_{C\ni 1_\GG} \frac{1}{\abs{C}}\;\mu(\Phi_C)	\;,
\end{align}
and observe that when $\mu$ is $\GG$-invariant,
the right-hand sides of~\eqref{eq:hamiltoniano_tom:argument:expression:1}
and~\eqref{eq:hamiltoniano_tom:argument:expression:2} coincide.
\qed

\end{document}